\documentclass[reqno,10pt,dvips,a4paper]{amsart}

\usepackage{amssymb,mathptmx,cite,psfrag,eucal,array,enumitem,setspace,geometry,color}
\usepackage[dvips]{graphicx}
\usepackage{tikz}

\DeclareSymbolFont{largesymbols}{OMX}{zplm}{m}{n} 

\numberwithin{equation}{section}

\newcolumntype{C}{>{$}c<{$}} 
\allowdisplaybreaks

\newcommand{\tl}[1]{\mathsf{TL}_{#1}}

\newcommand{\func}[2]{#1 \left( #2 \right)}

\newcommand{\brac}[1]{\left( #1 \right)}

\newcommand{\set}[1]{\left\{ #1 \right\}}
\newcommand{\st}{\mspace{5mu} : \mspace{5mu}}

\newcommand{\abs}[1]{\left| #1 \right|}

\newcommand{\inner}[2]{\bigl\langle #1 , #2 \bigr\rangle}

\newcommand{\flr}[1]{\left\lfloor #1 \right\rfloor}

\newcommand{\ZZ}{\mathbb{Z}}
\newcommand{\NN}{\mathbb{N}}

\newcommand{\CC}{\mathbb{C}}

\newcommand{\ii}{\mathfrak{i}}

\newcommand{\wun}{\mathbf{1}}

\newcommand{\wall}[2]{\mspace{3mu} \rule[-4pt]{1pt}{13pt} \mspace{3mu} #1 \quad #2 \mspace{3mu} \rule[-4pt]{1pt}{13pt} \mspace{3mu}}

\newcommand{\NatMod}{\mathcal{V}}
\newcommand{\IrrMod}{\mathcal{L}}
\newcommand{\RadMod}{\mathcal{R}}
\newcommand{\ProjMod}{\mathcal{P}}

\newcommand{\LinkMod}{\mathcal{M}}

\newcommand{\Res}[1]{#1{}\hspace{-0.2em} \downarrow {}}
\newcommand{\Ind}[1]{#1{}\hspace{-0.2em} \uparrow {}}

\newcommand{\qnum}[1]{\left[ #1 \right]_q}
\newcommand{\qnumber}[2]{\left[ #1 \right]_{#2}}

\newcommand{\ses}[3]{0 \rightarrow #1 \rightarrow #2 \rightarrow #3 \rightarrow 0}

\newcommand{\dses}[5]{0 \longrightarrow #1 \overset{#2}{\longrightarrow} #3 \overset{#4}{\longrightarrow} #5 \longrightarrow 0}

\newcommand{\rdses}[5]{#1 \overset{#2}{\longrightarrow} #3 \overset{#4}{\longrightarrow} #5 \longrightarrow 0}

\newcommand{\eqnref}[1]{Equation~\eqref{#1}}
\newcommand{\eqnDref}[2]{Equations~\eqref{#1} and \eqref{#2}}

\newcommand{\secref}[1]{Section~\ref{#1}}
\newcommand{\secDref}[2]{Sections~\ref{#1} and \ref{#2}}

\newcommand{\appref}[1]{Appendix~\ref{#1}}
\newcommand{\figref}[1]{Figure~\ref{#1}}

\newcommand{\thmref}[1]{Theorem~\ref{#1}}
\newcommand{\thmDref}[2]{Theorems~\ref{#1} and \ref{#2}}
\newcommand{\lemref}[1]{Lemma~\ref{#1}}
\newcommand{\propref}[1]{Proposition~\ref{#1}}

\newcommand{\corref}[1]{Corollary~\ref{#1}}

\DeclareMathOperator{\id}{id}
\DeclareMathOperator{\im}{im}

\DeclareMathOperator{\Hom}{Hom}
\DeclareMathOperator{\coker}{coker}

\theoremstyle{plain}
\newtheorem{theorem}{Theorem}[section]
\newtheorem{lemma}[theorem]{Lemma}
\newtheorem{proposition}[theorem]{Proposition}
\newtheorem{corollary}[theorem]{Corollary}

\definecolor{forestgreen}{rgb}{0.13,0.54,0.13}

\begin{document}

\title[Standard Modules, Induction and the Temperley-Lieb Algebra]{Standard Modules, Induction \\ and the Structure of the Temperley-Lieb Algebra}

\author[D Ridout]{David Ridout}

\address[David Ridout]{
Department of Theoretical Physics \\
Research School of Physics and Engineering;
and
Mathematical Sciences Institute;
Australian National University \\
Canberra, ACT 2600 \\
Australia
}

\email{david.ridout@anu.edu.au}

\author[Y Saint-Aubin]{Yvan Saint-Aubin}

\address[Yvan Saint-Aubin]{
D\'{e}partement de Math\'{e}matiques et de Statistique \\
Universit\'{e} de Montr\'{e}al \\
Qu\'{e}bec, Canada, H3C 3J7.
}

\email{saint@dms.umontreal.ca}

\date{\today}


\begin{abstract}
The basic properties of the Temperley-Lieb algebra $\tl{n}$ with parameter $\beta=q+q^{-1}$, $q\in\mathbb C \setminus \set{0}$, are reviewed in a pedagogical way. The link and standard (cell) modules that appear in numerous physical applications are defined and a natural bilinear form on the standard modules is used to characterise their maximal submodules. When this bilinear form has a non-trivial radical, some of the standard modules are reducible and $\tl{n}$ is non-semisimple. This happens only when $q$ is a root of unity. Use of restriction and induction allows for a finer description of the structure of the standard modules. Finally, a particular central element $F_n\in\tl{n}$ is studied; its action is shown to be non-diagonalisable on certain indecomposable modules and this leads to a proof that the radicals of the standard modules are irreducible. Moreover, the space of homomorphisms between standard modules is completely determined. The principal indecomposable modules are then computed concretely in terms of standard modules and their inductions. Examples are provided throughout and the delicate case $\beta=0$, that plays an important role in physical models, is studied systematically.
\end{abstract}

\maketitle

\onehalfspacing

%
%

\section{Introduction} \label{sec:Intro}

The Temperley-Lieb algebras are key objects both in mathematics and physics. Temperley and Lieb \cite{TemRel71} introduced them as complex associative algebras that arose in their study of transfer matrix approaches to (planar) lattice models.  This family of algebras, indexed by a positive integer $n$ and a complex number $\beta$, spread quickly through the physics community where it underlies the study of Potts models \cite{Martin}, ice-type models \cite{BaxExa82}, and the Andrews-Baxter-Forrester models \cite{AndEig84}.  That Temperley-Lieb algebras play a fundamental role in our modern understanding of phase transitions cannot be overstated.  These algebras were subsequently rediscovered by Jones \cite{JonInd83} who used them to define what is now known as the Jones polynomial in knot theory.  The Temperley-Lieb algebras are also intimately connected with the representation theory of the symmetric groups through their realisation as natural quotients of the (type A) Hecke algebras (see \cite{Mathas} for example).

As is usual in physical applications, it is the representation theory of the Temperley-Lieb algebra $\tl{n}$ which is the main focus of attention.  Indeed, Temperley and Lieb's original contribution takes place in a $2^n$-dimensional representation commonly used by physicists for the study of spin chains.  Such representations still form an active direction of research in mathematical physics.  There are, in addition, somewhat smaller representations that are perhaps more natural to consider including, in particular, those which we shall refer to in what follows as \emph{link representations}.  From these, one obtains quotients that have come to be described as being \emph{standard}.  It is well-known that these standard representations are irreducible for almost all values of the parameter $\beta$ and that for such $\beta$, the (finite-dimensional) representations of $\tl{n}$ are completely reducible.  However, it is a curious fact that the $\beta$ which are thereby excluded consist, to a large degree, of the parameter values that are of most interest to physicists.

This failure of complete reducibility is also well-known and there have been significant efforts by physicists and mathematicians to understand the representation theory for these exceptional values of $\beta$.  There is a continuing interest in this quest due to the current reinvigoration of the study of logarithmic conformal field theories.  To explain, the lattice models studied by physicists are believed to have limits (as $n \rightarrow \infty$ say) in which the model should be replaced by a field theory possessing conformal invariance.  The representations upon which the lattice model is founded are supposed to be replaced by representations of the quantised algebra of infinitesimal conformal transformations, the Virasoro algebra (see \cite{GaiCon11} for a recent attempt to formalise this).  Moreover, the corresponding conformal field theory is now believed not to be a minimal model, as the original paradigm suggested, but rather some logarithmic version thereof \cite{PeaLog06,ReaAss07,RidPer07}.  Here, ``logarithmic'' indicates that the Virasoro algebra representations involved also fail to be completely reducible.  Conjecturing results about these logarithmic theories through the study of their lattice counterparts remains a popular approach.

Perhaps the first to seriously address the structure of Temperley-Lieb representations for all $\beta \in \CC$ was Martin. In \cite{Martin}, he devotes two chapters to obtaining an explicit construction of the \emph{principal indecomposable} representations of the Temperley-Lieb algebra.  The arguments are rather involved and rely heavily upon intricate combinatorics and a detailed study of a collection of primitive idempotents introduced by Wenzl \cite{Wenzl88}.  Shortly thereafter, and independently, Goodman and Wenzl \cite{GoodWenzl93} applied a similarly detailed study of these idempotents to prove explicit results concerning the structure of the \emph{blocks} (two-sided ideals) of the Temperley-Lieb algebras.  The proofs require a long series of elementary but technical lemmata. Unfortunately, they exclude the case when $\beta = 0$ which is of considerable physical interest. Nevertheless, their methods lead to explicit and highly non-trivial descriptions for the radicals of the blocks.

We should also mention the well-known contribution of Westbury \cite{WesRep95}.  His article approaches the representation theory of the Temperley-Lieb algebras from a more algebraic, and less combinatorial, perspective.  First, a sufficient condition for complete reducibility is given whose form is very familiar to conformal field theorists.  This criterion is the non-degeneracy of an invariant bilinear form acting on the standard representations and Westbury computes the determinant of this form explicitly (the analogous result for the Virasoro algebra is, of course, the formula of Kac for the determinant of the form defined on Verma modules).  The technique employed involves recursion relations and was suggested by older work of James and Murphy \cite{JamDet79}. (The recursion Westbury employs contains a mistake, as is illustrated by the example he gives.)  The rest of Westbury's article addresses what happens when complete reducibility fails.  The method involves using induction and restriction to determine abstractly the spaces of homomorphisms between standard representations (although there is mention of certain explicit constructions along the lines of Martin).  Proofs are often minimal, incomplete or, in one case, referred to Martin's book.  Crucial points are therefore left without complete arguments: The exactness of the sequences satisfied by the restriction and induction of standard modules, the existence of non-trivial homomorphisms between certain standard modules when $q$ is a root of unity, and the construction and completeness of the set of principal indecomposable modules.

While the above makes our dissatisfaction with Westbury's article evident, his algebraic approach to the Temperley-Lieb algebra provides, in our opinion, an excellent road-map to learning this particular corner of representation theory.  Indeed, our motivation for writing this article derives in large part from a desire to have a pedagogical and (mostly) self-contained summary of Temperley-Lieb theory, including detailed proofs.  We believe that this will be of significant value to mathematical physicists working on lattice models and conformal field theories, as well as provide novice representation-theorists with an excellent worked example over $\mathbb C$ which explicitly illustrates some of the difficulties of non-semisimple associative algebras.

Of course, there are many possible approaches to Temperley-Lieb theory and we certainly do not claim that ours is in any way the best.  Indeed, we have tried to minimise the level of sophistication required wherever possible while still introducing the most basic tools that are indispensable for non-semisimple representation theory.  To illustrate our prejudices in this regard (and to borrow from John Cardy \cite{CarSLE05}), we mention that the word `quiver' has just made its only appearance.  Other elegant alternatives include deducing the results from their Hecke algebra analogues (see \cite{Mathas} for example) or using Schur-Weyl duality and quantum group representation theory (see \cite{GaiLat12} for a recent sketch in this direction).  We also mention a powerful category-theoretic approach due to Graham and Lehrer \cite{GL-Aff} in which structural results for $\tl{n}$ follow as special cases of their investigations into the \emph{affine} Temperley-Lieb algebras.

The article begins in \secref{sec:Diagrams} with a review of the diagrammatic and algebraic definitions of the Temperley-Lieb algebras, proving their equivalence using a standard combinatorial argument.  As far as we are aware, the diagrammatic approach to Temperley-Lieb algebras is due to Kauffman \cite{kauffman}, though the equivalence to the algebraic approach is only sketched there using an example.\footnote{We thank Fred Goodman for correspondence on this point.} Here, we mostly follow the seminal work of Jones \cite{JonInd83}.  We note, in particular, his result concerning a canonical form for monomials constructed from the standard Temperley-Lieb generators.  This turns out to be crucial for the analysis of inducing standard representations in \secref{sec:Ind}.  We remark that we make a ``heretical'' choice for Temperley-Lieb diagrams, orienting them at ninety degrees to that customarily found in the literature (see \cite{DiFMea97} for a precedent illustrating this heresy).  This choice facilitates the translation between the diagrammatic and algebraic depictions of multiplication: A diagram drawn on the left of another object corresponds to a left-action on this latter object.  We also find it notationally convenient for the following section.

\secref{sec:Reps} then introduces the link representations and their quotients, the standard representations.  Various basic, but essential, results are proven for the latter, roughly following the beautiful work of Graham and Lehrer on cellular algebras \cite{GraCel96} (the standard representations are cell representations in their formalism).  Particular attention is paid to the problematic case $\beta = 0$ in which there is a single standard representation upon which the usual invariant bilinear form vanishes identically.

The hard work begins in \secref{sec:Gram} in which we compute the determinant of an invariant bilinear form on the standard representations.  Here, we follow Westbury \cite{WesRep95} in using module restriction to block-diagonalise the Gram matrices and deduce recursion relations for the diagonalising matrices.  (One can also use idempotents to derive such recursion relations; see \cite{GL-Aff} for example.)  As mentioned above, the recursion relation that Westbury gives is incorrect and we discuss in detail the appropriate refinements that have to be made.  The determinant formula then leads to the well-known result concerning the generic semisimplicity of the Temperley-Lieb algebras.  \secref{sec:Explore} then uses the proof of the determinant formula to compute the dimension of the kernel of the Gram matrix and thence the dimensions of the radicals and irreducible quotients of the standard representations.  By utilising pictures known as Bratteli diagrams, we motivate the conjecture that these radicals are themselves irreducible (when non-trivial).

\secDref{sec:Ind}{sec:R=L} are devoted to proving this conjecture.  The first details what is obtained upon performing the induced module construction on a standard representation of $\tl{n}$ to obtain a representation of $\tl{n+1}$.  This result is stated in Westbury, but without proof.  We then use these induced representations, in the second, to demonstrate the existence of certain homomorphisms between standard representations.  This relies crucially upon a computation which shows that a particular element, which we call $F_n$, of the centre of $\tl{n}$, acts in a non-diagonalisable fashion upon appropriately chosen induced modules.  The existence of these homomorphisms is enough to prove the conjecture that the radicals of the standard modules are irreducible (or trivial).  To our knowledge, this irreducibility was first proven in \cite{GL-Aff}; our elementary proof appears to be quite different.

The analysis to this point leaves several questions unanswered, one of which is whether we have determined a complete set of mutually non-isomorphic irreducible representations.  Answering questions like these require a little more sophistication and so we turn, in \secref{sec:Proj}, to the consideration of the principal indecomposable representations.  For this we employ, once again, the induced module technique developed in \secref{sec:Ind} to concretely determine the structure of these representations in terms of that of the standard representations.  We remark that the proofs rely heavily on an extremely convenient property of the central element $F_n$ (and that this property is not shared by the central element that appears in Westbury).  After a brief summary reporting what has been proven, we conclude with two appendices.  The first defines $F_n$ and proves the properties that we require.  The second lists a selection of standard representation-theoretic results that are needed.  The text is liberally peppered throughout with examples chosen to highlight the theory being developed.  We hope that the reader will find them useful for understanding this beautiful corner of representation theory.

We end this introduction with a list of common symbols and terms.

\section*{Glossary of Terms and Symbols}

\begin{tabular}{ll}
$\LinkMod_{n}$ and $\LinkMod_{n,p}$: link modules      &  beginning of \secref{sec:Reps}\\
$\NatMod_{n,p}=\LinkMod_{n,p}/\LinkMod_{n,p+1}$: standard modules   & \eqnref{eqn:DefNatMod} \\
$\RadMod_{n,p}\subset\NatMod_{n,p}$: radical of $\langle\,\cdot, \cdot \rangle_{n,p}$ on $\NatMod_{n,p}$ & \eqnref{eq:DefRadMod} and \propref{prop:S/R}\\
$\IrrMod_{n,p}=\NatMod_{n,p}/\RadMod_{n,p}$: irreducible quotient of the standard modules & after \propref{prop:S/R} \\
$\ProjMod_{n,p}$: principal indecomposable module with quotient $\IrrMod_{n,p}$ & beginning of \secref{sec:Proj}\\
$\Res{\mathcal M}$: restriction of $\mathcal M$ from $\tl{n}$ to $\tl{n-1}$ & \propref{prop:Restrict} \\
$\Ind{\mathcal M}$: induction of $\mathcal M$ from $\tl{n}$ to $\tl{n+1}$ & beginning of \secref{sec:Ind} \\
basis $\mathcal S_{n,p}$ for the induced module $\Ind{\NatMod_{n,p}}$ & before \propref{prop:SisTheBasis}\\
$\langle\,\cdot, \cdot\rangle=\langle\,\cdot, \cdot\rangle_{n,p}$: bilinear form on $\NatMod_{n,p}$ & before \lemref{lem:Invariance} \\
$(n,p)$ is critical if $q^{2(n-2p+1)}=1$ & after the proof of \corref{cor:ResSplit} \\
critical lines, critical strips of the Bratteli diagram & after \corref{cor:DimL}\\
$d_{n,p}$: number of $(n,p)$-link states & before \eqnref{eqn:RecD}\\
$F_n$ and $f_{n,p}$: central element of $\tl n$ and its eigenvalues & \appref{app:Casimir} and \propref{lem:FEigs} \\
$G_{n,p}$: Gram matrices ($\langle\,\cdot, \cdot\rangle_{n,p}$ in the basis of link states) & beginning of \secref{sec:Gram}\\
Jones' normal form and reverse normal form & \propref{prop:JNO} \\
simple link in a diagram or link state & before \eqnref{eq:multByUi} and \lemref{lem:Basis} \\
admissible link state & before \lemref{lem:betterSpanningSet} \\
a symmetric pair $V_{n,p}$ and $V_{n,p'}$   & beginning of \secref{sec:R=L}\\
\end{tabular}

%
%

\section{Diagrams and Presentations} \label{sec:Diagrams}

Let us define an $n$-diagram algebra, for $n$ a positive integer, as follows.  First, draw two parallel vertical lines and mark $n$ points on each.  The $2n$ points obtained are then connected pairwise by $n$ links (arcs) which do not intersect and which lie entirely between the two vertical borders.  This gives what we will call an \emph{$n$-diagram}.  Second, form the complex vector space spanned formally by the set of all $n$-diagrams.  Third, endow this vector space with a multiplication given by concatenating two $n$-diagrams, replacing each interior loop by a factor $\beta \in \CC$, and identifying (and then removing) the two interior vertical borders.  For example, for $n = 3$,
\begin{equation}
\parbox[c]{12mm}{
\includegraphics[height=16mm]{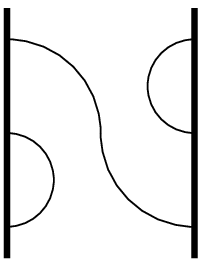}
}
\ \cdot \ 
\parbox[c]{12mm}{
\includegraphics[height=16mm]{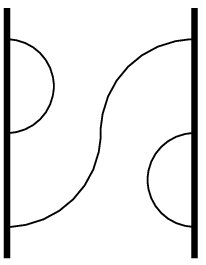}
}
\ = \ 
\parbox[c]{24mm}{
\includegraphics[height=16mm]{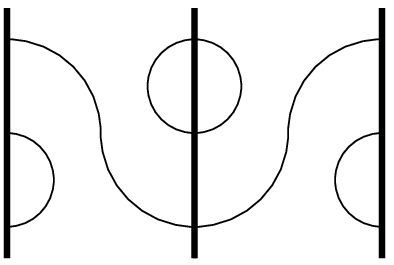}
}
\ = \beta \ 
\parbox[c]{12mm}{
\includegraphics[height=16mm]{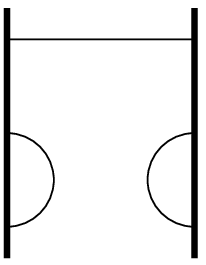}
}
\ .
\end{equation}
This multiplication defines, for each $\beta \in \CC$, an associative algebra called the $n$-diagram algebra.  It is easily checked that this algebra has a unit given by the diagram
\begin{center}
\parbox{12mm}{\begin{center}
\includegraphics[height=16mm]{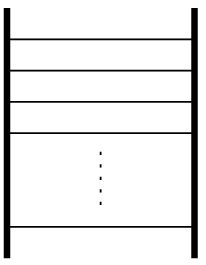}
\end{center}}
\ .
\end{center}
We note here for future reference that much of what follows is simplified algebraically upon making the identification $\beta = q + q^{-1}$
with $q\in\mathbb{C}^{\times}$.  Of course this implies that the resulting theory must be invariant under the exchange of $q$ with $q^{-1}$.

The $n$-diagram algebra is in fact isomorphic to the Temperley-Lieb algebra $\tl{n}$, as we shall see.  The latter algebra is abstractly defined as being generated by a unit $\wun$ and elements $u_i$, $i = 1, 2, \ldots , n-1$, satisfying
\begin{equation} \label{eqn:TLrels}
u_i^2 = \beta u_i, \qquad u_i u_{i \pm 1} u_i = u_i \qquad \text{and} \qquad u_i u_j = u_j u_i \quad \text{if $\abs{i-j} > 1$.}
\end{equation}
Indeed, if we make the identifications
\begin{equation} \label{eqn:DefTLGens}
\wun = \ 
\parbox{12mm}{\begin{center}
\includegraphics[height=16mm]{tl-id}
\end{center}}
\qquad \text{and} \qquad u_i = \ 
\parbox{16mm}{\begin{center}
\psfrag{1}[][]{$\scriptstyle 1$}
\psfrag{i-1}[][]{$\scriptstyle i-1$}
\psfrag{i}[][]{$\scriptstyle i$}
\psfrag{i+1}[][]{$\scriptstyle i+1$}
\psfrag{i+2}[][]{$\scriptstyle i+2$}
\psfrag{n}[][]{$\scriptstyle n$}
\includegraphics[height=24mm]{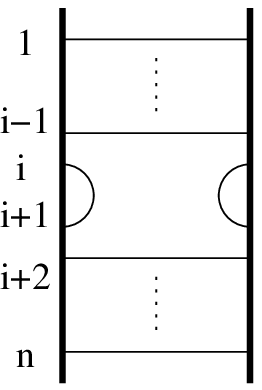}
\end{center}}
\ ,
\end{equation}
then we can verify explicitly the properties of the unit (which are clear) and the defining relations
\begin{subequations}
\begin{align}
u_i^2 &= \ 
\parbox{24mm}{\begin{center}
\includegraphics[height=24mm]{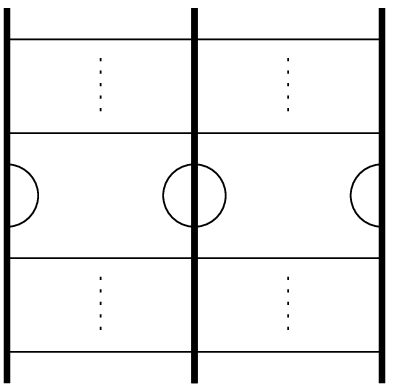}
\end{center}}
\ = \beta \ 
\parbox{12mm}{\begin{center}
\includegraphics[height=24mm]{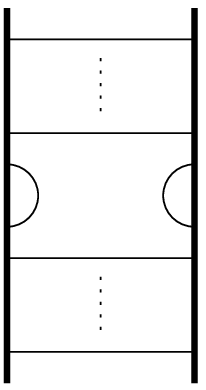}
\end{center}}
\ = \beta u_i, \\
u_i u_{i-1} u_i &= \ 
\parbox{36mm}{\begin{center}
\includegraphics[height=28mm]{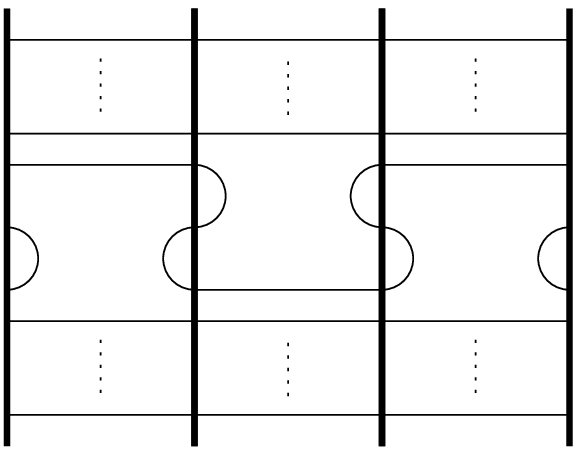}
\end{center}}
\ = \ 
\parbox{12mm}{\begin{center}
\includegraphics[height=28mm]{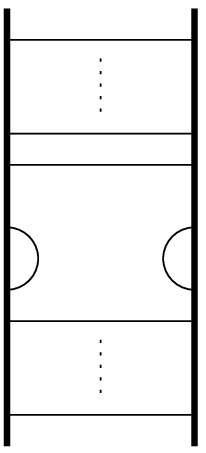}
\end{center}}
\ = u_i, \label{eqn:uuu=u} \\
\text{and} \qquad u_i u_j &= \ 
\parbox{24mm}{\begin{center}
\includegraphics[height=32mm]{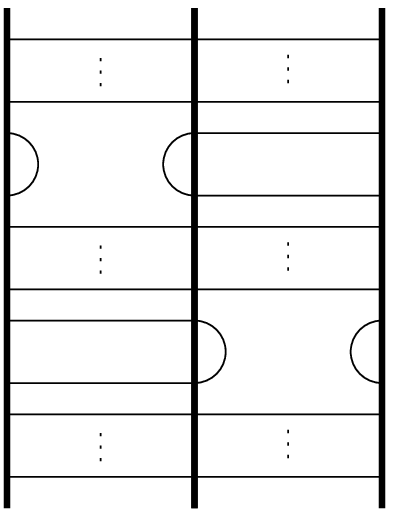}
\end{center}}
\ = 
\parbox{12mm}{\begin{center}
\includegraphics[height=32mm]{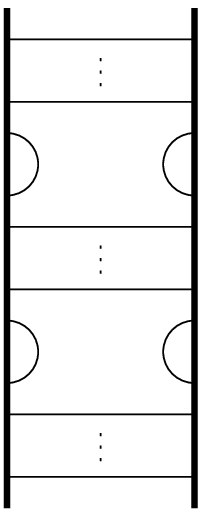}
\end{center}}
\ = u_j u_i \qquad \text{($\abs{i-j} > 1$).}
\end{align}
\end{subequations}
The verification of $u_i u_{i+1} u_i = u_i$ is just \eqref{eqn:uuu=u} viewed upside-down.  More formally, note that the Temperley-Lieb algebra $\tl{n}$ has an automorphism specified by $\wun \mapsto \wun$ and $u_i \mapsto u_{n-i}$.  This corresponds to reflecting our diagrams about a horizontal line. Finally, the map $u_i\mapsto -u_i$ defines an isomorphism between the two algebras $\tl{n}$ with parameters $\beta$ and $-\beta$.

It is clear that the identification \eqref{eqn:DefTLGens} defines a homomorphism from the algebra $\tl{n}$, defined abstractly through \eqref{eqn:TLrels}, to the $n$-diagram algebra. Let $W_n$ be the set of words in the letters $u_i$, $1\leqslant i\leqslant n-1$, and let $D_n$ be the set of $n$-diagrams. The following result notes that the algebra of $n$-diagrams is generated by the diagrams in \eqref{eqn:DefTLGens}.
\begin{lemma} \label{lem:surjectivity}
The map \eqref{eqn:DefTLGens} from $W_n$ to $D_n$ is surjective.
\end{lemma}
\noindent The proof will be deferred until the end of the section.

This lemma demonstrates that the $n$-diagram algebra is a quotient of the abstract Temperley-Lieb algebra --- the diagram algebra might satisfy further independent relations.  To show that there are no further relations, hence that we have the isomorphism of algebras claimed above, we only need to show that the dimensions of the diagram algebra and the Temperley-Lieb algebra coincide.  For this, let us first consider the ``half-diagrams'' obtained from an $n$-diagram by cutting vertically down the middle.  Each half then has $n$ marked points, but only some of these will still be connected by ``links''.  If a half-diagram has $p$ links, then there will be $n-2p$ points which are not connected to anything.  The latter points will be referred to as \emph{defects} and a half-diagram with $n$ points and $p$ links will be called an $\brac{n,p}$-\emph{link state}.  An example with $n=7$ and $p=3$, hence $1$ defect, is given by
\begin{center}
\parbox{8mm}{\begin{center}
\includegraphics[height=28mm]{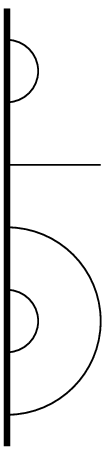}
\end{center}}
\ .
\end{center}
We will always orient our link states so that the links and defects face to the right.

An $n$-diagram can now be cut in half and reassembled into a $\brac{2n,n}$-link state by rotating the incorrectly oriented half so that it lies below the other, and then rejoining the defects of each half as they were joined in the original $n$-diagram:
\begin{center}
\parbox{12mm}{\begin{center}
\includegraphics[height=28mm]{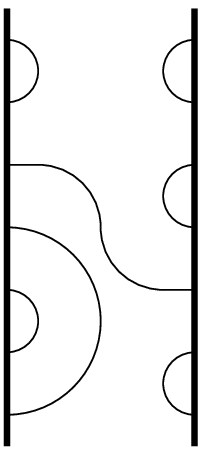}
\end{center}}
\hspace{5mm} $\longrightarrow$ \hspace{5mm}
\parbox{12mm}{\begin{center}
\includegraphics[height=28mm]{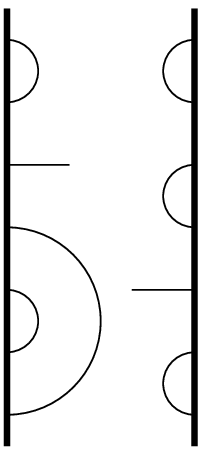}
\end{center}}
\hspace{5mm} $\longrightarrow$ \hspace{5mm}
\parbox{6mm}{\begin{center}
\includegraphics[height=56mm]{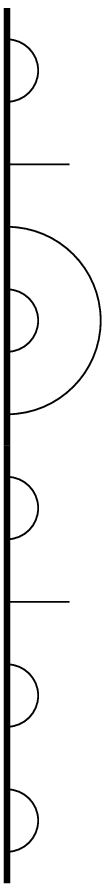}
\end{center}}
\hspace{5mm} $\longrightarrow$ \hspace{5mm}
\parbox{10mm}{\begin{center}
\includegraphics[height=56mm]{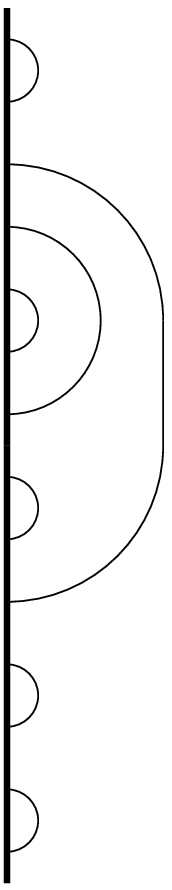}
\end{center}}
\ .
\end{center}
As this procedure is obviously reversible, this establishes a bijection between $n$-diagrams and $\brac{2n,n}$-link states.

The set of all $\brac{n,p}$-link states is in turn in bijection with the increasing walks on $\ZZ^2$ from $\brac{0,0}$ to $\brac{n-p,p}$ which avoid crossing the diagonal $\brac{m,m}$.  Here, ``increasing'' means that the walker may only move up or to the right at each step.  This bijection is easy to describe:  Reading an $\brac{n,p}$-link state from top to bottom, the walker moves up at the $k$-th step if the $k$-th marked point closes a link.  Otherwise, the walker moves right.  For the example above with $n=7$, $p=3$,
\begin{center}
\parbox{6mm}{\begin{center}
\includegraphics[height=28mm]{tl-link-ex}
\end{center}}
\hspace{5mm} $\longrightarrow$ \hspace{5mm}
\parbox{36mm}{\begin{center}
\psfrag{a}[][]{$\scriptstyle \brac{0,0}$}
\psfrag{b}[][]{$\scriptstyle \brac{4,3}$}
\includegraphics[height=24mm]{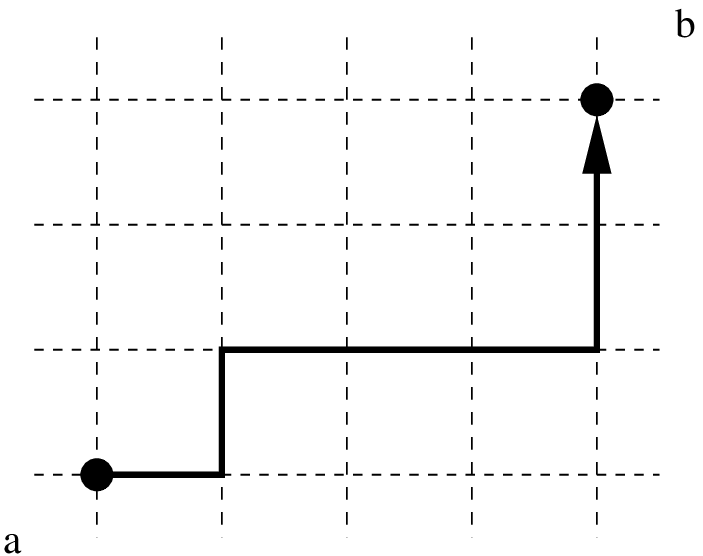}
\end{center}}
\ .
\end{center}
The walk will never cross the diagonal because we cannot close a link without first opening it.  Conversely, given such an increasing path, we follow it backwards starting from $\brac{n-p,p}$ and construct the corresponding $\brac{n,p}$-link state from bottom to top:  Every time we move down, we open a link.  Every time we move left, we close the newest open link, if one exists to be closed.  If one does not, then that marked point becomes a defect.  As the walk is never above the diagonal, there are never more down moves than left moves remaining, so we are never left with an open link that cannot be closed.  The bijection is now clear.

This proves that the number of $\brac{n,p}$-link states is equal to the number of increasing walks on $\ZZ^2$ from $\brac{0,0}$ to $\brac{n-p,p}$ which avoid crossing the diagonal.  In particular, the dimension of the $n$-diagram algebra is therefore the number of increasing walks on $\ZZ^2$ from $\brac{0,0}$ to $\brac{n,n}$ which avoid crossing the diagonal.  We will now establish the same result for the abstract Temperley-Lieb algebra $\tl{n}$.  Here we follow the seminal paper of Jones \cite{JonInd83}.

First, define a \emph{word} to be a monomial in the Temperley-Lieb generators, $u_{i_1} u_{i_2} \cdots u_{i_k}$ for example.  A word is said to be \emph{reduced} if we cannot use the relations \eqref{eqn:TLrels} to rewrite it with fewer generators.  The key insight into converting a given word into reduced form is contained in the following result.
\begin{lemma} \label{lem:Jones}
In any reduced Temperley-Lieb word $u_{i_1} u_{i_2} \cdots u_{i_k}$, the maximal index $m = \max \set{i_j \st j = 1 , \ldots , k}$ occurs only once.
\end{lemma}
\begin{proof}
We use induction on the maximal index $m$ of our reduced word, the case $m=1$ being obvious.  Suppose then that we have a reduced word in which $u_m$ appears twice or more, so our word has the form $\cdots u_m U u_m \cdots$, in which the maximal index $m'$ of $U$ is less than $m$.  $U$ must also be reduced (otherwise our word would not be), so we may assume that its maximal index appears only once.  Now if $m' < m-1$, then $u_m$ commutes with every generator appearing in $U$, so we can bring the two $u_m$'s together and use $u_m^2 = \beta u_m$ to reduce our word even further.  But, $m' = m-1$ means that $u_m$ commutes with every generator in $U$ except $u_{m-1}$, so this generator appearing only once in $U$ means that we can sandwich it between the two $u_m$'s and use $u_m u_{m-1} u_m = u_m$.  Both possibilities contradict the assumption that our word is reduced, so we conclude that $u_m$ can only appear once.
\end{proof}

This simple argument allows one to \emph{order} reduced words by pushing the maximal index as far as possible to the right.  More specifically, if $U$ is a reduced word and $m$ the maximal index appearing then we claim that it is possible to commute the $u_i$ so that we have the form
\begin{equation}
U = U' \brac{u_m u_{m-1} \cdots u_{m - \ell}},
\end{equation}
where $U'$ is a reduced word whose maximal index is less than $m$.  If there were a gap in the sequence of indices following $U'$, then the element after the gap would commute with the generators to the left of the gap (up to $u_m$), so could be relocated to the left of $u_m$ (and absorbed in $U'$).  If the sequence of indices following $U'$ did not decrease, then one could use $u_i^2 = \beta u_i$ or $u_i u_{i-1} u_i = u_i$ to further reduce the word (contradiction).  The claim is therefore established.

Since $U'$ is a shorter reduced word with a smaller maximal index, the same arguments apply, allowing us to write it as $U''$ times another such uniformly decreasing sequence of generators.  Induction then leads us to Jones' normal form for reduced Temperley-Lieb words:
\begin{proposition}[Jones' Normal Form] \label{prop:JNO}
Any reduced Temperley-Lieb word $U\in\tl{n}$ may be written in the form
\begin{equation} \label{eqn:TLRedWords}
U = \brac{u_{j_1} u_{j_1 - 1} \cdots u_{k_1}} \brac{u_{j_2} u_{j_2 - 1} \cdots u_{k_2}} \cdots \brac{u_{j_r} u_{j_r - 1} \cdots u_{k_r}},
\end{equation}
where $0 < j_1 < \cdots < j_{r-1} < j_r < n$ and $0 < k_1 < \cdots < k_{r-1} < k_r < n$. Similarly, any reduced word may also be written as
\begin{equation} \label{eqn:revTLRedWords}
U = \brac{u_{j_1} u_{j_1+1} \cdots u_{k_1}} \brac{u_{j_2} u_{j_2+1} \cdots u_{k_2}} \cdots \brac{u_{j_r} u_{j_r+1}\cdots u_{k_r}},
\end{equation}
with $n>j_1>j_2>\cdots>j_r>0$ and $n>k_1>k_2>\cdots>k_r>0$.
\end{proposition}
\begin{proof}
The increasing nature of the $j_i$ follows from the fact that they are the maximal indices of the subwords $\ldots , U'' , U' , U$.  For the $k_i$, note that if $k_i \geqslant k_{i+1}$ for some $i$, then we could commute the $u_{k_i}$ appearing at the end of its decreasing sequence to the right until it bumped up against the $u_{k_i + 1} u_{k_i}$ appearing in the next decreasing sequence.  The defining relations \eqref{eqn:TLrels} again give us a contradiction to our word being reduced.  Thus, we must have $0 < k_1 < \cdots < k_{r-1} < k_r < n$ as well. The proof of the second form is similar.
\end{proof}
\noindent We shall refer to words satisfying \eqref{eqn:TLRedWords} as elements in Jones' normal form and those satisfying \eqref{eqn:revTLRedWords} as being in reverse Jones' normal form.

The point of this theory is to establish that every monomial in the Temperley-Lieb generators is (up to factors of $\beta$) equal to an ordered reduced word of the form \eqref{eqn:TLRedWords}, with strictly increasing $j_i$ and $k_i$.  The dimension of the algebra $\tl{n}$ is then bounded above by the number of such ordered reduced words.\footnote{When $\beta = 0$, we have the slight modification that every monomial in the Temperley-Lieb generators is either equal to a uniquely ordered reduced word of the form \eqref{eqn:TLRedWords}, with strictly increasing $j_i$ and $k_i$, or is identically zero.  The conclusion remains valid.}  Note that if these ordered reduced words were linearly independent, then this bound would become an equality.  Our next task is to show that this number coincides with the number of increasing walks on $\ZZ^2$ from $\brac{0,0}$ to $\brac{n,n}$ which avoid crossing the diagonal.  Then, we will have shown that $\tl{n}$ has a quotient (the $n$-diagram algebra) whose dimension is equal to our upper bound on the dimension of $\tl{n}$.  It then follows that the bound is an equality, so the set of ordered reduced words is in fact a basis.

Naturally, this is achieved by constructing a bijection.  This is done by stripping the indices $j_i$ and $k_i$ from the ordered reduced words \eqref{eqn:TLRedWords} of $\tl{n}$ and encoding them as a walk on $\ZZ^2$ as follows:
\begin{multline*}
\brac{u_{j_1} u_{j_1 - 1} \cdots u_{k_1}} \brac{u_{j_2} u_{j_2 - 1} \cdots u_{k_2}} \cdots \brac{u_{j_r} u_{j_r - 1} \cdots u_{k_r}} \\
\longrightarrow \Bigl\langle \brac{0,0} \rightarrow \brac{j_1,0} \rightarrow \brac{j_1,k_1} \rightarrow \brac{j_2,k_1} \rightarrow \brac{j_2,k_2} \rightarrow \cdots \rightarrow \brac{j_r,k_{r-1}} \rightarrow \brac{j_r,k_r} \rightarrow \brac{n,k_r} \rightarrow \brac{n,n} \Bigr\rangle.
\end{multline*}
(The empty word corresponding to the unit $\wun$ in $\tl n$ is encoded as $\langle (0,0)\rightarrow (n,0)\rightarrow (n,n)\rangle$.) For example, in $\tl{9}$ we have
\begin{center}
$\brac{u_1} \brac{u_4 u_3 u_2} \brac{u_6 u_5 u_4 u_3} \brac{u_8 u_7}$ \hspace{5mm} $\longrightarrow$ \hspace{5mm}
\parbox{40mm}{\begin{center}
\psfrag{a}[][]{$\scriptstyle \brac{0,0}$}
\psfrag{b}[][]{$\scriptstyle \brac{9,9}$}
\includegraphics[height=36mm]{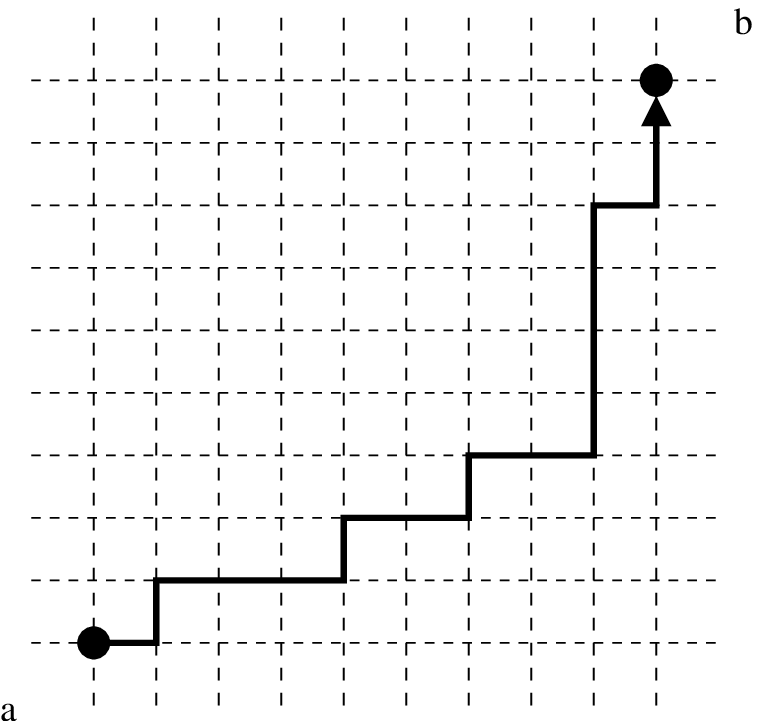}
\end{center}}
\ .
\end{center}
The walks so obtained cannot cross the diagonal because $j_{i+1} > j_i \geqslant k_i$, and they are increasing because $j_i < j_{i+1}$ and $k_i < k_{i+1}$.  Conversely, any increasing path from $\brac{0,0}$ to $\brac{n,n}$ is entirely determined by its corners.  The coordinates of these corners may be used to construct a reduced word of the form \eqref{eqn:TLRedWords} and it is easy to check that the $j_i$ and $k_i$ so defined satisfy $j_i < j_{i+1}$ and $k_i < k_{i+1}$.  Avoiding crossing the diagonal translates into $j_i \geqslant k_i$.

This bijection therefore completes the proof that the dimension of the $n$-diagram algebra coincides with that of the abstract Temperley-Lieb algebra $\tl{n}$ (this dimension is obviously finite).  As noted above, this is sufficient to conclude that these algebras are in fact isomorphic.  It is not difficult now to complete this circle of ideas and actually compute the dimensions of these algebras by counting the increasing paths from $\brac{0,0}$ to $\brac{n,n}$ which avoid crossing the diagonal.

In fact, we shall generalise slightly and count the increasing paths from $\brac{0,0}$ to $\brac{n-p,p}$ which avoid crossing the diagonal.  This more general computation then gives the number of $\brac{n,p}$-link states.  Let this number be $d_{n,p}$.  Since an increasing path ending at $\brac{n-p,p}$ must pass through either $\brac{n-p-1,p}$ or $\brac{n-p,p-1}$ (but not both!), we obtain the recursion relation
\begin{equation} \label{eqn:RecD}
d_{n,p} = d_{n-1,p} + d_{n-1,p-1}.
\end{equation}
With the ``boundary values'' $d_{n,0} = 1$ (the only increasing path which never goes up is the one which only moves to the right) and $d_{2p-1,p} = 0$ (the paths are not allowed to cross the diagonal), this completely determines the numbers $d_{n,p}$.  Its solution is easily verified to be
\begin{equation} \label{eqn:DimV}
d_{n,p} = \binom{n}{p} - \binom{n}{p-1},
\end{equation}
where we understand, as usual, that $\tbinom{n}{-1} = 0$.  It follows that the $n$-diagram algebra and the Temperley-Lieb algebra $\tl{n}$ have dimension
\begin{equation} \label{eqn:DimTL_n}
\dim \tl{n} = d_{2n,n} = \binom{2n}{n} - \binom{2n}{n-1} = \frac{1}{n+1} \binom{2n}{n}.
\end{equation}
This is of course the $n$-th Catalan number.

To summarise, we have proven the following result in this section.
\begin{theorem} \label{thm:diagrams=TL}
The abstract Temperley-Lieb algebra $\tl{n}$ and the algebra of $n$-diagrams are isomorphic (for given $n \in \ZZ_+$ and $\beta \in \CC$).  The dimensions of these algebras are given by \eqnref{eqn:DimTL_n}.
\end{theorem}
\noindent We mention that combining the result on dimensions with the ``cutting'' of $n$-diagrams into two $\brac{n,p}$-link states leads to the curious identity
\begin{equation} \label{eqn:Curious}
d_{2n,n} = \sum_{p=0}^{\flr{n/2}} d_{n,p}^2.
\end{equation}
This only requires observing that the two link states must have the same number of links $p$ so that each defect of the first may be joined to a defect of the second.

\medskip

We close this section with a proof of \lemref{lem:surjectivity} that is purely diagrammatic. An example of a similar construction is given by Kauffman (see Figure 16 and the proof of Theorem 4.3 in \cite{kauffman}). Before starting, we will introduce some new terms and make two preparatory observations concerning diagram multiplication. A \emph{simple} link (at position $i$) is a link that connects positions $i$ and $i+1$ on the same side of a diagram and a link connecting positions on opposite sides of a diagram is called a \emph{through-line}. Let $d$ be a diagram that has a simple link on the left at position $i-1$ and a through-line connecting $i+1$ on the left to $i+1$ on the right. Upon multiplying on the left by $u_i$, or rather by the diagram corresponding to $u_i$, this diagram becomes 
\begin{equation}\label{eq:multByUi}
d = 
\begin{tikzpicture}[baseline={(current bounding box.center)},scale=1/3]
	\draw[thick] (3,1) -- (3,7);
	\draw[thick] (0,1) -- (0,7);
	\draw (0,3) -- (3,3);
	\draw (0,4) arc (-90:90:0.5);
	\node at (-1,5) {$\scriptstyle{i-1}$};
	\node at (-1,4) {$\scriptstyle{i}$};
	\node at (-1,3) {$\scriptstyle{i+1}$};
	\node at (1.5,6.5) {$\vdots$};
	\node at (1.5,2) {$\vdots$};
\end{tikzpicture}
\quad \Rightarrow \quad u_i d = 
\begin{tikzpicture}[baseline={(current bounding box.center)},scale=1/3]
	\draw[thick] (6,1) -- (6,7);
	\draw[thick] (3,1) -- (3,7);
	\draw[thick] (0,1) -- (0,7);
	\draw (0,5) -- (3,5);
	\draw (0,3) arc (-90:90:0.5);
	\draw (3,4) arc (90:270:0.5);
	\node at (1.5,6.5) {$\vdots$};
	\node at (1.5,2) {$\vdots$};
	\draw (3,3) -- (6,3);
	\draw (3,4) arc (-90:90:0.5);
	\node at (-1,5) {$\scriptstyle{i-1}$};
	\node at (-1,4) {$\scriptstyle{i}$};
	\node at (-1,3) {$\scriptstyle{i+1}$};
	\node at (4.5,6.5) {$\vdots$};
	\node at (4.5,2) {$\vdots$};
\end{tikzpicture}
\ =
\begin{tikzpicture}[baseline={(current bounding box.center)},scale=1/3]
	\draw[thick] (3,1) -- (3,7);
	\draw[thick] (0,1) -- (0,7);
	\draw[out=0,in=180] (0,5) to (3,3);
	\draw (0,3) arc (-90:90:0.5);
	\node at (-1,5) {$\scriptstyle{i-1}$};
	\node at (-1,4) {$\scriptstyle{i}$};
	\node at (-1,3) {$\scriptstyle{i+1}$};
	\node at (1.5,6.5) {$\vdots$};
	\node at (1.5,2) {$\vdots$};
\end{tikzpicture}
\ .
\end{equation}
The vertical dots in $d$ indicate links and through-lines that are not explicitly drawn. Clearly, these remain unchanged upon multiplication by $u_i$. 
The result $u_i d$ is therefore identical to the original $d$ except that the simple link at position $i-1$ on the left has migrated to position $i$ and the through-line that started on the left at position $i+1$ now starts at $i-1$. A simple link on the left at $i-1$ and a through-line starting at $i+1$ may be exchanged by left-multiplying by $u_i$.

This observation generalises easily.  If $d$ also has a simple link at position $i-3$, then the product $u_{i-2} u_i d$ would have simple links at $i-2$ and $i$.  Moreover, the through-line would now start at $i-3$. Of course, right-multiplication by $u_i$ would result in similar exchanges to simple links and through-lines on the right side. The conclusion is that the generators $u_i$ can be used to exchange simple links and through-lines, without changing any other patterns in a given diagram.  This is the first preparatory observation.

For the second, we suppose that $d$ has, on its left side, $p$ consecutive simple links from positions $i$ to $i+2p-1$. Left-multiplication of $d$ by (the diagram representing the word) $u_{i+1}u_{i+3}\dots u_{i+2p-3}$ creates an overarching link from $i$ to $i+2p-1$, while shifting the remaining $p-1$ simple links so that they fit inside:
\begin{equation} \label{eq:multByManyUi}
d = \mspace{-10mu}
\begin{tikzpicture}[baseline={(current bounding box.center)},scale=1/3]
	\draw[thick] (3,1) -- (3,11);
	\draw[thick] (0,1) -- (0,11);
	\node at (2.25,10) {$\vdots$};
	\node at (2.25,3) {$\vdots$};
	\draw (0,2) arc (-90:90:0.5);
	\draw (0,4) arc (-90:90:0.5);
	\node at (0.5,6.3) {$\vdots$};
	\draw (0,7) arc (-90:90:0.5);
	\draw (0,9) arc (-90:90:0.5);
	\node at (-1.5,2) {$\scriptstyle{i+2p-1}$};
	\node at (-0.5,10) {$\scriptstyle{i}$};
	\node at (-0.85,8) {$\scriptstyle{i+2}$};
\end{tikzpicture}
\quad \Rightarrow \quad 
u_{i+1}u_{i+3}\cdots u_{i+2p-3}d = \mspace{-10mu}
\begin{tikzpicture}[baseline={(current bounding box.center)},scale=1/3]
	\draw[thick] (6,1) -- (6,11);
	\draw[thick] (3,1) -- (3,11);
	\draw[thick] (0,1) -- (0,11);
	\node at (5.25,10) {$\vdots$};
	\node at (5.25,3) {$\vdots$};
	\draw (0,10) -- (3,10);
	\draw (0,2) -- (3,2);
	\draw (3,2) arc (-90:90:0.5);
	\draw (3,4) arc (-90:90:0.5);
	\node at (3.5,6.3) {$\vdots$};
	\draw (3,7) arc (-90:90:0.5);
	\draw (3,9) arc (-90:90:0.5);
	\draw (0,3) arc (-90:90:0.5);
	\node at (0.5,5.3) {$\vdots$};
	\draw (0,6) arc (-90:90:0.5);
	\draw (0,8) arc (-90:90:0.5);
	\draw (3,4) arc (90:270:0.5);
	\node at (2.5,5.3) {$\vdots$};
	\draw (3,7) arc (90:270:0.5);
	\draw (3,9) arc (90:270:0.5);
	\node at (-0.85,9) {$\scriptstyle{i+1}$};
	\node at (-0.85,7) {$\scriptstyle{i+3}$};
	\node at (-1.5,4) {$\scriptstyle{i+2p-3}$};
\end{tikzpicture}
\ = \mspace{-10mu}
\begin{tikzpicture}[baseline={(current bounding box.center)},scale=1/3]
	\draw[thick] (3,1) -- (3,11);
	\draw[thick] (0,1) -- (0,11);
	\node at (2.25,10) {$\vdots$};
	\node at (2.25,3) {$\vdots$};
	\draw (0,3) arc (-90:90:0.5);
	\node at (0.5,5.3) {$\vdots$};
	\draw (0,6) arc (-90:90:0.5);
	\draw (0,8) arc (-90:90:0.5);
	\draw[rounded corners] (0,2) -- (1,2) -- (1,10) -- (0,10);
	\node at (-1.5,2) {$\scriptstyle{i+2p-1}$};
	\node at (-0.5,10) {$\scriptstyle{i}$};
	\node at (-0.85,8) {$\scriptstyle{i+2}$};
\end{tikzpicture}
\ \ .
\end{equation}
As in the previous example, the other patterns in $d$ remain unchanged.  This is the second preparatory observation.

Finally, we shall define a \emph{nested island} to be a pattern of $p$ links joining $2p$ consecutive points on one side of an $n$-diagram satisfying two conditions.  First, that the top point $i$ and the bottom $i+2p-1$ are linked (the island does not consist of two or more smaller ``\emph{sub-islands}''), and second, that no point above (less than) $i$ is linked to any point below (greater than) $i+2p-1$ (the island is not a sub-island of a greater island). To any $n$-diagram $d$, we may therefore associate another $n$-diagram $d'$ in which all the nested islands of $d$ are replaced by consecutive simple links. We illustrate this replacement with a diagram $d$ that has four nested islands, one on the left and three on the right, each being circumscribed for clarity by a dotted line.
\begin{equation} \label{eq:diagramD}
d =
\begin{tikzpicture}[baseline={(current bounding box.center)},scale=1/3]
	\draw[thick] (3,0) -- (3,11);
	\draw[thick] (0,0) -- (0,11);
	\draw (0,4) arc (-90:90:0.5);
	\draw[rounded corners] (0,2) -- (0.5,2) -- (1.2,3) -- (1.2,8) -- (0.5,9) -- (0,9);
	\draw (0,7) arc (-90:90:0.5);
	\draw[out=0,in=180] (0,1) to (3,3);
	\draw[out=0,in=180] (0,10) to (3,8);
	\draw[out=0,in=0] (0,3) to (0,6);
	\draw (3,10) arc (90:270:0.5);
	\draw[out=180,in=180] (3,4) to (3,7);
	\draw (3,6) arc (90:270:0.5);
	\draw (3,2) arc (90:270:0.5);
	\draw[dotted,thick,rounded corners] (3.5,9.5) -- (3.5,10.5) -- (2,10.5) -- (2,8.5) -- (3.5,8.5) -- (3.5,9.5);	
	\draw[dotted,thick,rounded corners] (3.5,1.5) -- (3.5,2.5) -- (2,2.5) -- (2,0.5) -- (3.5,0.5) -- (3.5,1.5);
	\draw[dotted,thick,rounded corners] (3.5,4.5) -- (3.5,7.5) -- (2,7.5) -- (2,3.5) -- (3.5,3.5) -- (3.5,4.5);	
	\draw[dotted,thick,rounded corners] (-0.5,4.5) -- (-0.5,1.5) -- (0.75,1.5) -- (1.5,3) -- (1.5,8) -- (0.75,9.5) -- (-0.5,9.5) -- (-0.5,4.5);
\end{tikzpicture}
\qquad \Rightarrow \qquad d' =
\begin{tikzpicture}[baseline={(current bounding box.center)},scale=1/3]
	\draw[thick] (3,0) -- (3,11);
	\draw[thick] (0,0) -- (0,11);
	\draw (0,2) arc (-90:90:0.5);
	\draw (0,4) arc (-90:90:0.5);
	\draw (0,6) arc (-90:90:0.5);
	\draw (0,8) arc (-90:90:0.5);
	\draw[out=0,in=180] (0,1) to (3,3);
	\draw[out=0,in=180] (0,10) to (3,8);
	\draw (3,10) arc (90:270:0.5);
	\draw (3,7) arc (90:270:0.5);
	\draw (3,5) arc (90:270:0.5);
	\draw (3,2) arc (90:270:0.5);
	\draw[dotted,thick,rounded corners] (3.5,9.5) -- (3.5,10.5) -- (2,10.5) -- (2,8.5) -- (3.5,8.5) -- (3.5,9.5);	
	\draw[dotted,thick,rounded corners] (3.5,1.5) -- (3.5,2.5) -- (2,2.5) -- (2,0.5) -- (3.5,0.5) -- (3.5,1.5);
	\draw[dotted,thick,rounded corners] (3.5,4.5) -- (3.5,7.5) -- (2,7.5) -- (2,3.5) -- (3.5,3.5) -- (3.5,4.5);	
	\draw[dotted,thick,rounded corners] (-0.5,4.5) -- (-0.5,1.5) -- (0.75,1.5) -- (1.5,3) -- (1.5,8) -- (0.75,9.5) -- (-0.5,9.5) -- (-0.5,4.5);
\end{tikzpicture}\ \ .
\end{equation}
With these preparatory observations and definitions, the following proof of \lemref{lem:surjectivity} proceeds in two elementary steps. We will illustrate these steps for the diagram $d$ after giving the proof.
\begin{proof}[Proof of \lemref{lem:surjectivity}] 
Given a diagram $d$, the proof constructs a word $w$ in the $u_i$, with $1 \leqslant i \leqslant n-1$, such that $w$ is identified with $d$ under \eqref{eqn:DefTLGens}. We first obtain a word $w'$ which is identified with $d'$, the $n$-diagram associated to $d$ in which its nested islands are replaced by simple links. Then, this word $w'$ is extended to the desired word $w$ by reconstructing the nested islands. Note that if all links starting from the left side of $d$ cross to the right side, the diagram is in fact the diagram representing the unit. In this case, the word $w$ is simply the empty word corresponding to $\wun\in\tl{n}$.

Suppose then that $2p>0$ is the number of points on the left side of $d$ that are linked pairwise. Clearly, there are also $2p$ linked points on the right side and $n-2p\geqslant 0$ is the number of through-lines in $d$. The diagram $d'$ thus contains $p$ simple links on each of its sides. To construct $w'$, we start from the word $u_1u_3\dots u_{2p-1}$. The corresponding diagram has, on each of its sides, $p$ consecutive simple links followed by $n-2p$ through-lines. Some of these through-lines might not be at the positions of those of $d'$ (and $d$), but the first preparatory observation above indicates how to move them to the correct positions using left- and right-multiplication by the $u_i$. The uppermost through-line in $u_1u_3\dots u_{2p-1}$ (that at $2p+1$) is first moved, using the observation, to the position of the uppermost through-line in $d'$. When this is accomplished, the positions under this through-line are consecutive simple links followed by the remaining $n-2p-1$ through-lines, if any. Again, if the positions of these through-lines do not match those of $d'$, the observation can be used again. Repeating this process at most $n-2p$ times gives a word $w'$ that corresponds precisely to the diagram $d'$.

Having constructed a word $w'$ that corresponds to $d'$, it remains to convert the simple links of $d'$ into nested islands as they appear in $d$. For each nested island that is not a simple link, we first draw the outermost link. More precisely, if a nested island lies between positions $i$ and $i+2j-1$ on the left (right) side of $d$, then the word $w'$ is multiplied from the left (right) by $u_{i+1}u_{i+3}\cdots u_{i+2j-3}$, as the second preliminary observation advises. Once this has been done for each nested island, the outermost links of all nested islands have been constructed and we can turn to the interior nested ``sub-islands''. If there are any which are not simple links, then the second observation can be used again to draw the outermost link of these, and so on. When there are no nested (sub-)islands left but simple links, the word $w$ thus obtained is the desired one: It corresponds to $d$ under the identification \eqref{eqn:DefTLGens}.
\end{proof}
A word $w$ for the diagram $d$ exhibited in \eqref{eq:diagramD} may be easily found following this proof. Since the diagram $d$ has $p=4$ links on each side, the proof of \lemref{lem:surjectivity} begins with the word $u_1u_3u_5u_7$. Then, the highest through-line (at position $9$) is moved to the position of the highest through-line in $d'$ by left-multiplying by $u_2u_4u_6u_8$ and right-multiplying by $u_4u_6u_8$. Moving the second highest, and only other, through-line to the correct position in $d'$ then requires right-multiplication by $u_9$. This gives a word $w'$ that corresponds to $d'$:
\begin{gather}
u_1u_3u_5u_7 = \ 
\begin{tikzpicture}[baseline={(current bounding box.center)},scale=1/3]
	\draw[thick] (3,0) -- (3,11);
	\draw[thick] (0,0) -- (0,11);
	\draw (0,3) arc (-90:90:0.5);
	\draw (0,5) arc (-90:90:0.5);
	\draw (0,7) arc (-90:90:0.5);
	\draw (0,9) arc (-90:90:0.5);
	\draw (3,10) arc (90:270:0.5);
	\draw (3,8) arc (90:270:0.5);
	\draw (3,6) arc (90:270:0.5);
	\draw (3,4) arc (90:270:0.5);
	\draw (0,1) -- (3,1);
	\draw (0,2) -- (3,2);
\end{tikzpicture}
\quad \Rightarrow \quad
(u_2u_4u_6u_8)(u_1u_3u_5u_7)(u_4u_6u_8) = \ 
\begin{tikzpicture}[baseline={(current bounding box.center)},scale=1/3]
	\draw[thick] (3,0) -- (3,11);
	\draw[thick] (0,0) -- (0,11);
	\draw (0,2) arc (-90:90:0.5);
	\draw (0,4) arc (-90:90:0.5);
	\draw (0,6) arc (-90:90:0.5);
	\draw (0,8) arc (-90:90:0.5);
	\draw (3,9) arc (90:270:0.5);
	\draw (3,7) arc (90:270:0.5);
	\draw (3,5) arc (90:270:0.5);
	\draw (3,3) arc (90:270:0.5);
	\draw (0,1) -- (3,1);
	\draw (0,10) -- (3,10);
	\draw[thick] (6,0) -- (6,11);
	\draw (3,3) arc (-90:90:0.5);
	\draw (3,5) arc (-90:90:0.5);
	\draw (3,7) arc (-90:90:0.5);
	\draw (3,9) arc (-90:90:0.5);
	\draw (6,10) arc (90:270:0.5);
	\draw (6,8) arc (90:270:0.5);
	\draw (6,6) arc (90:270:0.5);
	\draw (6,4) arc (90:270:0.5);
	\draw (3,1) -- (6,1);
	\draw (3,2) -- (6,2);
	\draw[thick] (9,0) -- (9,11);
	\draw (6,2) arc (-90:90:0.5);
	\draw (6,4) arc (-90:90:0.5);
	\draw (6,6) arc (-90:90:0.5);
	\draw (9,7) arc (90:270:0.5);
	\draw (9,5) arc (90:270:0.5);
	\draw (9,3) arc (90:270:0.5);
	\draw (6,1) -- (9,1);
	\draw (6,8) -- (9,8);
	\draw (6,9) -- (9,9);
	\draw (6,10) -- (9,10);
\end{tikzpicture}
\ = \ 
\begin{tikzpicture}[baseline={(current bounding box.center)},scale=1/3]
	\draw[thick] (3,0) -- (3,11);
	\draw[thick] (0,0) -- (0,11);
	\draw (0,2) arc (-90:90:0.5);
	\draw (0,4) arc (-90:90:0.5);
	\draw (0,6) arc (-90:90:0.5);
	\draw (0,8) arc (-90:90:0.5);
	\draw[out=0,in=180] (0,10) to (3,8);
	\draw (3,10) arc (90:270:0.5);
	\draw (3,7) arc (90:270:0.5);
	\draw (3,5) arc (90:270:0.5);
	\draw (3,3) arc (90:270:0.5);
	\draw (0,1) -- (3,1);
\end{tikzpicture}
\notag \\
\Rightarrow \quad w'=(u_2u_4u_6u_8)(u_1u_3u_5u_7)(u_4u_6u_8)(u_9) = \ 
\begin{tikzpicture}[baseline={(current bounding box.center)},scale=1/3]
	\draw[thick] (3,0) -- (3,11);
	\draw[thick] (0,0) -- (0,11);
	\draw (0,2) arc (-90:90:0.5);
	\draw (0,4) arc (-90:90:0.5);
	\draw (0,6) arc (-90:90:0.5);
	\draw (0,8) arc (-90:90:0.5);
	\draw[out=0,in=180] (0,10) to (3,8);
	\draw (3,10) arc (90:270:0.5);
	\draw (3,7) arc (90:270:0.5);
	\draw (3,5) arc (90:270:0.5);
	\draw (3,3) arc (90:270:0.5);
	\draw (0,1) -- (3,1);
	\draw[thick] (6,0) -- (6,11);
	\draw (3,1) arc (-90:90:0.5);
	\draw (6,2) arc (90:270:0.5);
	\draw (3,10) -- (6,10);
	\draw (3,9) -- (6,9);
	\draw (3,8) -- (6,8);
	\draw (3,7) -- (6,7);
	\draw (3,6) -- (6,6);
	\draw (3,5) -- (6,5);
	\draw (3,4) -- (6,4);
	\draw (3,3) -- (6,3);
\end{tikzpicture}
\ = \ 
\begin{tikzpicture}[baseline={(current bounding box.center)},scale=1/3]
	\draw[thick] (3,0) -- (3,11);
	\draw[thick] (0,0) -- (0,11);
	\draw (0,2) arc (-90:90:0.5);
	\draw (0,4) arc (-90:90:0.5);
	\draw (0,6) arc (-90:90:0.5);
	\draw (0,8) arc (-90:90:0.5);
	\draw[out=0,in=180] (0,1) to (3,3);
	\draw[out=0,in=180] (0,10) to (3,8);
	\draw (3,10) arc (90:270:0.5);
	\draw (3,7) arc (90:270:0.5);
	\draw (3,5) arc (90:270:0.5);
	\draw (3,2) arc (90:270:0.5);
\end{tikzpicture}
\ =d'.
\end{gather}
The second part of the proof now proceeds as follows.  There are two nested islands in $d$ that are not simple links, one on the left from positions $2$ to $9$ and another on the right from positions $4$ to $7$. Their outermost links are first closed by multiplying $w'$ by $u_3u_5u_7$ on the left and by $u_5$ on the right. Now, the resulting diagram differs from $d$ only in that the simple links on the left at positions $5$ and $7$ are, in $d$, replaced by a nested sub-island from positions $5$ to $8$. Left-multiplication by $u_6$ then brings us, finally, to $d$:
\begin{gather}
w'= d' \qquad \Rightarrow \qquad (u_3u_5u_7)w'(u_5) = \ 
\begin{tikzpicture}[baseline={(current bounding box.center)},scale=1/3]
	\draw[thick] (3,0) -- (3,11);
	\draw[thick] (0,0) -- (0,11);
	\draw (0,3) arc (-90:90:0.5);
	\draw (0,5) arc (-90:90:0.5);
	\draw (0,7) arc (-90:90:0.5);
	\draw (3,8) arc (90:270:0.5);
	\draw (3,6) arc (90:270:0.5);
	\draw (3,4) arc (90:270:0.5);
	\draw (0,1) -- (3,1);
	\draw (0,2) -- (3,2);
	\draw (0,9) -- (3,9);
	\draw (0,10) -- (3,10);
	\draw[thick] (6,0) -- (6,11);
	\draw (3,2) arc (-90:90:0.5);
	\draw (3,4) arc (-90:90:0.5);
	\draw (3,6) arc (-90:90:0.5);
	\draw (3,8) arc (-90:90:0.5);
	\draw[out=0,in=180] (3,1) to (6,3);
	\draw[out=0,in=180] (3,10) to (6,8);
	\draw (6,10) arc (90:270:0.5);
	\draw (6,7) arc (90:270:0.5);
	\draw (6,5) arc (90:270:0.5);
	\draw (6,2) arc (90:270:0.5);
	\draw[thick] (9,0) -- (9,11);
	\draw (6,5) arc (-90:90:0.5);
	\draw (9,6) arc (90:270:0.5);
	\draw (6,1) -- (9,1);
	\draw (6,2) -- (9,2);
	\draw (6,3) -- (9,3);
	\draw (6,4) -- (9,4);
	\draw (6,7) -- (9,7);
	\draw (6,8) -- (9,8);
	\draw (6,9) -- (9,9);
	\draw (6,10) -- (9,10);
\end{tikzpicture}
\ = \ 
\begin{tikzpicture}[baseline={(current bounding box.center)},scale=1/3]
	\draw[thick] (3,0) -- (3,11);
	\draw[thick] (0,0) -- (0,11);
	\draw (0,3) arc (-90:90:0.5);
	\draw (0,5) arc (-90:90:0.5);
	\draw[rounded corners] (0,2) -- (0.5,2) -- (1.2,3) -- (1.2,8) -- (0.5,9) -- (0,9);
	\draw (0,7) arc (-90:90:0.5);
	\draw[out=0,in=180] (0,1) to (3,3);
	\draw[out=0,in=180] (0,10) to (3,8);
	\draw (3,10) arc (90:270:0.5);
	\draw[out=180,in=180] (3,4) to (3,7);
	\draw (3,6) arc (90:270:0.5);
	\draw (3,2) arc (90:270:0.5);
\end{tikzpicture}
\notag \\
\Rightarrow \qquad (u_6)(u_3u_5u_7)w'(u_5) = \ 
\begin{tikzpicture}[baseline={(current bounding box.center)},scale=1/3]
	\draw[thick] (3,0) -- (3,11);
	\draw[thick] (0,0) -- (0,11);
	\draw (0,4) arc (-90:90:0.5);
	\draw (3,5) arc (90:270:0.5);
	\draw (0,1) -- (3,1);
	\draw (0,2) -- (3,2);
	\draw (0,3) -- (3,3);
	\draw (0,6) -- (3,6);
	\draw (0,7) -- (3,7);
	\draw (0,8) -- (3,8);
	\draw (0,9) -- (3,9);
	\draw (0,10) -- (3,10);
	\draw[thick] (6,0) -- (6,11);
	\draw (3,3) arc (-90:90:0.5);
	\draw (3,5) arc (-90:90:0.5);
	\draw[rounded corners] (3,2) -- (3.5,2) -- (4.2,3) -- (4.2,8) -- (3.5,9) -- (3,9);
	\draw (3,7) arc (-90:90:0.5);
	\draw[out=0,in=180] (3,1) to (6,3);
	\draw[out=0,in=180] (3,10) to (6,8);
	\draw[out=180,in=180] (6,4) to (6,7);
	\draw (6,10) arc (90:270:0.5);
	\draw (6,6) arc (90:270:0.5);
	\draw (6,2) arc (90:270:0.5);
\end{tikzpicture}
\ = \ 
\begin{tikzpicture}[baseline={(current bounding box.center)},scale=1/3]
	\draw[thick] (3,0) -- (3,11);
	\draw[thick] (0,0) -- (0,11);
	\draw (0,4) arc (-90:90:0.5);
	\draw[rounded corners] (0,2) -- (0.5,2) -- (1.2,3) -- (1.2,8) -- (0.5,9) -- (0,9);
	\draw (0,7) arc (-90:90:0.5);
	\draw[out=0,in=180] (0,1) to (3,3);
	\draw[out=0,in=180] (0,10) to (3,8);
	\draw[out=0,in=0] (0,3) to (0,6);
	\draw (3,10) arc (90:270:0.5);
	\draw[out=180,in=180] (3,4) to (3,7);
	\draw (3,6) arc (90:270:0.5);
	\draw (3,2) arc (90:270:0.5);
\end{tikzpicture}
\ =d.
\end{gather}
The desired word is therefore
\begin{equation}
w=(u_6)(u_3u_5u_7)w'(u_5)=(u_6)(u_3u_5u_7)(u_2u_4u_6u_8)(u_1u_3u_5u_7)(u_4u_6u_8)(u_9)(u_5) = d.
\end{equation}
We remark that this word $w$ is clearly not in Jones' normal form.

%
%

\section{Standard Representations} \label{sec:Reps}

The $\brac{n,p}$-link states are not just convenient for combinatorial arguments.  They in fact admit a very natural action of the Temperley-Lieb algebra $\tl{n}$.  More precisely, if we let $\LinkMod_n$ denote the complex span of the $\brac{n,p}$-link states (over all $p$), then $\LinkMod_n$ is naturally a left $\tl{n}$-module under the concatenation of diagram with link state.  We will refer to this module as the \emph{link module}.  An example should serve to make the action clear:
\begin{equation}
\parbox{13mm}{\begin{center}
\includegraphics[height=28mm]{tl-diag-ex}
\end{center}}
\ \cdot \ 
\parbox{7mm}{\begin{center}
\includegraphics[height=28mm]{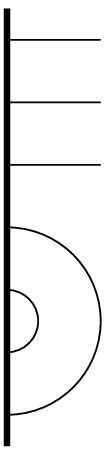}
\end{center}}
\ = \ 
\parbox{20mm}{\begin{center}
\includegraphics[height=28mm]{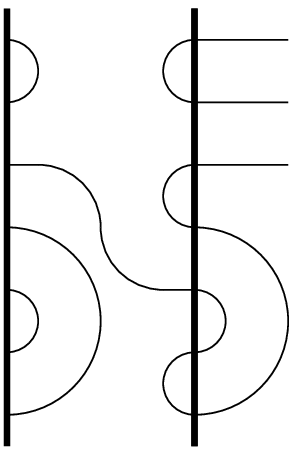}
\end{center}}
\ = \ 
\parbox{7mm}{\begin{center}
\includegraphics[height=28mm]{tl-link-ex}
\end{center}}
\ .
\end{equation}
\eqnref{eq:exampleBubble} provides an example in which a loop is closed upon concatenation.  We note that the number of defects need not be conserved under the $\tl{n}$-action.  However, this action can only close defects in pairs, so the number of defects never increases.  Equivalently, the number of links never decreases under the $\tl{n}$-action.

We can therefore identify $\tl{n}$-submodules $\LinkMod_{n,p} \subseteq \LinkMod_{n}$ which are spanned by the $\brac{n,p'}$-link states with $p' \geqslant p$.  This gives us an obvious filtration
\begin{equation} \label{eqn:Filtration}
0 \subset \LinkMod_{n , \flr{n/2}} \subset \cdots \subset \LinkMod_{n,1} \subset \LinkMod_{n,0} = \LinkMod_{n},
\end{equation}
whose consecutive quotients will be denoted by
\begin{equation} \label{eqn:DefNatMod}
\NatMod_{n,p} = \frac{\LinkMod_{n,p}}{\LinkMod_{n,p+1}}.
\end{equation}
The $\NatMod_{n,p}$ are therefore $\tl{n}$-modules in which the action is by concatenation when the number of links (and defects) is conserved and zero otherwise.  We shall refer to them as the \emph{standard modules}.  In the literature, they are also often called \emph{cell modules}, following the seminal work \cite{GraCel96} upon which a significant proportion of this section is based.  As vector spaces, they are spanned by the $\brac{n,p}$-link states (we shall often forget to distinguish an element of $\LinkMod_{n,p}$ from its coset in $\NatMod_{n,p}$ when it is not crucial), so their dimensions are the $d_{n,p}$ computed in \eqnref{eqn:DimV}.  For example, $\NatMod_{n,0}$ is always one-dimensional, spanned by the link state with no links (all defects).  The Temperley-Lieb generators all act trivially on $\NatMod_{n,0}$ (they would all close two defects), except the unit of course, which acts as the identity.

It is natural to ask at this point if the $\NatMod_{n,p}$ are irreducible as $\tl{n}$-modules.  We will see, eventually, that the answer is generically ``yes''.  The story is subtle however and will not be completed until the end of \secref{sec:Gram}.  We will approach this irreducibility in a manner familiar to physicists, by studying the non-degeneracy of an invariant bilinear form.  Let us therefore recall that we have oriented our link states so that the links and defects point to the right.  This facilitates a left action of the Temperley-Lieb algebra.  We could also introduce the reflections of the link states across a vertical line.  The links and defects would then point to the left, and we would have a natural right action of $\tl{n}$ on these reflected link states.  This suggests that we might be able to treat reflected link states as linear functionals acting on the link states, hence that combining a reflected link state with a standard link state may lead to an interesting bilinear pairing.

We therefore define, on each $\NatMod_{n,p}$, a form $\inner{\cdot}{\cdot} \equiv \inner{\cdot}{\cdot}_{n,p}$ as follows.  If $x$ and $y$ are two $\brac{n,p}$-link states, then $\inner{x}{y}$ is computed by reflecting $x$ across a vertical line and identifying its vertical border with that of $y$.  The value $\inner{x}{y}$ is then non-zero only if every defect of $x$ ends up being connected to one of $y$; when this is so, $\inner{x}{y}=\beta^m$ where $m$ is the number of closed loops so obtained.  This is then extended bilinearly\footnote{What follows can easily be adapted to the case where one extends sesquilinearly.  We have chosen bilinearity for simplicity.} to all of $\NatMod_{n,p}$.  As examples, we compute that in $\NatMod_{4,1}$,
\begin{equation}
\inner{\ 
\parbox{5mm}{\begin{center}
\includegraphics[height=16mm]{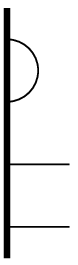}
\end{center}}
\ }{\ 
\parbox{5mm}{\begin{center}
\includegraphics[height=16mm]{tl-link-41a}
\end{center}}
\ } = \ 
\parbox{9mm}{\begin{center}
\includegraphics[height=16mm]{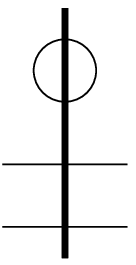}
\end{center}}
\ = \beta, \quad \inner{\ 
\parbox{5mm}{\begin{center}
\includegraphics[height=16mm]{tl-link-41a}
\end{center}}
\ }{\ 
\parbox{5mm}{\begin{center}
\includegraphics[height=16mm]{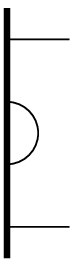}
\end{center}}
\ } = \ 
\parbox{9mm}{\begin{center}
\includegraphics[height=16mm]{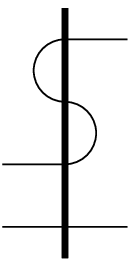}
\end{center}}
\ = 1 \quad \text{and} \quad \inner{\ 
\parbox{5mm}{\begin{center}
\includegraphics[height=16mm]{tl-link-41a}
\end{center}}
\ }{\ 
\parbox{5mm}{\begin{center}
\includegraphics[height=16mm]{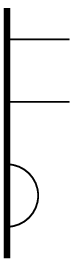}
\end{center}}
\ } = \ 
\parbox{9mm}{\begin{center}
\includegraphics[height=16mm]{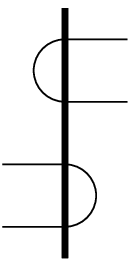}
\end{center}}
\ = 0.
\end{equation}
Note that this bilinear form is symmetric --- $\inner{x}{y}$ and $\inner{y}{x}$ are just reflections of one another in the pictorial interpretation.

Consider now the reflection of an $n$-diagram about a vertical line.  This will be another $n$-diagram, so reflection defines a linear map from $\tl{n}$ to itself.  This obviously preserves the multiplication except that the order will be reversed --- this map is an antiautomorphism of $\tl{n}$.\footnote{We will take this adjoint to be linear as we are considering $\inner{\cdot}{\cdot}$ as a bilinear form.  If one prefers a sesquilinear form, then the adjoint should be taken to be antilinear.  Note however, that this conflicts with the defining relations \eqref{eqn:TLrels} unless $\beta$ is real.}  We will therefore regard it as an adjoint, denoting it by $U \mapsto U^{\dag}$.  It is clear from \eqnref{eqn:DefTLGens} that $\wun^{\dag} = \wun$ and $u_i^{\dag} = u_i$.  This definition is natural and useful.  Indeed, it shows that the bilinear forms we have defined turn out to be \emph{invariant} with respect to the $\tl{n}$-action.
\begin{lemma} \label{lem:Invariance}
The bilinear form $\inner{\cdot}{\cdot}$ on $\NatMod_{n,p}$ satisfies
\begin{equation} \label{eqn:DefAdj}
\inner{x}{U y} = \inner{U^{\dag} x}{y} \qquad \text{for all $U \in \tl{n}$ and $x , y \in \NatMod_{n,p}$.}
\end{equation}
\end{lemma}
\noindent To see this, we merely note that the two sides of \eqref{eqn:DefAdj} are identical when expressed in terms of diagrams and link states.

We introduce a convenient notation $\wall{x}{y}$ for the unique $n$-diagram which, when cut in half vertically, decomposes into the $\brac{n,p}$-link state $x$ and the vertical reflection of the $\brac{n,p}$-link state $y$.  Extending linearly, we obtain a map $\wall{\cdot}{\cdot}$ from $\NatMod_{n,p} \times \NatMod_{n,p}$ into $\tl{n}$ (in fact we obtain such a map for each $p$, but we will not bother to distinguish them).  The utility of the bilinear form we have defined stems from the following relation.
\begin{lemma} \label{lem:Magic}
If $x,y,z \in \NatMod_{n,p}$, then
\begin{equation} \label{eqn:Magic}
\wall{x}{y} z = \inner{y}{z} x.
\end{equation}
\end{lemma}
\begin{proof}
Observe first that linearity allows us to assume that $x$, $y$ and $z$ are in fact $\brac{n,p}$-link states.  If any of the defects of $y$ are closed by a link in $z$, then $\inner{y}{z} = 0$ by definition, so the right-hand side of \eqref{eqn:Magic} vanishes.  As the defects of $x$ join those of $y$, this will close two defects of $x$, leading to an additional link in $\wall{x}{y} z$.  But then, the resulting link state will vanish in $\NatMod_{n,p}$ due to its definition as a quotient, hence the left-hand side of \eqref{eqn:Magic} likewise vanishes.

It remains to check the case in which none of the defects of $y$ are closed by a link in $z$.  But then it is clear that the left-hand side of \eqref{eqn:Magic} will be proportional to $x$.  The proportionality constant is then given by $\beta$ to the power of the number of loops in the concatenation of $\wall{x}{y}$ and $z$.  Since $x$ contributes no loops, this constant is just $\inner{y}{z}$ as required.
\end{proof}

We illustrate this deceptively simple result with an example in $\NatMod_{7,3}$:  Compare
\begin{equation}\label{eq:exampleBubble}
\parbox{13mm}{\begin{center}
\includegraphics[height=28mm]{tl-diag-ex}
\end{center}}
\ \cdot \ 
\parbox{7mm}{\begin{center}
\includegraphics[height=28mm]{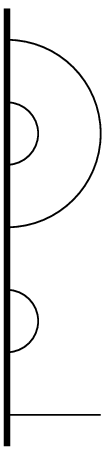}
\end{center}}
\ = \beta \ 
\parbox{7mm}{\begin{center}
\includegraphics[height=28mm]{tl-link-ex}
\end{center}}
\qquad \text{with} \qquad
\inner{\ 
\parbox{5mm}{\begin{center}
\includegraphics[height=28mm]{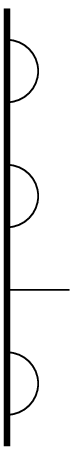}
\end{center}}
\ }{\ 
\parbox{7mm}{\begin{center}
\includegraphics[height=28mm]{tl-link-ex3}
\end{center}}
\ } = \beta.
\end{equation}
We remark that it is tempting (and sometimes useful) to think of $\wall{x}{y}$ as the combination $x y^{\dag}$, where the ``adjoint'' of the link state $y$ refers likewise to its vertical reflection.  Then, one could ``prove'' \lemref{lem:Magic} as follows:
\begin{equation} \label{eqn:Magic2}
\wall{x}{y} z = \brac{x y^{\dag}} z = x \brac{y^{\dag} z} = \inner{y}{z} x.
\end{equation}
However, caution should be exercised with such manipulations.  If $x$ and $y$ belong to $\NatMod_{n,p'}$ and $z$ is an $\brac{n,p}$-link state (with $p' \neq p$), then \eqref{eqn:Magic} and \eqref{eqn:Magic2} do not make sense --- the bilinear form is no longer defined.  Nevertheless, the product $\wall{x}{y} z$ may be non-zero (in $\LinkMod_{n}$) and not even proportional to $x$.  For example,
\begin{equation}
\parbox{8mm}{\begin{center}
\includegraphics[height=16mm]{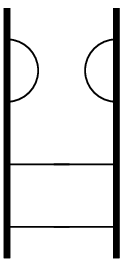}
\end{center}}
\ \cdot \ 
\parbox{7mm}{\begin{center}
\includegraphics[height=16mm]{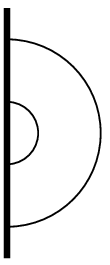}
\end{center}}
\ = \ 
\parbox{3mm}{\begin{center}
\includegraphics[height=16mm]{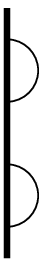}
\end{center}}
\ .
\end{equation}
We will therefore refrain from trying to define an adjoint of a link state (what one means is of course a linear functional in the dual vector space).

We can now understand how this bilinear form will help to decide the irreducibility of the $\NatMod_{n,p}$.  For this, we consider the \emph{radical} $\RadMod_{n,p}$ of the bilinear form on $\NatMod_{n,p}$:
\begin{equation}\label{eq:DefRadMod}
\RadMod_{n,p} = \set{x \in \NatMod_{n,p} \st \inner{x}{y} = 0 \text{ for all } y \in \NatMod_{n,p}}.
\end{equation}
The invariance (\lemref{lem:Invariance}) of the bilinear form implies that the radical $\RadMod_{n,p}$ is a submodule of $\NatMod_{n,p}$.
\begin{proposition} \label{prop:S/R}
If $\inner{\cdot}{\cdot}$ is not identically zero on $\NatMod_{n,p}$, then $\NatMod_{n,p}$ is cyclic and indecomposable.  Moreover, $\NatMod_{n,p} / \RadMod_{n,p}$ is then irreducible.  Equivalently, $\RadMod_{n,p}$ is the unique maximal proper submodule of $\NatMod_{n,p}$.
\end{proposition}
\noindent (We shall discuss the case when $\inner{\cdot}{\cdot}$ is identically zero shortly.)
\begin{proof}
Since $\inner{\cdot}{\cdot}$ is not identically zero, there exist $y,z \in \NatMod_{n,p}$ such that $\inner{y}{z} = 1$.  Then, \eqnref{eqn:Magic} says that for every $x \in \NatMod_{n,p}$, we may form $\wall{x}{y} \in \tl{n}$ to obtain
\begin{equation}
\wall{x}{y} z = \inner{y}{z} x = x.
\end{equation}
Thus, $z$ generates $\NatMod_{n,p}$, proving that this module is cyclic.  Now, any $z \notin \RadMod_{n,p}$ has such a partner $y$, so every such $z$ is a generator of $\NatMod_{n,p}$.  It follows that every non-zero element of $\NatMod_{n,p} / \RadMod_{n,p}$ generates this quotient, hence that the quotient is irreducible.

Suppose then that $\NatMod_{n,p}$ can be written as $\NatMod_{n,p} = A \oplus B$. Let $z$ be a generator of $\NatMod_{n,p}$ and $z_A\in A, z_B\in B$ be such that $z=z_A+z_B$. If both $z_A$ and $z_B$ were in $\RadMod_{n,p}$, then they would generate a submodule $\tl n z_A\oplus \tl n z_B\subseteq \RadMod_{n,p}$ that includes $\tl n z$ and is distinct from $\NatMod_{n,p}$. So, at least one of $z_A$ or $z_B$ is not an element of $\RadMod_{n,p}$. If it is $z_A$, then $z_A$ generates $\NatMod_{n,p}$, so $\NatMod_{n,p} = \tl{n} z_A \subseteq A$, giving $B=0$. If $z_B \notin \RadMod_{n,p}$, then $A=0$ by the same argument. This proves that $\NatMod_{n,p}$ is indecomposable.
\end{proof}

\noindent We will find it convenient to denote the quotient $\NatMod_{n,p} / \RadMod_{n,p}$ by $\IrrMod_{n,p}$, even when $\inner{\cdot}{\cdot}$ vanishes identically.  In this latter case, $\IrrMod_{n,p}$ is the trivial module $\set{0}$; otherwise, $\IrrMod_{n,p}$ is irreducible. For this reason, we shall often refer to the $\IrrMod_{n,p}$ as the irreducibles, understanding that one should exclude any $\IrrMod_{n,p}$ that vanish.

To prove that $\NatMod_{n,p}$ is irreducible, it is therefore enough to show that the radical $\RadMod_{n,p}$ is zero.  Equivalently, we must show that the bilinear form we have defined on $\NatMod_{n,p}$ is non-degenerate.  This will be the strategy of \secref{sec:Gram}.  Note however that $\RadMod_{n,p} \neq \set{0}$ only implies that $\NatMod_{n,p}$ is reducible when $\inner{\cdot}{\cdot} \neq 0$ (that is, when $\RadMod_{n,p} \neq \NatMod_{n,p}$).  In case $\inner{\cdot}{\cdot}=0$, we cannot decide the irreducibility of $\NatMod_{n,p}$ using \propref{prop:S/R}.

However, when $\beta \neq 0$, $\inner{x}{x} = \beta^p \neq 0$ for all $\brac{n,p}$-link states $x$, so the bilinear form is non-zero.  Indeed, if $\beta = 0$ but $n \neq 2p$ (so there is at least one defect), then the form is non-zero because we may choose $x$ and $y$ so as to form a ``snake'':
\begin{equation} \label{eqn:Snake}
\inner{x}{y} = \ 
\parbox{13mm}{\begin{center}
\includegraphics[height=48mm]{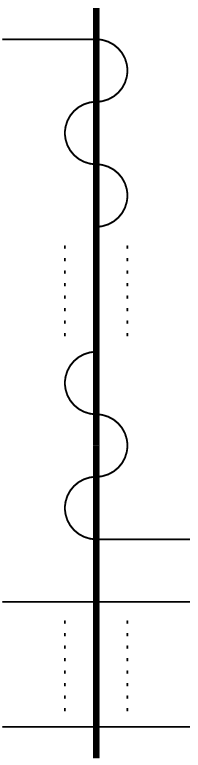}
\end{center}}
\ = 1.
\end{equation}
However, when $\beta = 0$ and there are no defects ($n = 2p$), the bilinear form vanishes identically.  It is therefore useful to renormalise it, defining
\begin{equation} \label{eqn:RenormalisedForm}
\inner{x}{y}' = \lim_{\beta \rightarrow 0} \frac{\inner{x}{y}}{\beta} \qquad \text{($x,y \in \NatMod_{2p,p}$),}
\end{equation}
where the forms appearing within the limit are those of the $\NatMod_{2p,p}$ with $\beta \neq 0$.  Since $\inner{x}{y}$ is a polynomial in $\beta$ with vanishing constant coefficient, $\inner{x}{y}'$ is defined --- in fact, when $x$ and $y$ are $\brac{2p,p}$-link states, it is $1$ when a single loop is formed and $0$ otherwise.  Moreover, $\inner{\cdot}{\cdot}'$ inherits bilinearity, symmetry and invariance from $\inner{\cdot}{\cdot}$.  Note that this renormalised bilinear form is \emph{not} identically zero because $x$ and $y$ may be chosen so as to obtain
\begin{equation}
\inner{x}{y}' = \lim_{\beta \rightarrow 0} \frac{1}{\beta} \ 
\parbox{9mm}{\begin{center}
\includegraphics[height=36mm]{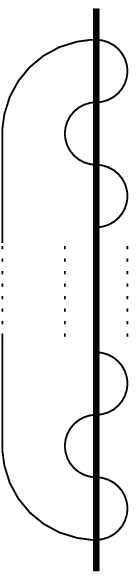}
\end{center}}
\ = \lim_{\beta \rightarrow 0} \frac{\beta}{\beta}= 1.
\end{equation}

We cannot immediately apply \propref{prop:S/R} to the radical of this renormalised form because the proof relied crucially upon \lemref{lem:Magic}, special to the form $\inner{\cdot}{\cdot}$.  However, we have the following replacement:
\begin{lemma} \label{lem:Magic2}
Let $x$ and $y$ be $\brac{2p,p-1}$-link states, so that they have precisely two defects.  Denote by $x'$ and $y'$ the respective $\brac{2p,p}$-link states formed by closing these defects.  Then for $\beta = 0$,
\begin{equation}
\wall{x}{y} z = \inner{y'}{z}' x', \qquad \text{for all $z \in \NatMod_{2p,p}$.}
\end{equation}
This extends linearly to all $x,y \in \NatMod_{2p,p-1}$.
\end{lemma}
\begin{proof}
Note first that the defects of $y$ must be closed by a link in $z$ (assuming by linearity that the latter is a $\brac{2p,p}$-link state).  Thus, $\wall{x}{y} z$ will be proportional to the link state $x'$ obtained by closing the defects of $x$.  The constant of proportionality is given by $\beta$ to the power of the number of loops formed by $y$ and $z$.  This is $0$ (since $\beta = 0$) unless there are no such loops, in which case it is $1$.  However, this matches the value of $\inner{y'}{z}'$, as closing the defects of $y$ leads to an additional loop which is dealt with by dividing by $\beta$ in \eqref{eqn:RenormalisedForm}.
\end{proof}

If $\RadMod'_{2p,p}$ denotes the radical of the renormalised $\beta = 0$ form $\inner{\cdot}{\cdot}'$, we can now mimic the proof of \propref{prop:S/R} to obtain the analogous result.
\begin{proposition} \label{prop:S/R'}
The form $\inner{\cdot}{\cdot}$ is identically zero if and only if $\beta = 0$ and $n=2p$.  Then, $\NatMod_{2p,p}$ is cyclic, indecomposable and has an irreducible quotient $\NatMod_{2p,p} / \RadMod'_{2p,p}$.
\end{proposition}
\noindent It remains only to remark that this quotient is never zero because we have shown that $\RadMod'_{2p,p} \neq \NatMod_{2p,p}$ (the form $\inner{\cdot}{\cdot}'$ is never identically zero).  

We conclude this section with a study of whether the $\NatMod_{n,p}$, and their irreducible quotients $\IrrMod_{n,p} = \NatMod_{n,p} / \RadMod_{n,p}$ are all mutually distinct as $\tl{n}$-modules (up to isomorphism of course).  The fact that $\NatMod_{n,p}$ and $\NatMod_{n,p'}$ involve different numbers of links, for $p \neq p'$, can be misleading.  For example, $\NatMod_{2,1}$ and $\NatMod_{2,0}$ are isomorphic one-dimensional $\tl{2}$-modules, when $\beta = 0$ ($\wun$ and $u_1$ are represented by $1$ and $0$, respectively, on both).  Nevertheless, this behaviour is rather exceptional.
\begin{proposition} \label{prop:HomLinks}
Let $\mathcal{N}$ and $\mathcal{N}'$ be submodules of $\NatMod_{n,p}$ and $\NatMod_{n,p'}$, respectively, where $p > p'$ and $\inner{\cdot}{\cdot}_{n,p} \neq 0$.  Then, the only module homomorphism $\theta \colon \NatMod_{n,p} / \mathcal{N} \rightarrow \NatMod_{n,p'} / \mathcal{N}'$ is the zero homomorphism.
\end{proposition}
\begin{proof}
Let $\gamma$ be the canonical homomorphism from $\NatMod_{n,p}$ onto $\NatMod_{n,p} / \mathcal{N}$.  Choose $y,z \in \NatMod_{n,p}$ such that $\inner{y}{z} = 1$.  Then, for all $x \in \NatMod_{n,p}$,
\begin{equation}
\wall{x}{y} \func{\theta}{\func{\gamma}{z}} = \func{\theta}{\func{\gamma}{\wall{x}{y} z}} = \func{\theta}{\func{\gamma}{x}}.
\end{equation}
But when $p > p'$, $\wall{x}{y} \func{\theta}{\func{\gamma}{z}} = 0$, since $\func{\theta}{\func{\gamma}{z}}$ has $p'$ links but left-multiplying by $\wall{x}{y}$ leads to at least $p$ links.  Thus, $\func{\theta}{\func{\gamma}{x}} = 0$, so $\theta = 0$ as $\gamma$ is surjective.
\end{proof}
\noindent Putting $\mathcal{N} = \mathcal{N}' = \set{0}$ or $\mathcal{N} = \RadMod_{n,p}$ and $\mathcal{N}' = \RadMod_{n,p'}$ in \propref{prop:HomLinks}, and using the basic fact that an isomorphism has an inverse, we obtain:
\begin{corollary} \label{cor:VLDistinct}
When $\inner{\cdot}{\cdot}_{n,p}$ and $\inner{\cdot}{\cdot}_{n,p'}$ are non-zero,
\begin{equation*}
\NatMod_{n,p} \cong \NatMod_{n,p'} \quad \Rightarrow \quad p = p' \qquad \text{and} \qquad \IrrMod_{n,p} \cong \IrrMod_{n,p'} \quad \Rightarrow \quad p = p'.
\end{equation*}
\end{corollary}
\noindent As we have seen, $\inner{\cdot}{\cdot}_{2p,p}$ vanishes identically when $\beta = 0$, so it follows that $\NatMod_{2p,p}$ ($\IrrMod_{2p,p}$) could coincide with one of the other $\NatMod_{2p,p'}$ ($\IrrMod_{2p,p'}$) in this case.  This is what allows the ($\beta = 0$) isomorphism $\NatMod_{2,1} \cong \NatMod_{2,0}$ which we remarked upon above.

We record a related result for future reference:
\begin{proposition} \label{prop:EndLinks}
Every module homomorphism $\theta \colon \NatMod_{n,p} \rightarrow \NatMod_{n,p}$ is a multiple of the identity.
\end{proposition}
\begin{proof}
When $\inner{\cdot}{\cdot}_{n,p} \neq 0$, this follows readily by choosing $y,z \in \NatMod_{n,p}$ such that $\inner{y}{z} = 1$.  Then, for all $x \in \NatMod_{n,p}$,
\begin{equation}
\func{\theta}{x} = \func{\theta}{\wall{x}{y} z} = \wall{x}{y} \func{\theta}{z} = \inner{y}{\func{\theta}{z}} x.
\end{equation}
For the remaining case, when $\inner{\cdot}{\cdot}_{n,p} = 0$, we must have $n=2p$ and $\beta = 0$.  However, the renormalised form $\inner{\cdot}{\cdot}'$ of \eqnref{eqn:RenormalisedForm} does not vanish identically, so there exist $y',z \in \NatMod_{2p,p}$ such that $\inner{y'}{z}' = 1$.  For every $x' \in \NatMod_{2p,p}$, let $x \in \NatMod_{2p,p-1}$ be obtained by cutting an outermost link of $x'$ (say the one closing at $n = 2p$ for definiteness).  Form $y \in \NatMod_{2p,p-1}$ from $y'$ in the same fashion.  Then, \lemref{lem:Magic2} gives
\begin{equation}
\func{\theta}{x'} = \func{\theta}{\wall{x}{y} z} = \wall{x}{y} \func{\theta}{z} = \inner{y'}{\func{\theta}{z}}' x',
\end{equation}
completing the proof.
\end{proof}

%
%

\section{Gram Matrices} \label{sec:Gram}

We can now turn to a study of the irreducibility of the $\NatMod_{n,p}$, based on the non-degeneracy of the invariant bilinear forms $\inner{\cdot}{\cdot}_{n,p}$.  Recall that each $\NatMod_{n,p}$ has a canonical basis given by the $\brac{n,p}$-link states.  With respect to this basis, the corresponding form is represented by a symmetric $d_{n,p} \times d_{n,p}$ matrix which we shall denote by $G_{n,p}$.  Such matrices are called \emph{Gram matrices}.  For example,
\begin{align}
&G_{4,0} = 
\bigl( 1 \bigr), & &G_{4,1} = 
\begin{pmatrix}
\beta & 1 & 0 \\
1 & \beta & 1 \\
0 & 1 & \beta
\end{pmatrix}
& &\text{and} & &G_{4,2} = 
\begin{pmatrix}
\beta^2 & \beta \\
\beta & \beta^2
\end{pmatrix}
,
\intertext{when we adopt the (respective) ordered bases}
&\set{\ 
\parbox{5mm}{\begin{center}
\includegraphics[height=16mm]{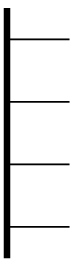}
\end{center}}
\ }, & &\set{\ 
\parbox{5mm}{\begin{center}
\includegraphics[height=16mm]{tl-link-41a}
\end{center}}
\ , \ 
\parbox{5mm}{\begin{center}
\includegraphics[height=16mm]{tl-link-41b}
\end{center}}
\ , \ 
\parbox{5mm}{\begin{center}
\includegraphics[height=16mm]{tl-link-41c}
\end{center}}
\ } & &\text{and} & &\set{\ 
\parbox{3mm}{\begin{center}
\includegraphics[height=16mm]{tl-link-42b}
\end{center}}
\ , \ 
\parbox{7mm}{\begin{center}
\includegraphics[height=16mm]{tl-link-42a}
\end{center}}
\ }. \notag
\end{align}
As the radical $\RadMod_{n,p}$ is represented by the kernel of the Gram matrix $G_{n,p}$, we see that the irreducibility of the $\NatMod_{n,p}$ is equivalent to $\func{\det}{G_{n,p}} \neq 0$.\footnote{We shall defer addressing the exceptional case of $\NatMod_{2p,p}$ with $\beta = 0$ until the end of the section.}  Our aim in this section is to compute $\func{\det}{G_{n,p}}$ explicitly.  The strategy is to use restriction to derive a recursion relation for $\func{\det}{G_{n,p}}$ when the $\RadMod_{n,p}$ all vanish.  This turns out to occur for generic $\beta \in \CC$, excluding only a countable set.  Continuity therefore takes care of the outstanding cases.

\begin{proposition} \label{prop:Restrict}
Consider the inclusion of $\tl{n-1}$ in $\tl{n}$ (for fixed $\beta$) given by sending the unit to the unit and the $u_i$ with $i < n-1$ to their counterparts in $\tl{n}$.  Denote the corresponding restriction of $\NatMod_{n,p}$ to a $\tl{n-1}$-module by $\Res{\NatMod_{n,p}}$.  Then, we have an exact sequence of $\tl{n-1}$-modules,
\begin{equation} \label{SES:Res}
\dses{\NatMod_{n-1,p}}{}{\Res{\NatMod_{n,p}}}{}{\NatMod_{n-1,p-1}},
\end{equation}
meaning that $\NatMod_{n-1,p}$ is a submodule of $\Res{\NatMod_{n,p}}$ and $\Res{\NatMod_{n,p}} / \NatMod_{n-1,p} \cong \NatMod_{n-1,p-1}$.
\end{proposition}
\begin{proof}
The inclusion $\NatMod_{n-1,p} \hookrightarrow \Res{\NatMod_{n,p}}$ is defined to extend an $\brac{n-1,p}$-link state to an $\brac{n,p}$-link state by adding a defect at position $n$.  This is clearly an injective homomorphism of $\tl{n-1}$-modules as the inclusion of $\tl{n-1}$ in $\tl{n}$ will preserve the defect at $n$.  The quotient $\Res{\NatMod_{n,p}} / \NatMod_{n-1,p}$ is then a $\tl{n-1}$-module with a basis of cosets which are represented by the $\brac{n,p}$-link states in which $n$ is part of a link.

There is an obvious vector space isomorphism $\Psi$ from $\Res{\NatMod_{n,p}} / \NatMod_{n-1,p}$ to $\NatMod_{n-1,p-1}$ obtained by cutting the link which closes at $n$ and then removing the newly-created defect at $n$.  We wish to show that $\Psi$ is in fact an isomorphism of $\tl{n-1}$-modules, thereby completing the proof.  So for a given basis element $z$ of $\Res{\NatMod_{n,p}} / \NatMod_{n-1,p}$, let $m$ denote the opening point of the link which closes at $n$.  Applying any $u_i$ ($i < n-1$) to $z$ usually then gives another such basis element in which $n$ is linked to some $m'$.  It is easy to see that $\func{\Psi}{u_i z} = u_i \func{\Psi}{z}$ in this case.  The only exception occurs if $i = m-1$ and $m-1$ is a defect in $z$.  Then, applying $u_{m-1}$ to $z$ leads to $n$ becoming a defect:
\begin{equation}
\parbox{19mm}{\begin{center}
\includegraphics[height=26mm]{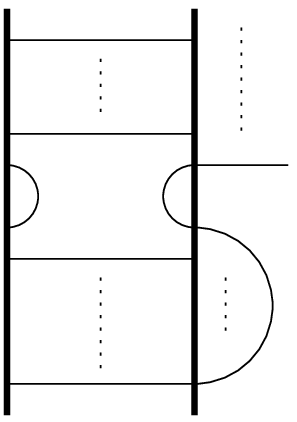}
\end{center}}
\ = \ 
\parbox{5mm}{\begin{center}
\includegraphics[height=26mm]{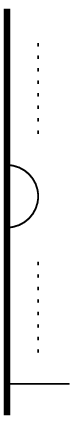}
\end{center}}
\ .
\end{equation}
In this case, $\func{\Psi}{u_{m-1} z} = \func{\Psi}{0} = 0$.  But, $\func{\Psi}{z}$ will have defects at both $m-1$ and $m$, so applying $u_{m-1}$ will close the defect leading to an extra link.  Thus, $u_{m-1} \func{\Psi}{z} = 0$ too and $\Psi$ is a homomorphism.
\end{proof}
\begin{corollary} \label{cor:ResSplit}
When $q^{2 \brac{n-2p+1}} \neq 1$ (recall that we write $\beta = q + q^{-1}$), the exact sequence \eqref{SES:Res} splits, so we have
\begin{equation}
\Res{\NatMod_{n,p}} \cong \NatMod_{n-1,p} \oplus \NatMod_{n-1,p-1} \qquad \text{(as $\tl{n-1}$-modules).}
\end{equation}
\end{corollary}
\begin{proof}
This is easy to see by using the eigenvalues of the central element $F_{n-1}$ on $\tl{n-1}$ (see \appref{app:Casimir}).  Since $F_{n-1}$ is central, left-multiplication by $F_{n-1} - \lambda \wun$ is a homomorphism from any module $\mathcal{N}$ to itself, for any $\lambda \in \CC$.  Thus, the (generalised) eigenspaces of $F_{n-1}$ on $\mathcal{N}$ are submodules.  Applying this to $\mathcal{N} = \Res{\NatMod_{n,p}}$ and recalling that $\NatMod_{n-1,p}$ and $\NatMod_{n-1,p-1}$ are indecomposable, we see that $F_{n-1}$ has at most two eigenvalues: $f_{n-1,p}$ with eigenspace $\NatMod_{n-1,p}$ and $f_{n-1,p-1}$ with eigenspace (isomorphic to) $\NatMod_{n-1,p-1}$.  If these eigenvalues are distinct, both eigenspaces are submodules and \eqref{SES:Res} splits.  It remains to determine when $f_{n-1,p} = f_{n-1,p-1}$.  But, \propref{lem:FEigs} gives
\begin{equation}
f_{n-1,p-1}-f_{n-1,p} = \brac{q-q^{-1}} \brac{q^{n-2p+1} - q^{-\brac{n-2p+1}}}.
\end{equation}
If either factor is zero, then $q^{2 \brac{n-2p+1}}$ will be $1$, contradicting the hypothesis.
\end{proof}

We introduce a useful definition:  Say that the pair $\brac{n,p}$ is \emph{critical} for a given $q \in \CC^{\times}$ if $q^{2 \brac{n-2p+1}} = 1$.  Similarly, any quantity indexed by $n$ and $p$ will be said to be critical if $\brac{n,p}$ is critical.  A rephrasing of \corref{cor:ResSplit} is therefore that a non-critical restricted module $\Res{\NatMod_{n,p}}$ always splits as $\NatMod_{n-1,p} \oplus \NatMod_{n-1,p-1}$.

Suppose now that $\Res{\NatMod_{n,p}}$ is non-critical so that there exists a splitting $\psi \colon \Res{\NatMod_{n,p}} \longrightarrow \NatMod_{n-1,p} \oplus \NatMod_{n-1,p-1}$.  Then, if we order the canonical basis of link states of $\NatMod_{n,p}$ so that those of $\NatMod_{n-1,p}$ (which have $n$ as a defect) come first, $\psi$ may be chosen so that it is represented by a matrix of the form
\begin{equation} \label{eqn:UDec}
U_{n,p} = 
\begin{pmatrix}
\id & V_{n,p} \\
0 & \id
\end{pmatrix}
.
\end{equation}
Here, the submatrix $V_{n,p}$ encodes the non-trivial part of the splitting (the embedding of $\NatMod_{n-1,p-1}$ in $\Res{\NatMod_{n,p}}$).

At this point, we make an inductive assumption:  We suppose that when $q$ is not a root of unity (so \eqref{SES:Res} always splits), the $\RadMod_{n',p'}$ \emph{vanish} for all $n' < n$ and all $p'$ --- this is certainly true for $n' \leqslant 2$.  With this assumption, the $\NatMod_{n',p'}$ are all irreducible and mutually distinct by \propref{prop:S/R} and \corref{cor:VLDistinct}.
\begin{lemma} \label{lem:Schur}
If the splitting $\psi$ exists, we may define a bilinear form on $\NatMod_{n-1,p} \oplus \NatMod_{n-1,p-1}$ by
\begin{equation}
\left\langle \mspace{-5mu} \inner{x+x'}{y+y'} \mspace{-5mu} \right\rangle = \inner{\func{\psi^{-1}}{x+x'}}{\func{\psi^{-1}}{y+y'}}_{n,p},\qquad \textrm{for\ }x,y\in\NatMod_{n-1,p}\textrm{\ and\ } x',y'\in\NatMod_{n-1,p-1}.
\end{equation}
This form is symmetric and invariant, hence
\begin{equation}
\left\langle \mspace{-5mu} \inner{x + x'}{y + y'} \mspace{-5mu} \right\rangle = \inner{x}{y}_{n-1,p} + \alpha_{n,p} \inner{x'}{y'}_{n-1,p-1},
\end{equation}
for some $\alpha_{n,p} \in \CC$.
\end{lemma}
\begin{proof}
This is a well-known argument:
A bilinear form induces a map from a module $V$ to its dual by $v \mapsto \inner{v}{\cdot}$.  The invariance of the form translates into this induced map being an intertwiner.  When $V$ is irreducible, so is its dual, so Schur's lemma tells us that the induced maps form a one-dimensional vector space, hence so too do the invariant bilinear forms.  In the application at hand, $V$ is the direct sum of two non-isomorphic irreducible modules, so there is a two-dimensional space of bilinear forms.  Comparing the form on $\NatMod_{n-1,p}$ with that on the direct sum fixes one of the latter's degrees of freedom to be unity.
\end{proof}
\noindent In matrix form, this becomes
\begin{gather}
G_{n-1,p} \oplus \alpha_{n,p} G_{n-1,p-1} = \brac{U_{n,p}^{-1}}^T G_{n,p} U_{n,p}^{-1} \notag \\
\Rightarrow \qquad G_{n,p} = U_{n,p}^T 
\begin{pmatrix}
G_{n-1,p} & 0 \\
0 & \alpha_{n,p} G_{n-1,p-1}
\end{pmatrix}
U_{n,p}. \label{eqn:Master}
\end{gather}
This is the recurrence relation which we shall use to compute $\func{\det}{G_{n,p}}$.  For this, it is useful to introduce the familiar notation $\qnum{m}$ for the $q$-number
\begin{equation}
\qnum{m} = \frac{q^m - q^{-m}}{q - q^{-1}},
\end{equation}
with the limiting case for $q=\pm 1$ being $[m]_q=mq^{m-1}$. Note that $\brac{n,p}$ is critical if and only if $\qnum{n-2p+1} = 0$ or $q = \pm 1$. 

\begin{proposition} \label{prop:Alpha}
When $\qnum{n-2p+1} \neq 0$ and $p>0$, $\alpha_{n,p}$ is finite and is given by
\begin{equation} \label{eqn:Alpha}
\alpha_{n,p} = \frac{\qnum{n-2p+2}}{\qnum{n-2p+1}}.
\end{equation}
\end{proposition}
\begin{proof}
We will prove this first under the assumption that $q$ is not a root of unity.  The general case then follows from continuity.

We begin by writing the Gram matrix in block form.  We again choose to order the canonical basis of link states so that those with $n$ a defect come first.  Thus,
\begin{equation} \label{eqn:GDec}
G_{n,p} = 
\begin{pmatrix}
G_{n,p}^{1,1} & G_{n,p}^{1,2} \\
G_{n,p}^{2,1} & G_{n,p}^{2,2}
\end{pmatrix}
,
\end{equation}
where $G_{n,p}^{1,1}$ is in fact just $G_{n-1,p}$ (removing the defect at $n$ has no effect on the values taken by the bilinear form).  Substituting \eqnDref{eqn:UDec}{eqn:GDec} into \eqnref{eqn:Master} now gives two independent (non-trivial) constraints:
\begin{subequations}
\begin{align}
G_{n,p}^{1,2} &= G_{n-1,p} V_{n,p} \label{eqn:G12} \\
\text{and} \qquad G_{n,p}^{2,2} &= \brac{V_{n,p}}^T G_{n-1,p} V_{n,p} + \alpha_{n,p} G_{n-1,p-1} \label{eqn:G22'} \\
&= \brac{V_{n,p}}^T G_{n,p}^{1,2} + \alpha_{n,p} G_{n-1,p-1}, \label{eqn:G22}
\end{align}
\end{subequations}
where the last equality uses \eqref{eqn:G12}.

Just as we chose the basis of $\NatMod_{n,p}$ so that the link states with $n$ a defect came first, we now refine it so that when $n$ is a defect, the link states with $n-1$ a defect come before those with $n-1$ part of a link.  Similarly, when $n$ is part of a link, put those link states with $n-1$ linked to $n$ before the rest.  Pictorially, the order is:
\begin{equation*}
\parbox{5mm}{\begin{center}
\includegraphics[height=12mm]{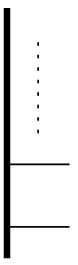}
\end{center}}
\quad , \quad
\parbox{6mm}{\begin{center}
\includegraphics[height=12mm]{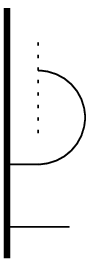}
\end{center}}
\quad , \quad 
\parbox{5mm}{\begin{center}
\includegraphics[height=12mm]{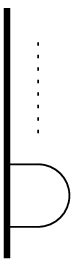}
\end{center}}
\quad , \quad 
\parbox{8mm}{\begin{center}
\includegraphics[height=12mm]{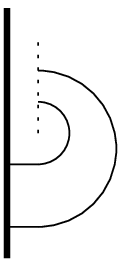}
\end{center}}
\quad .
\end{equation*}
This results in $G_{n,p}^{1,2}$, $G_{n,p}^{2,2}$ and $V_{n,p}$ having the (refined) block forms
\begin{equation} \label{eqn:BlockForms}
G_{n,p}^{1,2} = 
\begin{pmatrix}
0 & * \\
G_{n-2,p-1} & *
\end{pmatrix}
, \quad G_{n,p}^{2,2} = 
\begin{pmatrix}
\beta G_{n-2,p-1} & * \\
* & *
\end{pmatrix}
\quad \text{and} \quad V_{n,p} = 
\begin{pmatrix}
V_{n,p}^{1,1} & * \\
V_{n,p}^{2,1} & *
\end{pmatrix}
.
\end{equation}
To explain, the top-left block of $G_{n,p}^{1,2}$ corresponds to the scalar product of link states of types $1$ and $3$ (in the above order), which clearly gives zero.  Similarly, the bottom-left block of $G_{n,p}^{1,2}$ corresponds to the scalar product of link states of types $2$ and $3$, and it is easy to see that the value will not change if we remove $n-1$ and $n$ from both (cutting one link in each).  Finally, the top-left block of $G_{n,p}^{2,2}$ represents the scalar product of two type $3$ link states, hence the factor of $\beta$.

We substitute the block forms \eqref{eqn:BlockForms} and \eqref{eqn:GDec} (with $n \rightarrow n-1$) into \eqnref{eqn:G12}.  The first column then yields two equations:
\begin{subequations}
\begin{align}
0 &= G_{n-2,p} V_{n,p}^{1,1} + G_{n-1,p}^{1,2} V_{n,p}^{2,1} \label{eqn:G1211} \\
\text{and} \qquad G_{n-2,p-1} &= G_{n-1,p}^{2,1} V_{n,p}^{1,1} + G_{n-1,p}^{2,2} V_{n,p}^{2,1}. \label{eqn:G1212}
\end{align}
\end{subequations}
Using \eqnref{eqn:G12} (again with $n \rightarrow n-1$) on \eqref{eqn:G1211} gives
\begin{equation} \label{eqn:V11}
V_{n,p}^{1,1} + V_{n-1,p} V_{n,p}^{2,1} = 0,
\end{equation}
since $G_{n-2,p}$ is invertible by assumption ($\RadMod_{n-2,p} = \set{0}$).  To simplify \eqnref{eqn:G1212}, we note that $G_{n-1,p}^{2,1}$ is the transpose of $G_{n-1,p}^{1,2}$ and apply \eqref{eqn:GDec}, \eqref{eqn:G22'} (both with $n \rightarrow n-1$) and then \eqref{eqn:V11} to get
\begin{equation} \label{eqn:V21}
V_{n,p}^{2,1} = \frac{1}{\alpha_{n-1,p}} \id.
\end{equation}
Here we have used the invertibility of $G_{n-2,p-1}$.

Finally, we look at the top-left block of \eqnref{eqn:G22}.  Substituting the block forms \eqref{eqn:BlockForms}, applying \eqref{eqn:V21} and the invertibility of $G_{n-2,p-1}$ once more, we arrive at a recursion relation for the $\alpha_{n,p}$:
\begin{equation} \label{eqn:RecAlpha}
\alpha_{n,p} = \beta - \frac{1}{\alpha_{n-1,p}}.
\end{equation}
This simple relation allows us to reduce the computation of any $\alpha_{n,p}$ to that with $n$ smaller.  In particular, if $n = 2p$, we have $G_{n,p} = \beta G_{n-1,p-1}$, since $\NatMod_{n-1,p}$ is not defined and $\brac{n-1,p-1}$-link states are lifted to $\NatMod_{n,p}$ by adding a link.  This gives the starting point for the recursion, $\alpha_{2p,p} = \beta$.  It is easy to check that \eqnref{eqn:Alpha} is the unique solution.
\end{proof}

Since $\func{\det}{U_{n,p}} = 1$, \eqnref{eqn:Master} now yields an explicit recursion relation for the determinants of the Gram matrices:
\begin{subequations} \label{eqn:RecG}
\begin{equation} \label{eqn:RecG1}
\func{\det}{G_{n,p}} = \alpha_{n,p}^{d_{n-1,p-1}} \func{\det}{G_{n-1,p}} \func{\det}{G_{n-1,p-1}}.
\end{equation}
Here, $d_{n-1,p-1}$ is the dimension of $\NatMod_{n-1,p-1}$ as given in \eqnref{eqn:DimV}.  The starting points for this recursion are the results
\begin{equation} \label{eqn:RecG2}
\func{\det}{G_{n,0}} = 1 \qquad \text{and} \qquad \func{\det}{G_{2p-1,p}} = 1.
\end{equation}
\end{subequations}
The first follows because $\NatMod_{n,0}$ is spanned by a single link state consisting entirely of defects.  The second follows from putting $G_{2p,p} = \beta G_{2p-1,p-1}$ into \eqref{eqn:RecG1}.

It is clear from \eqnDref{eqn:Alpha}{eqn:RecG} that the determinant of $G_{n,p}$ can only vanish when $q$ is a root of unity.  This gives the following fundamental result:
\begin{proposition} \label{prop:AllIrr}
When $q$ is not a root of unity, the determinant of the Gram matrices of the $\NatMod_{n,p}$ are all non-zero, hence the $\NatMod_{n,p}$ are irreducible $\tl{n}$-modules.
\end{proposition}
\begin{corollary} \label{cor:RootOfUnity}
When $q$ is not a root of unity, $\tl{n}$ is a semisimple algebra and the $\NatMod_{n,p}$, $0 \leqslant p \leqslant \flr{n/2}$, form a complete set of non-isomorphic irreducible modules.
\end{corollary}
\begin{proof}
The $\NatMod_{n,p}$ are irreducible by \propref{prop:AllIrr}, distinct by \corref{cor:VLDistinct}, and form a complete set (meaning that there are no other non-isomorphic irreducibles) because the sum of the squares of their dimensions gives that of $\tl{n}$ (\eqnref{eqn:Curious}).  The result now follows from \propref{prop:WeddConverse}.
\end{proof}

We remark that the irreducibility of the $\NatMod_{n,p}$ for all $p$ is not sufficient to conclude that $\tl{n}$ is semisimple.  A counterexample is $\tl{2}$ with $\beta = 0$ for which the left regular representation is reducible but indecomposable, despite both $\NatMod_{2,1}$ and $\NatMod_{2,0}$ being irreducible (they both have dimension $1$).  The issue here is that, as we saw in \secref{sec:Reps}, these two standard modules are in fact isomorphic.  The sum of the squares of the dimensions of the \emph{distinct} irreducibles is thus $1 < \dim \tl{2} = 2$, so we cannot deduce semisimplicity.  For this, it is sufficient to have both the irreducibility and the inequivalence of the $\NatMod_{n,p}$.

Let us return now to the recursion relation \eqref{eqn:RecG} for the $\func{\det}{G_{n,p}}$.  Its solution gives us the main result of this section.
\begin{theorem} \label{thm:DetG}
For all $n,p$ and all $\beta = q + q^{-1} \in \CC$, the Gram matrix of $\NatMod_{n,p}$ has determinant
\begin{equation} \label{eqn:DetG}
\func{\det}{G_{n,p}} = \prod_{j=1}^p \brac{\frac{\qnum{n-2p+1+j}}{\qnum{j}}}^{d_{n,p-j}}.
\end{equation}
\end{theorem}
\begin{proof}
As in the proof of \propref{prop:Alpha}, we may first assume that $q$ is not a root of unity so that the recursion \eqref{eqn:RecG} is well-defined for all $n$ and $p$ ($q$ general then follows from continuity).  Checking \eqnref{eqn:RecG2} is easy, and \eqref{eqn:RecG1} follows from considering the numerators and denominators separately and applying \eqnref{eqn:RecD} several times.
\end{proof}
\noindent Note that \eqnref{eqn:DetG} is manifestly invariant under $q \leftrightarrow q^{-1}$ because $q$-numbers are unchanged under this transformation.  We remark that it is not entirely obvious from this formula that $\func{\det}{G_{n,p}}$ is actually finite for all $q \neq 0$, though this is clear from the definition (it is polynomial in $\beta$).  Here is the most useful consequence of the determinant formula \eqref{eqn:DetG}:

\begin{corollary} \label{cor:RadCrit}
If $\brac{n,p}$ is critical, then $\RadMod_{n,p} = \set{0}$, so $\NatMod_{n,p}$ is irreducible.
\end{corollary}
\begin{proof}
Recall that $\brac{n,p}$ critical means that $q^{2 \brac{n-2p+1}} = 1$.  We let $\ell$ be the smallest positive integer for which $q^{2 \ell} = 1$, so that $n-2p+1 = m \ell$ for some positive integer $m$, and
\begin{equation} \label{eqn:DetGCrit}
\func{\det}{G_{n,p}} = \prod_{j=1}^p \brac{\frac{\qnum{m \ell + j}}{\qnum{j}}}^{d_{n,p-j}}.
\end{equation}
The numerator only vanishes when $j = m' \ell$ for some positive integer $m'$ (this also requires $\ell > 1$), but then the denominator vanishes likewise.  Now note that
\begin{equation}
\qnum{m \ell} = \frac{q^{\ell} - q^{-\ell}}{q - q^{-1}} \frac{q^{m \ell} - q^{-m \ell}}{q^{\ell} - q^{-\ell}} = \qnum{\ell} \qnumber{m}{q^{\ell}},
\end{equation}
so the zero of $\qnum{m \ell}$ is always first order (because $q^{\ell} = \pm 1$, so $\qnumber{m}{q^{\ell}} \neq 0$ for $m \neq 0$).  Thus, the zeroes in the numerator of \eqref{eqn:DetGCrit} precisely cancel those in the denominator, hence $\func{\det}{G_{n,p}} \neq 0$.
\end{proof}
\noindent This corollary implies the semisimplicity of $\tl{n}$ when $\beta = \pm 2$ ($q = \pm 1$) because all $\brac{n,p}$ are then critical. We similarly obtain the semisimplicity of the $\tl{n}$ with $n$ odd at $\beta=0$ ($q = \pm \ii$).  Indeed, $\ell = 2$ divides $n-2p+1$ in this case, so the $\brac{n,p}$ with $n$ odd are all critical.  The inequivalence of the $\NatMod_{n,p}$ follows from \corref{cor:VLDistinct} as $n$ odd means, in particular, that $n \neq 2p$ for any $p$. Even though the $(n,p)$ are not critical when $\beta=0$ and $n$ is even, the tools developed up to now allow for the proof that $\NatMod_{2p,p}$, which then coincides with $\RadMod_{2p,p}$, is irreducible.
\begin{proposition} \label{prop:V=R=L}
When $\beta = 0$, the radical $\RadMod_{2p,p}$ is irreducible.
\end{proposition}
\begin{proof}
Since $\RadMod_{2p,p} = \NatMod_{2p,p}$ when $\beta = 0$, we must show that $\RadMod'_{2p,p} = \set{0}$ (\propref{prop:S/R'}).  Let $G'_{2p,p}$ denote the Gram matrix of the renormalised form $\inner{\cdot}{\cdot}'$ defined in \eqnref{eqn:RenormalisedForm}.  As noted before \propref{prop:AllIrr}, $G_{2p,p} = \beta G_{2p-1,p-1}$ and $G'_{2p,p} = \lim_{\beta\rightarrow 0} G_{2p,p}/\beta = \left. G_{2p-1,p-1} \right|_{\beta=0}$.  For $\beta=0$, the values of $q$ are $\pm \ii$ and $\qnum{m}$ is $0$ when $m$ is even and $\pm 1$ when it is odd.  Therefore,
\begin{equation}
\func{\det}{G'_{2p,p}} = \left. \func{\det}{G_{2p-1,p-1}} \right|_{\beta=0} = \Biggl. \prod_{j=1}^{p-1} \left( \frac{\qnum{j+2}}{\qnum{j}} \right)^{d_{2p-1,p-1-j}} \Biggr|_{q = \pm \ii}.
\end{equation}
Each factor of this product with $j$ odd is a power of $\pm 1$. When $j$ is even, we use $\qnum{2m} = \qnum{2} \qnumber{m}{q^2}$ to see that the factor's numerator and denominator are both polynomials in $\beta$ with a simple root at $\beta=0$.  Their quotient is therefore non-zero at $\beta = 0$. It follows that $\det \bigl( G'_{2p,p} \bigr) \neq 0$, whence the result.
\end{proof}
%

%
%

\section{Explorations at Roots of Unity} \label{sec:Explore}

The results of \secref{sec:Gram} not only tell us when the radical $\RadMod_{n,p}$ is trivial, but they also tell us its dimension when it is not.  Recall that \eqnref{eqn:RecD} amounts to a simple recursion relation for the dimensions of the standard modules:
\begin{equation} \label{eqn:RecDimV}
\dim \NatMod_{n,p} = \dim \NatMod_{n-1,p} + \dim \NatMod_{n-1,p-1}.
\end{equation}
The analogue for the $\RadMod_{n,p}$ is given as follows.  Let $q$ be a root of unity and let $\ell$ be the minimal positive integer satisfying $q^{2 \ell} = 1$.  For given $\brac{n,p}$ then, set
\begin{equation}\label{eq:kANDr}
n-2p+1 = \func{k}{n,p} \ell + \func{r}{n,p},
\end{equation}
where $\func{k}{n,p} \in \NN$ and $\func{r}{n,p} \in \set{1,\ldots,\ell-1,\ell}$.  The criticality of $\brac{n,p}$ is therefore equivalent to $\func{r}{n,p} = \ell$ when $q$ is a root of unity.
\begin{proposition} \label{prop:DimR}
The dimensions of the radicals $\RadMod_{n,p}$ satisfy the recursion relation
\begin{equation} \label{eqn:RecDimR}
\dim \RadMod_{n,p} = 
\begin{cases}
0 & \text{if $\func{r}{n,p} = \ell$,} \\
\dim \RadMod_{n-1,p} + \dim \NatMod_{n-1,p-1} & \text{if $\func{r}{n,p} = \ell - 1$,} \\
\dim \RadMod_{n-1,p} + \dim \RadMod_{n-1,p-1} & \text{otherwise,}
\end{cases}
\end{equation}
with initial conditions $\dim \RadMod_{n,0} = 0$ and $\dim \RadMod_{2p-1,p} = 0$.
\end{proposition}
\begin{proof}
The initial conditions are clear as $G_{n,0} = \bigl( 1 \bigr)$ for all $n$ and $\NatMod_{2p-1,p}$ is not defined for all $p$.  The recurrence when $\func{r}{n,p} = \ell$, which is when $\brac{n,p}$ is critical, also follows directly from \corref{cor:RadCrit}.  We may therefore assume that $\func{r}{n,p} \neq \ell$.  Then, $q^{2 \brac{n-2p+1}} = q^{2 \func{r}{n,p}} \neq 1$ by minimality of $\ell$, so \corref{cor:ResSplit} applies.  We can therefore make a change of bases to bring $G_{n,p}$ to the block diagonal form of \eqnref{eqn:Master}:
\begin{equation}
G_{n,p} = U_{n,p}^T 
\begin{pmatrix}
G_{n-1,p} & 0 \\
0 & \alpha_{n,p} G_{n-1,p-1}
\end{pmatrix}
U_{n,p}.
\end{equation}
We want the dimension of the kernel of $G_{n,p}$.  But, as $U_{n,p}$ is invertible,
\begin{equation}
G_{n,p} v = 0 \qquad \iff \qquad
\begin{pmatrix}
G_{n-1,p} w_1 \\
\alpha_{n,p} G_{n-1,p-1} w_2
\end{pmatrix}
 = 0, \qquad \text{where }
\begin{pmatrix}
w_1 \\
w_2
\end{pmatrix}
 = U_{n,p} v.
\end{equation}
Thus when $\alpha_{n,p} \neq 0$, $\ker (G_{n,p}) = \ker(G_{n-1,p}) \oplus \ker(G_{n-1,p-1})$ as vector spaces, whereas when  $\alpha_{n,p} = 0$, $\ker(G_{n,p}) = \ker(G_{n-1,p}) \oplus \NatMod_{n-1,p-1}$ as vector spaces.  The result now follows from \propref{prop:Alpha}.
\end{proof}
\begin{corollary} \label{cor:DimL}
The dimensions of the irreducible quotients $\IrrMod_{n,p} = \NatMod_{n,p} / \RadMod_{n,p}$ satisfy the recursion relation
\begin{equation} \label{eqn:RecDimL}
\dim \IrrMod_{n,p} = 
\begin{cases}
\dim \NatMod_{n,p} & \text{if $\func{r}{n,p} = \ell$,} \\
\dim \IrrMod_{n-1,p} & \text{if $\func{r}{n,p} = \ell - 1$,} \\
\dim \IrrMod_{n-1,p} + \dim \IrrMod_{n-1,p-1} & \text{otherwise,}
\end{cases}
\end{equation}
with initial conditions $\dim \IrrMod_{n,0} = 1$ and $\dim \IrrMod_{2p-1,p} = 0$.
\end{corollary}

We can use \propref{prop:DimR} and \corref{cor:DimL} to build up tables of dimensions of radicals and irreducibles.  It proves convenient to arrange these tables so that $n$ increases as we go down and $p$ increases in a south-westerly direction (staying constant to the south-east).  The important quantity $n-2p+1$ is therefore constant along vertical lines and increasing from left to right.  This corresponds to the standard arrangement of the Bratteli diagram of the family of Temperley-Lieb algebras.  To illustrate this, we present the first few rows of this diagram in \figref{fig:Bratteli}.  It is useful to mark the critical $\brac{n,p}$ --- these form vertical lines on the diagram which we will refer to as \emph{critical lines}. The regions bounded by consecutive critical lines will be referred to as \emph{critical strips}.

\begin{figure}
\begin{center}
\begin{tabular}{*{9}{@{}C@{\hspace{1mm}}}}
          &         |&          &          &          &          &          &          &          \\
          &\brac{1,0}&          &          &          &          &          &          &          \\
\brac{2,1}&         |&\brac{2,0}&         |&          &          &          &          &          \\
          &\brac{3,1}&          &\brac{3,0}&          &          &          &          &          \\
\brac{4,2}&         |&\brac{4,1}&         |&\brac{4,0}&         |&          &          &          \\
          &\brac{5,2}&          &\brac{5,1}&          &\brac{5,0}&          &          &          \\
\brac{6,3}&         |&\brac{6,2}&         |&\brac{6,1}&         |&\brac{6,0}&         |&          \\
          &\brac{7,3}&          &\brac{7,2}&          &\brac{7,1}&          &\brac{7,0}&          \\
\brac{8,4}&         |&\brac{8,3}&         |&\brac{8,2}&         |&\brac{8,1}&         |&\brac{8,0}\\
          &    \vdots&          &    \vdots&          &    \vdots&          &   
\vdots&          \\
      \end{tabular}
\caption{An example of the arrangement of pairs called a Bratteli diagram.  The critical lines in this picture correspond to $\beta = 0$, hence $\ell = 2$.} \label{fig:Bratteli}
\end{center}
\end{figure}

When we replace the pairs in the Bratteli diagram by some quantity indexed by $n$ and $p$, we shall also refer to the result as a Bratteli diagram.  In \figref{fig:IsingDims}, we show the first few rows of the Bratteli diagrams obtained by replacing $\brac{n,p}$ by $\dim \RadMod_{n,p}$ and $\dim \IrrMod_{n,p}$ when $\beta = \pm \sqrt{2}$ ($q = \pm e^{\ii \pi / 4}, \pm e^{3 \pi \ii / 4}$).  The features of these tables are similar for other roots of unity.  Notice that the radicals vanish on the critical lines in accordance with \corref{cor:RadCrit}.

\begin{figure}
\begin{center}
\begin{tabular}{*{13}{@{}C@{\hspace{1mm}}}@{}C@{\hspace{5mm}}*{13}{@{}C@{\hspace{1mm}}}}
   &  0&   &\phantom{265}&   &   &   &\phantom{265}&   &   &   &\phantom{265}&   &   &   &  1&   &\phantom{265}&   &   &   &\phantom{265}&   &   &   &\phantom{265}&    \\
  0&   &  0&  |&   &   &   &   &   &   &   &   &   &   &  1&   &  1&  |&   &   &   &   &   &   &   &   &    \\
   &  0&   &  0&   &   &   &   &   &   &   &   &   &   &   &  2&   &  1&   &   &   &   &   &   &   &   &    \\
  0&   &  1&  |&  0&   &   &   &   &   &   &   &   &   &  2&   &  2&  |&  1&   &   &   &   &   &   &   &    \\
   &  1&   &  0&   &  0&   &   &   &   &   &   &   &   &   &  4&   &  4&   &  1&   &   &   &   &   &   &    \\
  1&   &  5&  |&  0&   &  0&  |&   &   &   &   &   &   &  4&   &  4&  |&  5&   &  1&  |&   &   &   &   &    \\
   &  6&   &  0&   &  0&   &  0&   &   &   &   &   &   &   &  8&   & 14&   &  6&   &  1&   &   &   &   &    \\
  6&   & 20&  |&  0&   &  1&  |&  0&   &   &   &   &   &  8&   &  8&  |& 20&   &  6&  |&  1&   &   &   &    \\
   & 26&   &  0&   &  1&   &  0&   &  0&   &   &   &   &   & 16&   & 48&   & 26&   &  8&   &  1&   &   &    \\
 26&   & 74&  |&  1&   &  9&  |&  0&   &  0&  |&   &   & 16&   & 16&  |& 74&   & 26&  |&  9&   &  1&  |&    \\
   &100&   &  0&   & 10&   &  0&   &  0&   &  0&   &   &   & 32&   &165&   &100&   & 44&   & 10&   &  1&    \\
100&   &265&  |& 10&   & 54&  |&  0&   &  1&  |&  0&   & 32&   & 32&  |&265&   &100&  |& 54&   & 10&  |&  1 \\
\phantom{265}&\phantom{265}&\phantom{265}&\vdots&\phantom{265}&\phantom{265}&\phantom{265}&\vdots&\phantom{265}&\phantom{265}&\phantom{265}&\vdots&\phantom{265}&&\phantom{265}&\phantom{265}&\phantom{265}&\vdots&\phantom{265}&\phantom{265}&\phantom{265}&\vdots&\phantom{265}&\phantom{265}&\phantom{265}&\vdots&\phantom{265} \\
\multicolumn{13}{C}{\dim \RadMod_{n,p}} & \multicolumn{13}{C}{\dim \IrrMod_{n,p}} \\
\end{tabular}
\caption{The Bratteli diagrams for the dimensions of the radicals $\RadMod_{n,p}$ and the irreducibles $\IrrMod_{n,p}$ when $\beta = \pm \sqrt{2}$, hence $\ell = 4$.} \label{fig:IsingDims}
\end{center}
\end{figure}

The feature which is of most interest to us here is that the dimensions of the non-critical $\RadMod_{n,p}$ appear to coincide with those of the $\IrrMod_{n,p'}$, where $p'$ is obtained from $p$ by reflecting about the critical line immediately to the right.  This suggests that when the radical $\RadMod_{n,p}$ does not vanish, it is in fact \emph{irreducible}.  Establishing this irreducibility is somewhat difficult and will be the focus of \secref{sec:R=L}.  For now, we content ourselves with demonstrating this observed coincidence of dimensions.
\begin{proposition} \label{prop:DimR=DimL}
Let $q$ be a root of unity and $\brac{n,p}$ be non-critical. Then, $\dim \RadMod_{n,p}$ is equal to $\dim \IrrMod_{n,p + \func{r}{n,p} - \ell}$ if $p + \func{r}{n,p} - \ell\geqslant 0$ and $0$ otherwise.
\end{proposition}
\noindent We remark that the indices $(n,p)$ of $\RadMod$ and $(n,p + \func{r}{n,p} - \ell)$ of $\IrrMod$ form a \emph{symmetric pair}, a concept introduced in \secref{sec:R=L}.
\begin{proof}
This is clear for $n=1$ as then, $\ell>2$ and the index $p+r(n,p)-\ell=2-\ell$ is negative, so that $\dim \RadMod_{n,p}=0$ as needed. So assume that the proposition is true for $n-1$ and all $p$.  We have to consider four cases for the pairs $\brac{n,p}$, corresponding to whether $\brac{n-1,p}$ and/or $\brac{n-1,p-1}$ lie on critical lines.  The position of the pairs $(n-1,p)$, $(n,p)$ and $(n-1,p-1)$, relative to the critical lines, is pictured below in each case.
\begin{center}
\begin{tabular}{ccc}
$
\begin{smallmatrix}
|\ \ \dots  &         & \dots &        &  \dots\ \ |  \\
            &  (n-1,p) &   &  (n-1,p-1)&    \\ 
|\ \ \dots  &         & (n,p) &        &  \dots\ \ | \\
 & (n+1,p+1) &   &  (n+1,p)&    \\
|\ \  \dots  &         & \dots &        &  \dots\ \ |
\end{smallmatrix}
$
&
\hspace{20mm}
&
$
\begin{smallmatrix}
\dots  &         & \dots &    |    &  \dots   \\
  &  (n-1,p) &   &  (n-1,p-1)&    \\ 
\dots  &         & (n,p) &    |    &  \dots  \\
 & (n+1,p+1) &   &  (n+1,p)&    \\
 \dots  &        & \dots &    |    &  \dots
\end{smallmatrix}
$
\\[8mm]
Case (1) & \hspace{20mm} & Case (2)
\end{tabular}
\end{center}
\begin{center}
\begin{tabular}{ccc}
$
\begin{smallmatrix}
\dots  &    |    & \dots &            &   \dots \\
       & (n-1,p) &       &  (n-1,p-1) & \\
\dots  &    |    & (n,p) &            &   \dots \\
       &(n+1,p+1)&       &  (n+1,p)   & \\
\dots  &    |    & \dots &            &   \dots 
\end{smallmatrix}
$
&
\hspace{20mm}
&
$
\begin{smallmatrix}
\dots & |       &   \dots &  | &\dots    \\
      & (n-1,p) &         &  (n-1,p-1) & \\ 
\dots & |       &  (n,p)  &  | &\dots    \\
      & (n+1,p+1)&        &  (n+1,p)   & \\
\dots & |      &   \dots  &  | &\dots  
\end{smallmatrix}
$
\\[8mm]
Case (3) & \hspace{20mm} & Case (4)
\end{tabular}
\end{center}
\begin{enumerate}
\item Both $\brac{n-1,p}$ and $\brac{n-1,p-1}$ are non-critical, so $\func{r}{n,p} \notin \set{1, \ell - 1, \ell}$.  Then,
\begin{subequations}
\begin{align}
\dim \RadMod_{n-1,p} &= \dim \IrrMod_{n-1,p + \func{r}{n-1,p} - \ell} = \dim \IrrMod_{n-1,p - 1 + \func{r}{n,p} - \ell} \label{eqn:Dim1} \\
\text{and} \qquad \dim \RadMod_{n-1,p-1} &= \dim \IrrMod_{n-1,p-1 + \func{r}{n-1,p-1} - \ell} = \dim \IrrMod_{n-1,p + \func{r}{n,p} - \ell}. \label{eqn:Dim2}
\end{align}
\propref{prop:DimR} now gives $\dim \RadMod_{n,p}$ as the sum of the left-hand sides.  To compute the sum of the right-hand sides, we note that $\func{r}{n,p+\func{r}{n,p}-\ell} = -\func{r}{n,p} \pmod{\ell}$, hence
\corref{cor:DimL} gives $\dim \IrrMod_{n,p + \func{r}{n,p} - \ell}$ for this sum.
\item Only $\brac{n-1,p-1}$ is critical.  Then, $\func{r}{n,p} = \ell - 1$ and \eqref{eqn:Dim1} is valid, but instead of \eqref{eqn:Dim2}, we have
\begin{equation} \label{eqn:Dim3}
\dim \NatMod_{n-1,p-1} = \dim \IrrMod_{n-1,p-1} = \dim \IrrMod_{n-1,p + \func{r}{n,p} - \ell}.
\end{equation}
Now apply \propref{prop:DimR} and \corref{cor:DimL} to the sum of \eqref{eqn:Dim1} and \eqref{eqn:Dim3}.
\item Only $\brac{n-1,p}$ is critical.  Now, $\func{r}{n,p} = 1$, \eqref{eqn:Dim2} is valid, and we use
\begin{equation} \label{eqn:Dim4}
\dim \RadMod_{n-1,p} = 0.
\end{equation}
Adding \eqref{eqn:Dim2} and \eqref{eqn:Dim4} then gives the result.
\item Both $\brac{n-1,p}$ and $\brac{n-1,p-1}$ are critical.  Thus, $\func{r}{n,p} = 1 = \ell - 1$, that is, $\ell = 2$ (and so $\beta = 0$).  The result now follows from adding \eqref{eqn:Dim3} and \eqref{eqn:Dim4}. \qedhere
\end{subequations}
\end{enumerate}
\end{proof}
We remark that there is one (non-trivial) case in which we already know that the radical is irreducible.  This is the content of \propref{prop:V=R=L} which asserts that $\RadMod_{2p,p}$ is irreducible when $\beta = 0$.  Of course, this does not tell us that $\RadMod_{2p,p}$ is isomorphic to $\IrrMod_{2p,p-1}$ (since $\ell=2$ when $\beta = 0$) as \propref{prop:DimR=DimL} would have us believe.

%
%

\section{Induced Modules} \label{sec:Ind}

Recall that in \propref{prop:Restrict}, we studied the result of restricting the action on a $\tl{n}$-module to the action of the subalgebra $\tl{n-1}$ spanned by the $n$-diagrams in which the two $n$-th points are joined.  Restriction has a close relative which constructs instead a $\tl{n+1}$-module from a given $\tl{n}$-module (subject to the analogous inclusion of $\tl{n}$ in $\tl{n+1}$).  This is, of course, the \emph{induced module} construction.  We therefore consider the induced modules
\begin{equation}
\Ind{\NatMod_{n,p}} = \tl{n+1} \otimes_{\tl{n}} \NatMod_{n,p}.
\end{equation}
These are $\tl{n+1}$-modules in which the action is given by left-multiplication:  $a \brac{b \otimes z} = \brac{ab} \otimes z$ for all $a,b \in \tl{n+1}$ and $z \in \NatMod_{n,p}$.  The subtlety of this definition lies in the subscript on the ``$\otimes$'' which informs us that the tensor product is ``permeable'' to elements of $\tl{n}$.  More precisely, this means that $ab \otimes z = a \otimes bz$ for all $a \in \tl{n+1}$, $b \in \tl{n}$ and $z \in \NatMod_{n,p}$ (the elements of $\tl{n}$ are scalars as far as the tensor product is concerned).  Indeed, $\Ind{\NatMod_{n,p}}$ may alternatively be characterised as the quotient of the $\CC$-tensor product $\tl{n+1} \otimes_{\CC} \NatMod_{n,p}$ (with the $\tl{n+1}$-action $a \brac{a' \otimes_{\CC} z} = \brac{aa'} \otimes_{\CC} z$ as before) by the submodule generated by the elements of the form $(ab \otimes_{\CC} z - a \otimes_{\CC} bz)$.  This characterisation will shortly prove useful.

By \propref{prop:JNO}, any monomial in the $\tl{n+1}$-generators $u_i$ can be written in reverse Jones' normal form. In this form, the generator $u_n$ appears at most once to the far left and may only be preceded by a string of the form $u_ru_{r+1}\cdots u_{n-1}$. We conclude that $\tl{n+1}$ is spanned by monomials of the forms $u_r \cdots u_{n-1} u_n U$ and $U'$, where $U$ and $U'$ are monomials in the $\tl{n}$-subalgebra.  Thus, if $D$ denotes a basis of $\NatMod_{n,p}$, then $\Ind{\NatMod_{n,p}}$ is \emph{spanned} by the set
\begin{equation}\label{eq:badSpanningSet}
\set{\ \wun \otimes d \ , \ u_n \otimes d \ , \ u_{n-1} u_n \otimes d \ , \dots, \ u_1 \cdots u_{n-1} u_n \otimes d \st d \in D \ }.
\end{equation}
Note that this spanning set is not usually a basis of $\Ind{\NatMod_{n,p}}$.  For example, if $d$ can be written as $u_{n-1} d'$ for some $d' \in \NatMod_{n,p}$, then
\begin{equation}
u_{n-1} u_n \otimes d = u_{n-1} u_n \otimes u_{n-1} d' = u_{n-1} u_n u_{n-1} \otimes d' = u_{n-1} \otimes d' = \wun \otimes u_{n-1} d' = \wun \otimes d.
\end{equation}
To obtain a basis, we generalise this computation.  As for simple links in an element of $\tl n$, we will say that a link in a link state $d$ is {\em simple} if it joins neighbouring positions $i$ and $i+1$. We then say that $d$ has a simple link at $i$.  Obviously, if $d = u_i d'$, where $d$ and $d'$ are link states, then $d$ will have a simple link at $i$.  The converse is the following:
\begin{lemma} \label{lem:Basis}
If $d$ is an $\brac{n,p}$-link state with $n \geqslant 3$ which has a simple link at $i$, then $d$ can always be expressed as $u_i d'$ with $d'$ another $\brac{n,p}$-link state.  When $n = 2$, such an expression is valid if and only if $\beta \neq 0$.
\end{lemma}
\begin{proof}
We construct $d'$ explicitly to be the $\brac{n,p}$-link state which is identical to $d$ except that the simple link joining $i$ and $i+1$ is swapped with whatever appears at $i-1$ (or $i+1$).  Schematically,
\begin{equation}
d = \ 
\parbox{10mm}{\begin{center}
\psfrag{i}[][]{$\scriptstyle i$}
\psfrag{im}[][]{$\scriptstyle i-1$}
\psfrag{ip}[][]{$\scriptstyle i+1$}
\includegraphics[height=21mm]{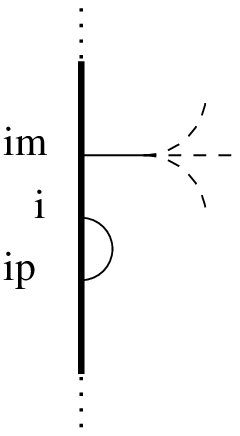}
\end{center}}
\qquad \Rightarrow \qquad d' = \ 
\parbox{11mm}{\begin{center}
\psfrag{i}[][]{$\scriptstyle i$}
\psfrag{ip}[][]{$\scriptstyle i+1$}
\psfrag{im}[][]{$\scriptstyle i-1$}
\includegraphics[height=21mm]{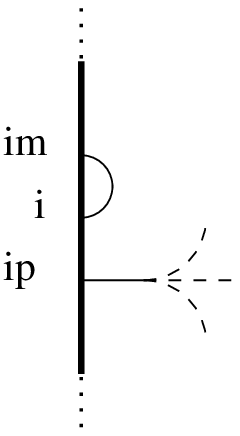}
\end{center}}
\ .
\end{equation}
It is clear that $d'$ satisfies $u_i d' = d$ and likewise clear that this construction requires $n \geqslant 3$.  The only remaining possibility then concerns the unique $\brac{2,1}$-link state $d = \ 
\parbox{2mm}{\begin{center}
\includegraphics[height=4mm]{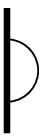}
\end{center}}
\ $.  This satisfies $d = \beta^{-1} u_1 d$ when $\beta \neq 0$, but when $\beta = 0$, $u_1$ acts as the zero operator on $\NatMod_{2,1}$.
\end{proof}

This lemma allows us to reduce the number of elements in the set \eqref{eq:badSpanningSet} while preserving its spanning property.  But before describing this reduction, we introduce some more vocabulary.  The $(n,p)$-link state $d$ will be called {\em $r$-admissible}, $1\leqslant r\leqslant n$, if $d$ has no simple link at $r$ or below, meaning at any $i$ with $i\geqslant r$. It follows that all $(n,p)$-link states are $n$-admissible. Moreover, an element $u_r u_{r+1} \cdots u_n \otimes d \in \Ind{\NatMod_{n,p}}$ of \eqref{eq:badSpanningSet} will be called {\em admissible} if $d$ is $r$-admissible. 

\begin{lemma} \label{lem:betterSpanningSet}
Let $n\geqslant 3$ and $0\leqslant p\leqslant \lfloor n/2\rfloor$. Let $u\in\tl{n+1}$ be a word in the generators and $d$ be an $(n,p)$-link state. Then, there exists an integer $s$, $1\leqslant s\leqslant n+1$, and an element $e\in \NatMod_{n,p}$ that is either $0$ or an $s$-admissible $(n,p)$-link state, such that $u\otimes d=u_su_{s+1}\cdots u_n\otimes e$ in $\Ind{\NatMod_{n,p}}$ (if $s=n+1$, then $u\otimes d=\wun\otimes e$).
\end{lemma}
\begin{proof}
We suppose that the word $u$ is written in reverse Jones' normal form. If $u$ does not contain $u_n$, then $u\otimes d=\wun\otimes ud$ and the form suggested is obtained by choosing $s=n+1$ and $e=ud$. Suppose now that $u$ is $u_ru_{r+1}\cdots u_nu'$ where $u'$ is a word of $\tl{n}$ and therefore $u\otimes d=u_ru_{r+1}\cdots u_n\otimes u'd$. If the element $u'd$ is zero in $\NatMod_{n,p}$, then again $u_ru_{r+1}\cdots u_n\otimes e$ is of the desired form with $e=0$. If $u'd$ is $r$-admissible, take $e=u'd$.

The only remaining case is when $u'd$ is not $r$-admissible. This means that there is a simple link at $r$ or below. Let $t$ be the position of the lowest simple link. Clearly $r\leqslant t\leqslant n-1$ and thus $u_t\in\tl{n}$. Because $n\geqslant 3$, \lemref{lem:Basis} ensures the existence of an $(n,p)$-link state $f$ such that $u_tf=u'd$. Then,
\begin{align}
u_ru_{r+1}\cdots u_n\otimes u'd&=u_r\cdots u_n\otimes u_tf\notag\\
&=u_r\cdots u_{t-1}u_tu_{t+1}u_tu_{t+2}\cdots u_n\otimes f\notag\\
&=u_r\cdots u_{t-1}u_tu_{t+2}\cdots u_n\otimes f\notag\\
&=u_{t+2}\cdots u_n\otimes u_r\cdots u_t f\notag\\
&=u_{t+2}\cdots u_n\otimes u_r\cdots u_{t-1} u'd.
\end{align}
(The two extreme values of $t$ lead to $u_{r+2}\cdots u_n\otimes u'd$ for $t=r$ and $\wun\otimes u_r\cdots u_{n-2}u'd$ for $t=n-1$.) The configurations of $u'd$ and of $u_r\cdots u_{t-1} u'd$ are shown in the following diagrams:
\begin{equation*}
\text{If\ \ }u'd\ =\parbox{9mm}{\begin{center}
\psfrag{r}[][]{$\scriptstyle r$}
\psfrag{t}[][]{$\scriptstyle t$}
\includegraphics[height=28mm]{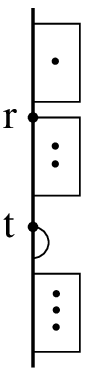}
\end{center}}\ ,\quad \text{then\ \ }u_r\cdots u_{t-1}u'd\ =
\parbox{13mm}{\begin{center}
\psfrag{r}[][]{$\scriptstyle r$}
\psfrag{t}[][]{$\scriptstyle t$}
\includegraphics[height=28mm]{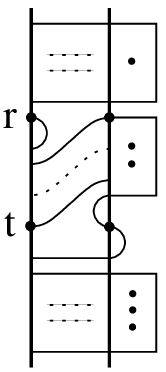}
\end{center}}\ =
\parbox{13mm}{\begin{center}
\psfrag{r}[][]{$\scriptstyle r$}
\psfrag{t+2}[][]{$\scriptstyle t+2$}
\includegraphics[height=28mm]{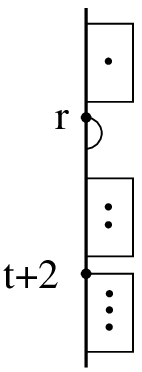}
\end{center}}\ .
\end{equation*}
Note that in the lower box of $u'd$, marked by three dots, all positions are either defects or linked with positions in the two top boxes, since the lowest simple link is at $t$. (The two top boxes may have defects, links within and between them and, as just said, links with the bottom box.) The action of $u_r\cdots u_{t-1}$ is seen to exchange the simple link at $t$ and the middle box. Since, after this exchange, the first position in the lower box is $t+2$, then the resulting link state is $(t+2)$-admissible. Clearly this is of the desired form if we set $s=t+2$ (and then $s\leqslant n+1$) and $e$ is taken to be $u_r\cdots u_{t-1} u'd$.
\end{proof}

Let $\mathcal S_{n,p}$ be the set consisting of all elements $\wun\otimes d$ for any $(n,p)$-link state $d$ and of elements $u_ru_{r+1}\cdots u_n\otimes d$ for all $1\leqslant r\leqslant n$ and all $r$-admissible $(n,p)$-link states $d$. Due to the previous lemma, $\mathcal S_{n,p}$ is a spanning set of $\Ind{\NatMod_{n,p}}$ for $n\geqslant 3$. (One can check that $\mathcal S_{2,p}$ spans $\Ind{\NatMod_{2,p}}$ except when $p=1$ and $\beta=0$.) As an example,
\begin{equation}
\mathcal S_{3,1}=\big\{\wun\otimes \parbox{2mm}{\begin{center}
\includegraphics[height=5mm]{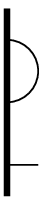}
\end{center}}, \quad
\wun\otimes \parbox{2mm}{\begin{center}
\includegraphics[height=5mm]{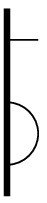}
\end{center}},\quad
u_3\otimes \parbox{2mm}{\begin{center}
\includegraphics[height=5mm]{tl-link-31a}
\end{center}},\quad
u_3\otimes \parbox{2mm}{\begin{center}
\includegraphics[height=5mm]{tl-link-31b}
\end{center}},\quad
u_2u_3\otimes\parbox{2mm}{\begin{center}
\includegraphics[height=5mm]{tl-link-31a}
\end{center}}
\big\}
\end{equation}
and $\abs{\mathcal S_{3,1}}=5$. Our aim is to show that $\mathcal S_{n,p}$ is actually a basis of $\Ind{\NatMod_{n,p}}$. For this, we first introduce a linear map $d \mapsto \underline{d}$ between the vector spaces underlying $\NatMod_{n,p}$ and $\NatMod_{n+2,p+1}$.  This is defined on the basis of $\brac{n,p}$-link states by adding two points to each basis element and a simple link from $n+1$ to $n+2$:
\begin{equation}
\underline{d}=
\parbox{13mm}{\begin{center}
\psfrag{z}[][]{$\textstyle d$}
\psfrag{r}[][]{$\scriptstyle n+1$}
\includegraphics[height=20mm]{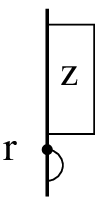}
\end{center}}.
\end{equation}
We now define
\begin{equation}
\Phi \colon \Ind{\NatMod_{n,p}} \rightarrow \Res{\NatMod_{n+2,p+1}}, \qquad \Phi(u\otimes d)=u\underline{d} \qquad \text{($u\in\tl{n+1}$, $d\in\NatMod_{n,p}$).}
\end{equation}

Our first concern is to check whether $\Phi$ is well-defined, that is, whether the images $\Phi(u\otimes d)$ and $\Phi(v\otimes e)$ are equal when $u\otimes d=v\otimes e$ for $u,v\in\tl{n+1}$ and $d,e\in\NatMod_{n,p}$.  From the characterisation of induced modules as quotients of the $\CC$-tensor product, it follows that $\Phi$ will be well-defined if $\Phi(uu'\otimes d) = \Phi(u\otimes u'd)$ for all $u\in\tl{n+1}$, $u'\in\tl{n}$ and $d\in\NatMod_{n,p}$.  But, this follows immediately from the associativity of the $\tl{n+1}$-action on $\Res{\NatMod_{n+2,p+1}}$:
\begin{equation}
\Phi(uu'\otimes d)=(uu')\underline{d}=u(u'\underline{d})=u(\underline{u'd})=\Phi(u\otimes u'd).
\end{equation}
We remark that the $u'$ in $(u'\underline{d})$ should be interpreted as the image of $u' \in \tl{n}$ in $\tl{n+2}$ under the inclusion which adds two defects at the bottom.

The associativity of the $\tl{n+1}$-action also makes $\Phi$ a homomorphism of $\tl{n+1}$-modules (here is where the restriction is necessary):
\begin{equation}
u'\Phi(u\otimes d)=u'(u\underline{d})=(u'u)\underline{d}=\Phi(u'u\otimes d) \qquad \text{for all $u,u' \in \tl{n+1}$, $d \in \NatMod_{n,p}$.}
\end{equation}
To prove that it is actually an isomorphism does not require much more effort. 
\begin{proposition}\label{prop:SisTheBasis}\ 

\begin{enumerate}[label=\textup{(\roman*)}]
\item For all $n\geqslant 1$, $0\leqslant p\leqslant \lfloor n/2\rfloor$ and $\beta\in\mathbb C$, the set $\mathcal S_{n,p}$ is a basis of $\Ind{\NatMod_{n,p}}$, except for $\Ind{\NatMod_{2,1}}$ at $\beta=0$ which has instead $\bigl\{ \wun \otimes 
\parbox{2mm}{\begin{center}
\includegraphics[height=6mm]{tl-link-21}
\end{center}}
\ , u_2 \otimes 
\parbox{2mm}{\begin{center}
\includegraphics[height=6mm]{tl-link-21}
\end{center}}
\ , u_1 u_2 \otimes 
\parbox{2mm}{\begin{center}
\includegraphics[height=6mm]{tl-link-21}
\end{center}}
\ \bigr\}$ as a basis. \label{it:Basis1}
\item When $(n,p)\neq (2,1)$ or $\beta\neq 0$, the $\tl{n+1}$-modules $\Ind{\NatMod_{n,p}}$ and $\Res{\NatMod_{n+2,p+1}}$ are isomorphic. \label{it:Basis2}
\end{enumerate}
\end{proposition}
\begin{proof} The basis for $\Ind{\NatMod_{2,1}}$ is obtained by direct computation. For the general case, we show that the homomorphism $\Phi$ is both surjective and injective. The two statements \ref{it:Basis1} and \ref{it:Basis2} are clear consequences.

Since the $(n+2,p+1)$-link states define a basis of $\Res{\NatMod_{n+2,p+1}}$, the surjectivity of $\Phi$ will follow if every $\brac{n+2,p+1}$-link state has a preimage in $\mathcal{S}_{n,p}$. Suppose that $e$ is such a link state and note that $e$ has at least one simple link. To construct $u\otimes d\in \mathcal S_{n,p}$ such that $\Phi(u\otimes d)=e$, let $r$ be the position of the lowest simple link in $e$. Denote by $d$ the $(n,p)$-link state obtained from $e$ by deleting this lowest simple link. Since there are no simple links of $e$ below $r$, the link $d$ is $r$-admissible and therefore $u_r u_{r+1} \cdots u_n\otimes d=u\otimes d\in\mathcal{S}_{n,p}$. (Note that if $r=n+1$, then $u_r u_{r+1} \cdots u_n$ is understood to be $\wun$.) It is now easy to verify that $\Phi(u_ru_{r+1}\cdots u_n\otimes d)=e$:
\begin{gather}
\text{If} \qquad 
e = \ 
\parbox{11mm}{\begin{center}
\psfrag{r}[][]{$\scriptstyle r$}
\psfrag{s}[][]{$\scriptstyle r+2$}
\psfrag{t}[][]{$\scriptstyle n+2$}
\includegraphics[height=24mm]{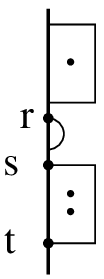}
\end{center}}
, \qquad \text{then} \qquad d = \ 
\parbox{11mm}{\begin{center}
\psfrag{r}[][]{$\scriptstyle r$}
\psfrag{t}[][]{$\scriptstyle n$}
\includegraphics[height=18mm]{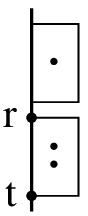}
\end{center}}
\notag \\
\text{and} \qquad \Phi (u_r u_{r+1} \cdots u_n\otimes d) = u_r u_{r+1} \cdots u_n\underline{d} = \ 
\parbox{15mm}{\begin{center}
\psfrag{r}[][]{$\scriptstyle r$}
\psfrag{s}[][]{$\scriptstyle n$}
\psfrag{t}[][]{$\scriptstyle n+2$}
\includegraphics[height=24mm]{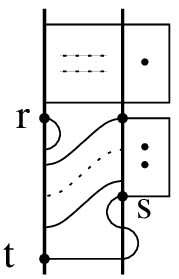}
\end{center}}
\ = \ 
\parbox{11mm}{\begin{center}
\psfrag{r}[][]{$\scriptstyle r$}
\psfrag{s}[][]{$\scriptstyle r+2$}
\psfrag{t}[][]{$\scriptstyle n+2$}
\includegraphics[height=24mm]{tl-F62}
\end{center}}
\ = e.
\end{gather}
Therefore, every $(n+2,p+1)$-link state has a preimage in $\mathcal S_{n,p}$ and $\Phi$ is surjective.

The map $\Phi$ sends every element of $\mathcal S_{n,p}$ to an $\brac{n+2,p+1}$-link state. Since the latter are linearly independent, the injectivity of $\Phi$ will be established if distinct elements of $\mathcal S_{n,p}$ have distinct images. Suppose then that $w_1=u_r u_{r+1}\cdots u_n\otimes d$ and $w_2=u_s u_{s+1}\cdots u_n\otimes d'$ are elements of $\mathcal S_{n,p}$ with the same image. The link state $d$, being $r$-admissible, can be represented as above by two blocks, that marked by two dots being devoid of simple links. The image $\Phi(w_1)$ is then represented by the above diagram for $e$ whose lowest simple link is at $r$. Similarly, the lowest simple link of $\Phi(w_2)$ must be at $s$. If $\Phi(w_1)=\Phi(w_2)$, then $r$ and $s$ must be equal. But then, the two link states $\Phi(w_1)$ and $\Phi(w_2)$ will coincide if and only if the upper boxes (marked by one dot) of $d$ and $d'$ coincide and similarly for their lower boxes (marked by two dots). This forces the original $d$ and $e$ to be equal and thus $w_1=w_2$.
\end{proof}
\noindent For completeness, we mention that as $\NatMod_{2,1} \cong \NatMod_{2,0}$ when $\beta = 0$, it follows that $\Ind{\NatMod_{2,1}} \cong \NatMod_{3,1} \oplus \NatMod_{3,0}$.  However, \propref{prop:Restrict} gives $\Res{\NatMod_{4,2}} \cong \NatMod_{3,1}$.

Because restricting a module does not change its dimension, the dimension of $\Ind{\NatMod_{n,p}}$ is that of $\NatMod_{n+2,p+1}$.  We therefore obtain:
\begin{corollary}\label{cor:inducedDimension}
The dimension of the induced module $\Ind{\NatMod_{n,p}}$ is
\begin{equation}
\dim \Ind{\NatMod_{n,p}} = 
\begin{cases}
3 & \text{if $\brac{n,p} = \brac{2,1}$ and $\beta = 0$,} \\
d_{n+2,p+1} & \text{otherwise.}
\end{cases}
\end{equation}
\end{corollary}
\noindent Moreover, the structure of induced modules now follows, using the isomorphism $\Phi$, from \propref{prop:Restrict} and \corref{cor:ResSplit} which describe the structure of restricted ones.
\begin{corollary} \label{cor:Induct}
When $\brac{n,p} \neq \brac{2,1}$ or $\beta \neq 0$, the sequence
\begin{equation} \label{SES:Ind}
\dses{\NatMod_{n+1,p+1}}{}{\Ind{\NatMod_{n,p}}}{}{\NatMod_{n+1,p}}
\end{equation}
is exact.
\end{corollary}
\begin{corollary} \label{cor:IndSplit}
When $\brac{n,p}$ is non-critical, the exact sequence \eqref{SES:Ind} splits, so we have (for $\brac{n,p} \neq \brac{2,1}$ or $\beta \neq 0$)
\begin{equation}
\Ind{\NatMod_{n,p}} \cong \NatMod_{n+1,p+1} \oplus \NatMod_{n+1,p} \qquad \text{(as $\tl{n+1}$-modules).}
\end{equation}
\end{corollary}

Whether or not the exact sequence \eqref{SES:Ind} splits, a submodule isomorphic to $\NatMod_{n+1,p+1}$ is easily identified in $\Ind{\NatMod_{n,p}}$. In $\Res{\NatMod_{n+2,p+1}}$, such a submodule is spanned by those $\brac{n+2,p+1}$-link states which have a defect at position $n+2$ (see the proof of \propref{prop:Restrict}). Their images under $\Phi^{-1}$ have the form $u_r u_{r+1}\cdots u_n\otimes d$, with $r\leqslant n$ and $d$ an $\brac{n,p}$-link state with a defect at $n$. Then the injection $\NatMod_{n+1,p+1} \rightarrow \Ind{\NatMod_{n,p}}$, call it $\alpha$, of the exact sequence \eqref{SES:Ind} sends $\brac{n+1,p+1}$-link states to elements of $\mathcal S_{n,p}$ as follows: If $d$ is an $\brac{n+1,p+1}$-link state, then $\alpha(d) = u_r u_{r+1}\cdots u_n\otimes d'$, where the lowest simple link in $d$ is at $r$ and $d'$ is the $(n,p)$-link state obtained from $d$ by removing this simple link and adding a defect at position $n$.

The exact sequences in \eqnDref{SES:Res}{SES:Ind} can be read off the Bratteli diagram quite easily.  For restriction, the $\tl{n-1}$-modules $\NatMod_{n-1,p}$ and $\NatMod_{n-1,p-1}$ appearing in the exact sequence \eqref{SES:Res} for $\Res{\NatMod_{n,p}}$ correspond to the entries immediately above, and to the left and right, respectively, of the entry corresponding to $\NatMod_{n,p}$.  For induction, we must instead look immediately below, and again to the left and right, to find the entries indicating the constituents of the exact sequence \eqref{SES:Ind}.  We have learned that these exact sequences split when the module being restricted or induced is non-critical.  The question of whether the sequences split when the module is critical will not be resolved until \secref{sec:Proj}.

We conclude with an example showing that the induced module $\Ind{\NatMod_{n,p}}$ may be different to the direct sum $\NatMod_{n+1,p+1} \oplus \NatMod_{n+1,p}$ and even to the quotient $\LinkMod_{n+1,p} / \LinkMod_{n+1,p+2}$.  These three $\tl{n+1}$-modules share the same exact sequence \eqref{SES:Ind}, the latter because \eqnDref{eqn:Filtration}{eqn:DefNatMod} give
\begin{equation}
\NatMod_{n+1,p+1} = \frac{\LinkMod_{n+1,p+1}}{\LinkMod_{n+1,p+2}} \subseteq \frac{\LinkMod_{n+1,p}}{\LinkMod_{n+1,p+2}} \qquad \text{and} \qquad \frac{\LinkMod_{n+1,p} / \LinkMod_{n+1,p+2}}{\LinkMod_{n+1,p+1} / \LinkMod_{n+1,p+2}} \cong \frac{\LinkMod_{n+1,p}}{\LinkMod_{n+1,p+1}} = \NatMod_{n+1,p}.
\end{equation}
For our example, we take $q = e^{\ii \pi / 3}$, so that $\beta=1$ and $\brac{n,p} = \brac{2,0}$ is critical.  Then, one can check that the central element $F_3\in\tl{3}$, introduced in \appref{app:Casimir}, is represented 
on $\Ind{\NatMod_{2,0}}$ and $\LinkMod_3 = \LinkMod_{3,0} / \LinkMod_{3,2}$, with respective (ordered) bases $\set{ u_2 \otimes 
\parbox{2mm}{\begin{center}
\includegraphics[height=4mm]{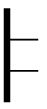}
\end{center}}
\ , u_1 u_2 \otimes 
\parbox{2mm}{\begin{center}
\includegraphics[height=4mm]{tl-link-20}
\end{center}}\ ,
\wun \otimes 
\parbox{2mm}{\begin{center}
\includegraphics[height=4mm]{tl-link-20}
\end{center}}
}$ and $\set{
\parbox{2mm}{\begin{center}
\includegraphics[height=5mm]{tl-link-31b}
\end{center}}
\ , 
\parbox{2mm}{\begin{center}
\includegraphics[height=5mm]{tl-link-31a}
\end{center}}
\ ,
\parbox{2mm}{\begin{center}
\includegraphics[height=5mm]{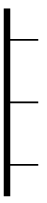}
\end{center}}
}$, by
\begin{equation}\label{eq:nonDiagF3}
\begin{pmatrix}
-1 & 0 & 3 \\
0 & -1 & -3 \\
0 & 0 & -1
\end{pmatrix}
\qquad \text{and} \qquad
\begin{pmatrix}
-1 & 0 & 0 \\
0 & -1 & 0 \\
0 & 0 & -1
\end{pmatrix}
.
\end{equation}
Thus, this central element can be diagonalised on $\LinkMod_3$ (and on $\NatMod_{3,1} \oplus \NatMod_{3,0}$ by \propref{prop:EndLinks}), but not on $\Ind{\NatMod_{2,0}}$.

%
%

\section{The Irreducibility of the Radicals} \label{sec:R=L}

The present section carries on the exploration, launched in \secref{sec:Explore}, of the standard modules at roots of unity. The goal here is to prove that the radicals $\RadMod_{n,p}$ of the standard modules $\NatMod_{n,p}$ are either trivial (meaning $\RadMod_{n,p} = \set{0}$) or irreducible. One can actually be more precise: Let $q$ be a root of unity with $\ell$ the smallest positive integer such that $q^{2\ell}=1$. We shall say that two pairs $(n,p')$ and $(n,p)$ with $0 < \abs{p'-p} < \ell$ form a \emph{symmetric pair} if $(n,p')$ and $(n,p)$ are non-critical and are located symmetrically on either side of the (single) critical line between them. The adjective \emph{symmetric} will also be used when two objects, for example $\NatMod_{n,p'}$ and $\NatMod_{n,p}$, are labelled by a symmetric pair $(n,p')$ and $(n,p)$. The relevance of this concept is already apparent from previous sections. In particular, \propref{prop:DimR=DimL} equates the dimensions of the symmetric pair (with $p' > p$) $\RadMod_{n,p'}$ and $\IrrMod_{n,p}$. In this section, we shall prove that the modules $\RadMod_{n,p'}$ and $\IrrMod_{n,p}$ of such a symmetric pair are \emph{isomorphic}. We remark that, as noted at the end of \secref{sec:Gram}, every $(n,p)$ is critical when $\beta = \pm 2$ or when $\beta = 0$ with $n$ odd.  It follows that there are no symmetric pairs in these cases.  However, for all other $\beta$ and $n$, such pairs exist whenever $n \geqslant \ell$.

A crucial role will be played in this section by the central element $F_n$ of $\tl{n}$ whose detailed properties may be found in \appref{app:Casimir}. The first part of this section shows that, even though $F_n$ is central, its action is non-diagonalisable on certain indecomposable modules. The second part then uses this fact to construct an isomorphism between the symmetric pair $\RadMod_{n,p'}$ and $\IrrMod_{n,p}$ (with $p' > p$).  The inspiration here comes from the theory of staggered modules in logarithmic conformal field theory, see \cite{RidSta09} for example.  Though the irreducibility of the radical is well-known to experts, the only proof that we are aware of \cite{GL-Aff} relies upon some rather abstract category-theoretic analysis.\footnote{Our proof is based on the existence of a non-zero homomorphism $\theta:\NatMod_{n,p}\rightarrow \NatMod_{n,p'}$ when the two modules form a symmetric pair with $p'>p$. Martin \cite{Martin} has proved that $\ker (\theta)$ and $\coker(\theta)$ are irreducible. This follows from, but is weaker than, \thmref{thm:RisL} below.} The last part describes completely the space of homomorphisms between standard modules.

\propref{lem:FEigs} states that $F_n$ acts as a multiple of the identity on the standard modules $\NatMod_{n,p}$. The next lemma provides a simple example where this action is non-diagonal. It will turn out to be key for what follows.
\begin{lemma} \label{lem:FnNonDiagonal}
Let $q$ be a root of unity other than $\pm 1$ and let $(n,p)$ be critical for this $q$ (so $n \neq 2p$). If $z_p$ denotes the $\brac{n,p}$-link state which has $p$ simple links at $1$, $3$, \ldots , $2p-1$, then when $F_{n+1}(\wun\otimes z_p)$ is expanded in the basis $\mathcal S_{n,p}$ of $\Ind{\NatMod_{n,p}}$, the coefficient of $u_1u_2\cdots u_n\otimes z_{p}=u_{2p+1}u_{2p+2}\dots u_n\otimes z_p$ does not vanish.  
\end{lemma}
\begin{proof}
We first study the case $p=0$. Since $F_{n+1}\in\tl{n+1}$, it can be written as a sum of words in reverse Jones' normal form, as in \eqnref{eqn:revTLRedWords}. For such a word to act non-trivially on $\wun\otimes z_0$, the rightmost $u_j$ must be $u_n$. So only words with a single ``flight'' $u_i u_{i+1} \cdots u_n$ will contribute, and among these, only $u_1 u_2 \cdots u_n$ will lead to the term $u_1 u_2 \cdots u_n \otimes z_0$ that we are seeking.  We therefore need to compute the coefficient of $u_1u_2\cdots u_n$ in $F_{n+1}$.

We first expand the two crossings of the top row of the diagram \eqref{eq:leFn} defining $F_{n+1}$. The points marked by dots in the following diagrams need to be linked if they are to lead to the word $u_1 u_2 \cdots u_n$.
\begin{equation}
\parbox{11mm}{\begin{center}
\includegraphics[height=14mm]{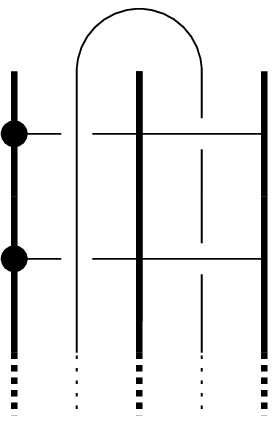}
\end{center}}
\ = \ 
- \ \parbox{11mm}{\begin{center}
\includegraphics[height=14mm]{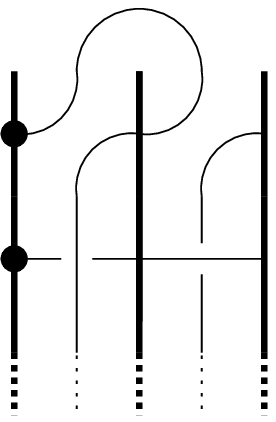}
\end{center}}
\ +q^{-1}\ 
\parbox{11mm}{\begin{center}
\includegraphics[height=14mm]{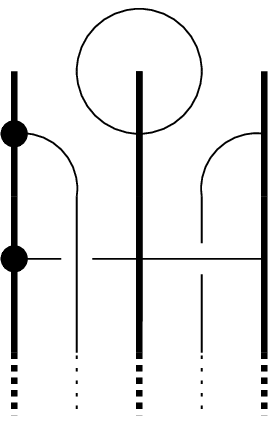}
\end{center}}
\ - \ 
\parbox{11mm}{\begin{center}
\includegraphics[height=14mm]{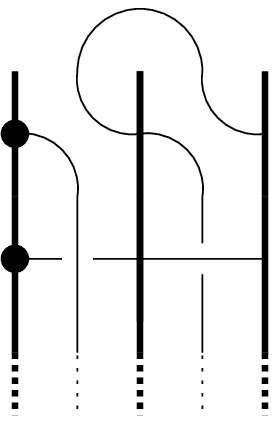}
\end{center}}
\ + q\ 
\parbox{11mm}{\begin{center}
\includegraphics[height=14mm]{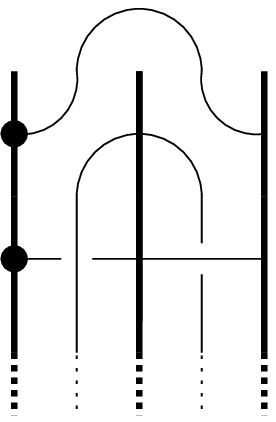}
\end{center}}
\ .
\end{equation}
The last term will not contribute to $u_1 u_2 \cdots u_n$, so we obtain
\begin{equation}
\parbox{11mm}{\begin{center}
\includegraphics[height=14mm]{tl-F71}
\end{center}}
\ \overset{\circ}{=} (-2 + \beta q^{-1}) \ 
\parbox{11mm}{\begin{center}
\includegraphics[height=14mm]{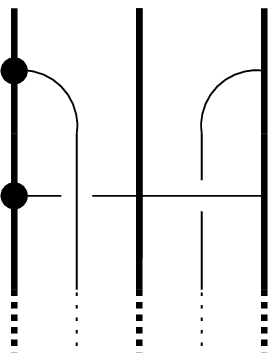}
\end{center}}
\ ,
\end{equation}
where the sign ``$\overset{\circ}{=}$'' indicates that the equality is restricted to the coefficient of $u_1 u_2 \cdots u_n$.  We now expand the second row. Only one term of the crossing on the right may contribute, as the other term can be seen to close a link on the right vertical line. But then only one term of the crossing on the left contributes to $u_1u_2\cdots u_n$. A similar argument can be repeated for all rows but the last one and the expansion of all crossings but those of the bottom row gives
\begin{equation}
\parbox{11mm}{\begin{center}
\includegraphics[height=14mm]{tl-F71}
\end{center}}
\ \overset{\circ}{=} \ (-2+\beta q^{-1})(-1)^{n-1} \ 
\parbox{11mm}{\begin{center}
\includegraphics[height=26mm]{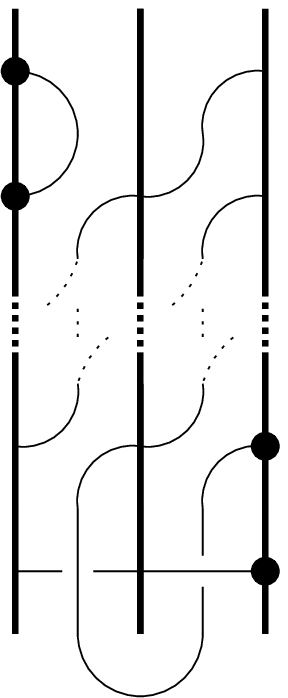}
\end{center}}
\ .
\end{equation}
Note that the top points on the left are now linked; those on the bottom right, marked by dots, still remain to be linked. Expanding the last row, we finally get that the coefficient of $u_1u_2\cdots u_n$ in $F_{n+1}$ is
\begin{equation}
(-1)^{n-1}(-2+\beta q^{-1})(-2+\beta q)=(-1)^n(\beta^2-4).
\end{equation}
Therefore, the coefficient of $u_1 u_2 \cdots u_n \otimes z_0$ in $F_{n+1}(\wun \otimes z_0)$ vanishes only when $\beta=\pm 2$ ($q=\pm 1$).

We shall now assume that $p \geqslant 1$ and consider elements in $\Ind{\NatMod_{n,p}}$ of the form $v = \wun \otimes v'$, where $v'$ is constrained to be an $\brac{n,p}$-link state with a simple link at $1$. In other words,
\begin{equation}
v = \wun \otimes \ 
\parbox{5mm}{\begin{center}
\includegraphics[height=12mm]{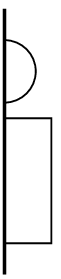}
\end{center}}
\ ,
\end{equation}
where the box stands for any $\brac{n,p-1}$-link state. We shall also make the hypothesis that $\beta \neq 0$ which allows us to use the identity
\begin{equation}\label{eq:betaNonZero}
v=\beta^{-1}u_1v
\end{equation}
(the case $\beta=0$ will be deferred until the end of the proof). We now expand the two crossings on the top right of $F_{n+1}$, extracting these crossings from the diagram representing $F_{n+1}$ and using dots to mark the points at which the subdiagrams we have extracted are to be connected:
\begin{align}
F_{n+1}v &= \ 
\parbox{11mm}{\begin{center}
\includegraphics[height=14.5mm]{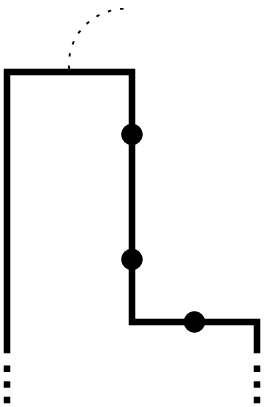}
\end{center}}
\ \times \Bigg( q \ 
\parbox{11mm}{\begin{center}
\includegraphics[height=13mm]{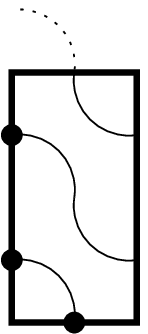}
\end{center}}
\ - \ 
\parbox{11mm}{\begin{center}
\includegraphics[height=13mm]{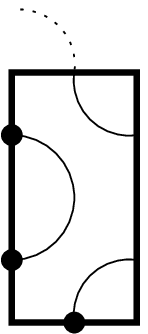}
\end{center}}
\ - \ 
\parbox{11mm}{\begin{center}
\includegraphics[height=13mm]{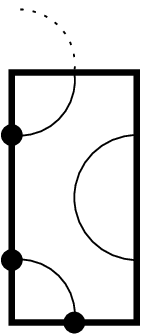}
\end{center}}
\ + q^{-1} \ 
\parbox{11mm}{\begin{center}
\includegraphics[height=13mm]{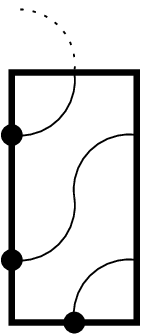}
\end{center}}
\ \Bigg)\  v \notag \\
&= \ 
\parbox{11mm}{\begin{center}
\includegraphics[height=14.5mm]{tl-F80}
\end{center}}
\ \times \Bigg( q \beta^{-1} \ 
\parbox{11mm}{\begin{center}
\includegraphics[height=13mm]{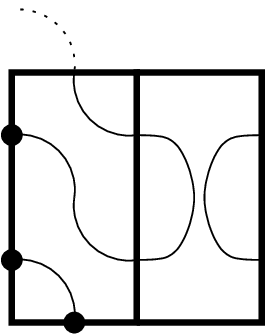}
\end{center}}
\ - \ 
\parbox{11mm}{\begin{center}
\includegraphics[height=13mm]{tl-F82}
\end{center}}
\ - \ 
\parbox{11mm}{\begin{center}
\includegraphics[height=13mm]{tl-F83}
\end{center}}
\ + q^{-1}\beta^{-1} \ 
\parbox{11mm}{\begin{center}
\includegraphics[height=13mm]{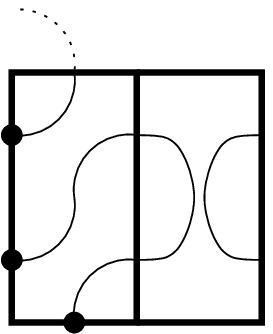}
\end{center}}
\ \Bigg)\ v \notag
\intertext{(where we have used the identity \eqref{eq:betaNonZero} on the first and last terms)}
&= \ 
\parbox{11mm}{\begin{center}
\includegraphics[height=14.5mm]{tl-F80}
\end{center}}
\ \times \Bigg( \left( (q+q^{-1}) \beta^{-1} - 1 \right) \ 
\parbox{11mm}{\begin{center}
\includegraphics[height=13mm]{tl-F83}
\end{center}}
\ - \ 
\parbox{11mm}{\begin{center}
\includegraphics[height=13mm]{tl-F82}
\end{center}}
\Bigg)\   v
= -\ 
\parbox{11mm}{\begin{center}
\includegraphics[height=13mm]{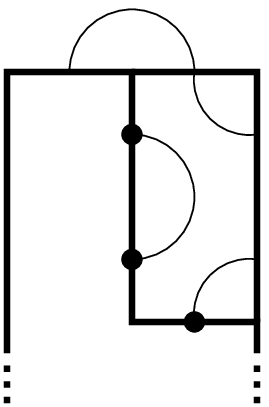}
\end{center}}\ v.
\end{align}
We now carry out the expansion of the two crossings on the left:
\begin{equation}
F_{n+1}v = \ 
\parbox{11mm}{\begin{center}
\includegraphics[height=14.5mm]{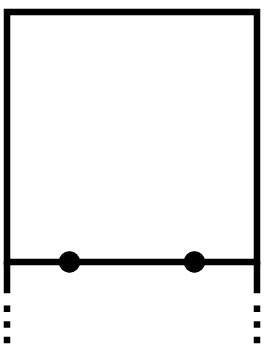}
\end{center}}
\ \times \Bigg( -q^{-1} \ 
\parbox{11mm}{\begin{center}
\includegraphics[height=13mm]{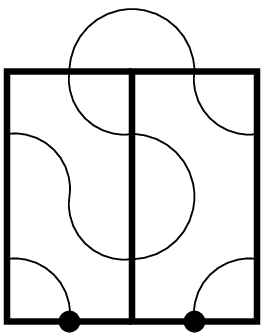}
\end{center}}
\ + \ 
\parbox{11mm}{\begin{center}
\includegraphics[height=13mm]{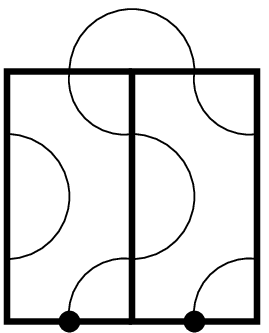}
\end{center}}
\ + \ 
\parbox{11mm}{\begin{center}
\includegraphics[height=13mm]{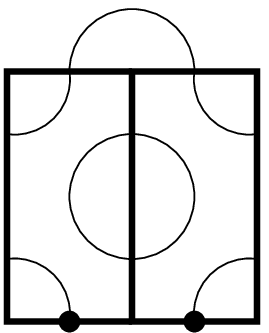}
\end{center}}
\ - q \ 
\parbox{11mm}{\begin{center}
\includegraphics[height=13mm]{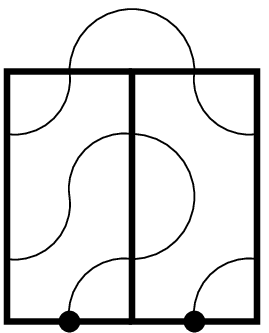}
\end{center}}
\ \Bigg) \ v.
\end{equation}
The first, third and last terms cancel and, by \eqref{eq:betaNonZero}, the second becomes
\begin{equation}
F_{n+1} v = \ 
\parbox{11mm}{\begin{center}
\includegraphics[height=14.5mm]{tl-F90}
\end{center}}
\ \times \Bigg( \ \beta^{-1} \ 
\parbox{12mm}{\begin{center}
\includegraphics[height=11mm]{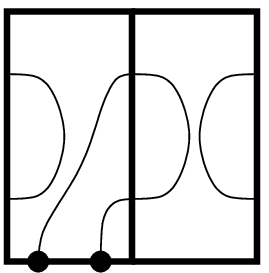}
\end{center}}
\ v \Bigg) = \beta^{-1} \ 
\parbox{11mm}{\begin{center}
\psfrag{Fn-1}[][]{$\scriptstyle{F_{n-1}}$}\includegraphics[height=13mm]{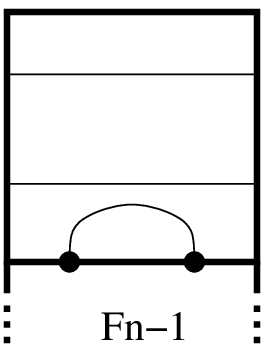}
\end{center}}
\ u_1v = \ 
\parbox{11mm}{\begin{center}
\psfrag{Fn-1}[][]{$\scriptstyle{F_{n-1}}$}\includegraphics[height=13mm]{tl-F96}
\end{center}}
\ v.
\end{equation}

Suppose now that the link state $v'$ in $v=\wun\otimes v'$ contains several consecutive simple links at the top. Then the above computation can be repeated for each of them with the result that
\begin{equation}
F_{n+1}(\wun\otimes z_p) = \ 
\parbox{15mm}{\begin{center}
\psfrag{Fn-1}[][]{$\scriptstyle{F_{n-2p+1}}$}
\psfrag{2p}[][]{$\scriptstyle{\ \ \ 2p}$}
\includegraphics[height=14mm]{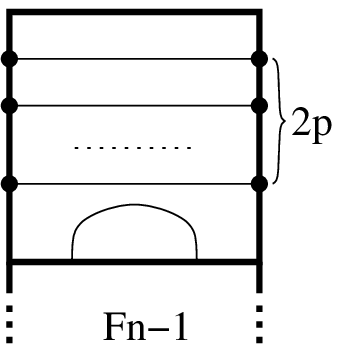}
\end{center}}
\ \ \ (\wun\otimes z_p), \qquad \text{in $\Ind{\NatMod_{n,p}}$.}
\end{equation}
The expansion of the $F_{n-2p+1}$ indicated at the bottom of this diagram will contain a term $u_{2p+1}u_{2p+2}\cdots u_n$ only if $n>2p$. In that case, we may use the previous $p=0$ computation to conclude that the coefficient of $u_1u_2\cdots u_n\otimes z_p$ in $F_{n+1}(\wun\otimes z_p)$ is $(-1)^n(\beta^2-4)$ as before. Again this coefficient does not vanish under the hypotheses of the lemma.

The case $\beta=0$, corresponding to $q=\pm i$, remains. In the link basis of the modules $\NatMod_{n,p}$, the elements of the matrices representing the generators $u_i$ are polynomials in $\beta$ or, equivalently, in $q$ and $q^{-1}$. The matrices representing the $u_i$ in the induced modules $\Ind{\NatMod_{n,p}}$ inherit this property if the basis $\mathcal S_{n,p}$ is used. Indeed, the action of $u_i$ on an arbitrary member $u_r u_{r+1} \cdots u_n \otimes d$ of $\mathcal S_{n,p}$ is easily computed:
\begin{equation}\label{eq:allCases}
u_i\brac{u_r\cdots u_n\otimes d} = 
\begin{cases}
u_r \cdots u_n \otimes u_i d, & i<r-1\\
u_{r-1} \cdots u_n \otimes d, & i=r-1\\
\beta u_r \cdots u_n \otimes d, & i=r\\
u_i \cdots u_n \otimes u_r \cdots u_{i-2} d, & r+1\leqslant i\leqslant n.
\end{cases}
\end{equation}
The right-hand sides are either elements of $\mathcal S_{n,p}$ or are equivalent to scalar multiples of such elements (using \lemref{lem:betterSpanningSet}). The central element $F_n$, defined diagrammatically by \eqref{eqn:Shorthand} and \eqref{eq:leFn}, is a linear combination of words in the generators with weights that are also polynomials in $q$ and $q^{-1}$. Therefore, the coefficients of $F_{n+1}(\wun\otimes z_p)$ in the basis $\mathcal S_{n,p}$ of $\Ind{\NatMod_{n,p}}$ are polynomials in $q$ and $q^{-1}$. Their actual values for a given $\beta$ can be obtained by evaluation of these polynomials at the corresponding value of $q$. The coefficient of $u_1 u_2 \cdots u_n \otimes z_p$ in $F_{n+1}(\wun \otimes z_p)$ when $\beta=0$ is therefore $(-1)^n(\beta^2-4)=4(-1)^{n+1}$ and is non-zero.
\end{proof}
\noindent \eqnref{eq:nonDiagF3} provides a simple example of this off-diagonal action --- there, we considered $F_3$ acting on the indecomposable module $\Ind{V_{2,0}}$ at $q=e^{\ii \pi/3}$. 

\propref{prop:HomLinks} has shown that the only homomorphism $\NatMod_{n,p'}\rightarrow\NatMod_{n,p}$ with $p'>p$ and $\langle\cdot,\cdot\rangle_{n,p}$ and $\langle\cdot,\cdot\rangle_{n,p'}\neq0$ is the zero homomorphism. The above properties of the central element $F_n$ allow us to construct a non-zero homomorphism $\NatMod_{n,p}\rightarrow\NatMod_{n,p'}$, $p'>p$, when the modules form a symmetric pair, and thereby reveal the structure of the radical $\RadMod_{n,p}$.
\begin{theorem} \label{thm:RisL}
The radical $\RadMod_{n,p}$ is zero or irreducible and, if the pair $\NatMod_{n,p}$ and $\NatMod_{n,p'}$, with $p'>p$, is symmetric, then $\RadMod_{n,p'}\cong\IrrMod_{n,p}$.
\end{theorem}
\begin{proof}
If $q$ is not a root of unity, these statements have already been proven. Indeed, the $\NatMod_{n,p}$ are then irreducible (\propref{prop:AllIrr}) and their radicals are trivial.  The same is true when $q = \pm 1$ and when $q = \pm \ii$ with $n$ odd (\corref{cor:RadCrit}). The theorem is therefore non-trivial only if $q$ is a root of unity with either $\ell \geqslant 3$ or $\ell = 2$ and $n$ even. We will therefore assume these conditions in what follows.

We first construct a non-zero homomorphism between the two standard modules $\NatMod_{n+1,p}$ and $\NatMod_{n+1,p+1}$, assuming that $(n,p)$ is critical. The symmetric pair therefore consists of $(n+1,p)$ and $(n+1,p+1)$ and we note that $f_{n+1,p}$ and $f_{n+1,p+1}$ are equal. Let $\varphi \colon \Ind{\NatMod_{n,p}} \rightarrow \Ind{\NatMod_{n,p}}$ be the map obtained from left-multiplication by $F_{n+1}-f_{n+1,p}\wun\in\tl{n+1}$. Because this element is central, $\varphi$ is a homomorphism.  Moreover, it is non-zero by \lemref{lem:FnNonDiagonal}.

Consider now the exact sequence \eqref{SES:Ind}
\begin{equation}
\dses{\NatMod_{n+1,p+1}}{\alpha}{\Ind{\NatMod_{n,p}}}{\gamma}{\NatMod_{n+1,p}}
\end{equation}
(the homomorphism $\alpha$ was constructed explicitly after \corref{cor:IndSplit}).  Since $\varphi$ acts as zero on $\NatMod_{n+1,p+1}$ (\lemref{lem:FEigs}), we have $\func{\im}{\alpha} \subseteq \func{\ker}{\varphi}$.  Similarly, $\gamma\circ\varphi = 0$ and therefore $\func{\im}{\varphi} \subseteq \func{\ker}{\gamma} = \func{\im}{\alpha}$.  Now, for any $w \in \NatMod_{n+1,p}$, we can find $v \in \Ind{\NatMod_{n,p}}$ such that $\gamma(v)=w$. Note that any other $v' \in \Ind{\NatMod_{n,p}}$ with $\gamma(v')=w$ satisfies $v-v' \in \func{\ker}{\gamma} = \func{\im}{\alpha} \subseteq \func{\ker}{\varphi}$. It follows that the map $w \mapsto \varphi(v)$ is independent of the choice of preimage $v$. This map, in turn, is a $\tl{n+1}$-homomorphism from $\NatMod_{n+1,p}$ into $\func{\im}{\alpha} \subset \Ind{\NatMod_{n,p}}$ since given any $u\in\tl{n+1}$, we may choose $uv$ as a preimage of $uw$ and $uw$ is mapped to $\varphi(uv)=u\varphi(v)$. Finally, let $i \colon \NatMod_{n+1,p} \rightarrow \NatMod_{n+1,p+1}$ be defined by $w \mapsto \brac{\alpha^{-1} \circ \varphi}(v)$. This is the non-zero $\tl{n+1}$-homomorphism that we set out to construct. The conclusion is therefore that $\dim \Hom_{\tl{n+1}}(\NatMod_{n+1,p},\NatMod_{n+1,p+1}) \geqslant 1$, when $(n,p)$ is critical.

Frobenius reciprocity (\propref{prop:Frob}) allows us to extend this result to arbitrary symmetric pairs.  Assume again that $(n,p)$ is critical. Then, the pairs $\NatMod_{n+j,p}$ and $\NatMod_{n+j,p+j}$ are symmetric for $1 \leqslant j < \ell$, where $\ell$ is the smallest positive integer such that $q^{2\ell}=1$, and any symmetric pair is of this form for some $j$ and critical $(n,p)$. Suppose now that $j\in\{2,\cdots, \ell-1\}$. We will justify each step of the following computation:
\begin{align}\label{eq:soNice}
\Hom_{\tl{n+j}}(\NatMod_{n+j,p},\NatMod_{n+j,p+j})
   &= \Hom_{\tl{n+j}}(\NatMod_{n+j,p}\oplus \NatMod_{n+j,p+1},\NatMod_{n+j,p+j})\notag\\
   &= \Hom_{\tl{n+j}}(\Ind{\NatMod_{n+j-1,p}},\NatMod_{n+j,p+j})\notag\\
   &= \Hom_{\tl{n+j-1}}(\NatMod_{n+j-1,p},\Res{\NatMod_{n+j,p+j}})\notag\\
   &= \Hom_{\tl{n+j-1}}(\NatMod_{n+j-1,p},\NatMod_{n+j-1,p+j}\oplus\NatMod_{n+j-1,p+j-1})\notag\\
   &= \Hom_{\tl{n+j-1}}(\NatMod_{n+j-1,p},\NatMod_{n+j-1,p+j-1}).
\end{align}
Recall that $f_{n,p}$ takes distinct values between any two adjacent critical lines and that $f_{n,p} \neq f_{n,p'}$ implies that $\Hom(\NatMod_{n,p}, \NatMod_{n,p'}) = \set{0}$. The first and last equalities in \eqref{eq:soNice} follow from this observation and the usual properties of $\Hom$. Since $(n,p)$ is critical and $2\leqslant j<\ell$, neither $(n+j-1,p)$ nor $(n+j,p+j)$ can be critical, so \corref{cor:IndSplit} explains the second equality and \corref{cor:ResSplit} the fourth. The third equality is Frobenius reciprocity. We therefore conclude that
\begin{equation}
\Hom_{\tl{n+j}}(\NatMod_{n+j,p},\NatMod_{n+j,p+j}) = \Hom_{\tl{n+1}}(\NatMod_{n+1,p},\NatMod_{n+1,p+1}) \neq \set{0}.
\end{equation}
In other words, there exists a non-zero homomorphism $\NatMod_{n,p}\rightarrow\NatMod_{n,p'}$ between an arbitrary symmetric pair with $p'>p$.

We are now ready to prove the statement of the theorem. Let $f\in\Hom(\NatMod_{n,p},\NatMod_{n,p'})$ be a non-zero homomorphism between a symmetric pair with $p'>p$. Then, $\func{\ker}{f}$ is a proper subset of $\NatMod_{n,p}$ and must therefore be a subset of $\RadMod_{n,p}$ by maximality of the radical. If $f$ were surjective, then $\NatMod_{n,p}/\func{\ker}f\cong \NatMod_{n,p'}$ which would contradict \propref{prop:HomLinks}. (Here, we must temporarily assume that $\beta \neq 0$ or $p' \neq n/2$ --- see the end of the present paragraph.) Thus, $\func{\im}{f}$ is a subset of $\RadMod_{n,p'}$, again by maximality. If either of $\func{\ker}{f}$ or $\func{\im}{f}$ is a proper subset of the corresponding radical, then we would have
\begin{equation}
\dim \func{\ker}{f}<\dim \RadMod_{n,p} \qquad \text{or} \qquad
\dim \func{\im}{f}<\dim \RadMod_{n,p'}
\end{equation}
which gives
\begin{equation}
\dim\NatMod_{n,p}=\dim\func{\ker}{f}+\dim\func{\im}{f}<\dim\RadMod_{n,p}+\dim\RadMod_{n,p'}.
\end{equation}
But then, the dimension of the irreducible quotient $\IrrMod_{n,p}=\NatMod_{n,p}/\RadMod_{n,p}$ would satisfy
\begin{equation}
\dim\IrrMod_{n,p}=\dim\NatMod_{n,p}-\dim\RadMod_{n,p}<\dim\RadMod_{n,p'},
\end{equation}
contradicting \propref{prop:DimR=DimL}. This proves that $\func{\ker}{f}=\RadMod_{n,p}$ and $\func{\im}{f}=\RadMod_{n,p'}$. The first isomorphism theorem now says that $\IrrMod_{n,p}=\NatMod_{n,p}/\func{\ker}f$ is isomorphic to $\func{\im}f=\RadMod_{n,p'}$ if the pair $\brac{n,p}$ and $\brac{n,p'}$, $p'>p$, is symmetric. The case $\beta=0$ and $2p'=n$ was set aside in this argument, but is in fact easier because in this case \corref{prop:V=R=L} has already established that $\RadMod_{2p',p'} = \NatMod_{2p',p'}$ is irreducible. Since $f$ is non-zero, $\func{\im}f=\RadMod_{2p',p'}$ (and $\dim\func{\im}f=\dim \RadMod_{n,p'}$). The rest of the proof is identical.

This argument proves the irreducibility of every $\RadMod_{n,p'}$ for which the image $(n,p)$ of $(n,p')$ by a reflection with respect to the critical line immediately to its right is well-defined (meaning $p\geqslant 0$). This test fails for some $(n,p')$ on the right of the second-rightmost critical line. For example, \figref{fig:IsingDims}, drawn for $\ell=4$, shows that $(4,0)$ and $(8,2)$ have no reflection with respect to the critical line on their right. For all pairs $(n,p')$ without such reflection, \propref{eqn:RecDimR} shows that the radicals attached to them are $\set{0}$.
\end{proof}

Since the only proper submodule of $\NatMod_{n,p}$ is the radical $\RadMod_{n,p}$ (and when $\beta = 0$, the exceptional case $\NatMod_{2p,p}$ is irreducible), the only quotients obtained from $\NatMod_{n,p}$ are the modules $\{0\}$, $\NatMod_{n,p}/\RadMod_{n,p}\cong\IrrMod_{n,p}$ and $\NatMod_{n,p}$ itself. Of course some of these may coincide. Therefore:
\begin{corollary} \label{cor:indecoOfQuotient}
Every quotient of a standard module is indecomposable.
\end{corollary}

The previous theorem also reveals the structure of the standard modules in terms of its composition factors.

\begin{corollary}\label{cor:2013}
With exception of the (irreducible) $\NatMod_{2p,p}$ when $\beta=0$, the standard module
$\NatMod_{n,p}$ is reducible (but indecomposable) if and only if $(n,p)$ forms a symmetric pair with $(n,p')$, where $p>p'$, and then the sequence
\begin{equation}\label{eq:2013}
0\longrightarrow \IrrMod_{n,p'}\longrightarrow \NatMod_{n,p}\longrightarrow \IrrMod_{n,p}\longrightarrow 0
\end{equation}
is exact and non-split. The reducible standard module $\NatMod_{n,p}$ therefore has two composition factors, $\IrrMod_{n,p'}$ and $\IrrMod_{n,p}$, and its (unique) composition series is $0\subset \IrrMod_{n,p'}\subset \NatMod_{n,p}$.
\end{corollary}
\begin{proof} The standard module $\NatMod_{n,p}$ is reducible if and only if it has a non-trivial submodule. This submodule is then its radical $\RadMod_{n,p}$ which, by \thmref{thm:RisL}, is isomorphic to $\IrrMod_{n,p'}$, the pair $(n,p)$ and $(n,p')$, with $p>p'$, being symmetric. The indecomposability of $\NatMod_{n,p}$ was proved in \propref{prop:S/R}, so the short exact sequence \eqref{eq:2013} cannot split. The rest of the statement follows from the definition of composition series (see \appref{app:Review}).
\end{proof}

\noindent We remark that the sequence \eqref{eq:2013} is still exact in the exceptional case $\beta = 0$ and $n=2p$ because we have defined $\IrrMod_{2p,p}$ to be $0$ in this case (see the discussion after \propref{prop:S/R}).

It is natural to decompose the set of allowed $p$ ($0\leqslant p\leqslant \lfloor n/2\rfloor$) into {\em orbits} under reflection about a critical line.\footnote{
These orbits were introduced in \cite{GoodWenzl93} and play a central role in the construction of the blocks of $\tl{n}$.
} Let $\brac{n,p_1}$ lie to the left of the first critical line of the Bratteli diagram and let $p_1 > p_2 > \cdots > p_m \geqslant 0$ be the indices obtained from $p_1$ by reflecting across the critical lines (so there is precisely one $p_i$ between each pair of consecutive critical lines).  The set $\set{\brac{n,p_i} \st 1 \leqslant i \leqslant m}$ is called the \emph{orbit} of $\brac{n,p_1}$ under these reflections and every non-critical $\brac{n,p}$ belongs to a unique orbit. (One may complete this definition by adding that a critical $(n,p)$ is alone in its orbit.) The notation $k(n,p)$ and $r(n,p)$ introduced in \eqref{eq:kANDr} allows for recursive expressions for the $p_i$ of a given non-critical orbit. Since, on line $n$ of the Bratteli diagram, the two $p_i$ and $p_{i+1}$ are separated by a single critical line and lie symmetrically on each of its side, their labels are related by
\begin{equation}\label{eq:krOrbit}
k(n,p_{i+1})=k(n,p_i)+1\qquad\textrm{and}\qquad r(n,p_{i+1})=\ell-r(n,p_i).
\end{equation}
Since $n-2p+1=k(n,p)\ell+r(n,p)$, one finds that
\begin{equation}\label{eq:LackOfImagination}
p_{i+1}=p_i+r(n,p_i)-\ell\qquad\textrm{and}\qquad p_{i-1}=p_i+r(n,p_i).
\end{equation}
\appref{app:Casimir} has shown that the central element $F_n$ takes distinct eigenvalues on (distinct) standard modules whose labels fall between two consecutive critical lines. This observation leads to another definition of non-critical orbits: Two labels $p,p'$ belong to the same (non-critical) orbit if and only if $f_{n,p} = f_{n,p'}$. With the relations above, the equivalence between the two definitions is easily established using $q^{2\ell}=1$:
\begin{align} \label{eqn:FOrbitEigs}
f_{n,p_{i-1}}
&=q^{n-2p_{i-1}+1} + q^{-(n-2p_{i-1}+1)}
=q^{k(n,p_{i-1}) \ell + r(n,p_{i-1})} + q^{-k(n,p_{i-1}) \ell - r(n,p_{i-1})} \notag \\
&=q^{(k(n,p_i)-1) \ell + \ell - r(n,p_i)} + q^{-(k(n,p_i)-1) \ell - \ell + r(n,p_i)}
=q^{k(n,p_i) \ell - r(n,p_i)} + q^{-k(n,p_i) \ell + r(n,p_i)} \notag \\
&=q^{2k(n,p_i)\ell} q^{-(k(n,p_i)\ell + r(n,p_i))}+q^{-2k(n,p_i)\ell} q^{k(n,p_i)\ell+r(n,p_i)} \notag \\
&=f_{n,p_i}.
\end{align}

\medskip

Frobenius reciprocity was a key ingredient in the proof of \thmref{thm:RisL}. It allowed us to deduce the existence of non-trivial homomorphisms between certain standard modules $\NatMod_{n,p}$ from the explicit construction of such a homomorphism in the simplest case. This underlines also the usefulness of a central element, in the present case $F_n$, that acts non-diagonalisably. A natural question which remains is whether these non-trivial homomorphisms are the only ones.  We conclude the section by settling this in the affirmative.

\begin{theorem} \label{thm:dimHom}
The dimension $\dim\Hom(\NatMod_{n,p},\NatMod_{n,p'})$ is $1$ if $p=p'$ or if the two standard modules form a symmetric pair with $p'>p$. Moreover, there is a single exceptional case: $\dim\Hom(\NatMod_{2,1},\NatMod_{2,0})=1$ when $\beta=0$. Otherwise, $\dim\Hom(\NatMod_{n,p},\NatMod_{n,p'}) = 0$.
\end{theorem}
\begin{proof}
We first omit, for $\beta=0$ and $n$ even, the study of $\Hom(\NatMod_{2p,p},\NatMod_{2p,p'})$ and $\Hom(\NatMod_{2p,p'},\NatMod_{2p,p})$. This omission makes available all results from \secref{sec:Reps} that call for the bilinear form $\inner{\cdot}{\cdot}$ to be non-zero. Then the homomorphisms $\NatMod_{n,p}\rightarrow \NatMod_{n,p'}$ have been determined by \propref{prop:HomLinks}, in the case $p>p'$, and \propref{prop:EndLinks}, in the case $p=p'$. We also know that $\Hom(\NatMod_{n,p},\NatMod_{n,p'}) = \set{0}$ when $f_{n,p} \neq f_{n,p'}$.  The outstanding cases are therefore:
\begin{enumerate}
\item $\brac{n,p}$ and $\brac{n,p'}$ are both critical with $p'>p$ and $f_{n,p} = f_{n,p'}$.
\item $\brac{n,p}$ and $\brac{n,p'}$ form a symmetric pair with $p'>p$. \label{Case2}
\item $\brac{n,p}$ and $\brac{n,p'}$ are non-critical with $p'-p \geqslant \ell$ and $f_{n,p} = f_{n,p'}$. \label{Case3}
\end{enumerate}

Case (1) is easily dealt with:  If two distinct $(n,p)$ and $(n,p')$ are critical, then $\NatMod_{n,p}$ and $\NatMod_{n,p'}$ are irreducible (\corref{cor:RadCrit}) and it follows from \corref{cor:VLDistinct} that $\Hom(\NatMod_{n,p},\NatMod_{n,p'})=0$.  In case (2), we know from \thmref{thm:RisL} that $\RadMod_{n,p'} \cong \IrrMod_{n,p}$.  Schur's lemma then implies that any two homomorphisms from $\IrrMod_{n,p}$ to $\RadMod_{n,p'}$ are equal to one another up to a multiplicative constant. Choose two non-zero homomorphisms $f,g\in\Hom(\NatMod_{n,p},\NatMod_{n,p'})$. The argument used to prove \thmref{thm:RisL} implies that we must necessarily have $\func{\ker}{f} = \func{\ker}{g}= \RadMod_{n,p}$ and $\func{\im}{f} = \func{\im}{g} = \RadMod_{n,p'}$. Both $f$ and $g$ then define homomorphisms $\hat{f}, \hat{g} \colon \IrrMod_{n,p} \rightarrow \RadMod_{n,p'}$ and so we learn that $\hat{f} = \mu \hat{g}$ for some $\mu \in \CC$. As
\begin{equation}
f(v) = \brac{\hat{f}\circ\pi}(v) = \mu \brac{\hat{g}\circ\pi}(v) = \mu g(v) \qquad \text{for all $v\in\NatMod_{n,p}$,}
\end{equation}
where $\pi$ is the projection $\NatMod_{n,p}\rightarrow \IrrMod_{n,p}$, we conclude that $f=\mu g$, hence that $\dim\Hom(\NatMod_{n,p},\NatMod_{n,p'})=1$.

We therefore turn to case (3). Let $(n,p'')$ be the reflection of $(n,p')$ in the critical line immediately to its right, so that $\NatMod_{n,p'}$ and $\NatMod_{n,p''}$ form a symmetric pair with $p'>p''>p$. Let $f \colon \NatMod_{n,p}\rightarrow\NatMod_{n,p'}$ be a non-zero homomorphism, so $\func{\ker}{f} \subseteq \RadMod_{n,p}$. But, $\RadMod_{n,p}$ is irreducible and therefore $\func{\ker}{f}$ is either $\set{0}$ or $\RadMod_{n,p}$. Similarly, $\func{\im}{f}$ must be either $\RadMod_{n,p'}$ or $\NatMod_{n,p'}$. The first isomorphism theorem, together with \corref{cor:VLDistinct}, now leads to a contradiction in each of the four possible combinations. For example:
\begin{itemize}[leftmargin=*]
\item If $\func{\ker}{f}=\set{0}$ and $\func{\im}{f}=\NatMod_{n,p'}$, then $\NatMod_{n,p} \cong \func{\im}{f} = \NatMod_{n,p'}$ with $p\neq p'$, contradicting \corref{cor:VLDistinct}.
\item If $\func{\ker}{f}=\set{0}$ and $\func{\im}{f}=\RadMod_{n,p'}$, then $\NatMod_{n,p} \cong \func{\im}{f} = \RadMod_{n,p'} \cong \IrrMod_{n,p''}$. But, $\NatMod_{n,p}$ being irreducible implies that $\IrrMod_{n,p} \cong \IrrMod_{n,p''}$ with $p\neq p''$, again contradicting \corref{cor:VLDistinct}.
\end{itemize}
The last two cases $\func{\ker}{f} = \RadMod_{n,p}$, $\func{\im}{f} = \NatMod_{n,p'}$ and $\func{\ker}{f} = \RadMod_{n,p}$, $\func{\im}{f} = \RadMod_{n,p'}$ are similar and are left to the reader. We conclude that $\Hom(\NatMod_{n,p},\NatMod_{n,p'})=\set{0}$ in all these cases.

It only remains to consider the omitted case involving $\NatMod_{2p,p}=\RadMod_{2p,p}$ at $\beta=0$. In this case, every $(2p,p')$ is non-critical (see \figref{fig:Bratteli}) and shares the same $F_{2p}$-eigenvalue:  $f_{2p,p'}=0$. $\NatMod_{2p,p'}$ and $\NatMod_{2p,p''}$ therefore form a symmetric pair if and only if $\abs{p'-p''}=1$. \corref{cor:VLDistinct} can be used to conclude that the $\IrrMod_{2p,p'}$ with $0 \leqslant p' \leqslant p-1$ are mutually non-isomorphic. Because $\IrrMod_{2p,p'}\cong\RadMod_{2p,p'+1}$, the radicals $\RadMod_{2p,p'}$ with $1\leqslant p'\leqslant p$ are likewise mutually non-isomorphic. Finally, the $\RadMod_{2p,p'}=\IrrMod_{2p,p'-1}$ are non-zero for $1\leqslant p'< p$, so the only irreducible modules among the $\NatMod_{2p,p'}$ are therefore $\NatMod_{2p,p}$ (by \propref{prop:V=R=L}) and $\NatMod_{2p,0}$.

We now look for a non-zero homomorphism $f \colon \NatMod_{2p,p}\rightarrow \NatMod_{2p,p'}$ with $p'<p$. If such a homomorphism exists, then $\func{\ker}{f}=\set{0}$ and $\func{\im}{f}\cong \NatMod_{2p,p}$ (because the latter is irreducible). It follows that either $\NatMod_{2p,p}\cong\RadMod_{2p,p'}$ or $\NatMod_{2p,p}\cong\NatMod_{2p,p'}$. The first possibility is excluded by the previous observation that the radicals $\RadMod_{2p,p'}$, $1\leqslant p'\leqslant p$, are mutually distinct as well as recalling that $\RadMod_{2p,p} = \NatMod_{2p,p} \neq \set{0} = \RadMod_{2p,0}$. We are therefore left with the second possibility, that $\NatMod_{2p,p'}$ is itself irreducible, which can only happen when $p'=0$. But, $\dim \NatMod_{2p,0}=1$ only coincides with $\dim\NatMod_{2p,p}$ when $p=1$. So we arrive at the exceptional case $\Hom(\NatMod_{2,1},\NatMod_{2,0})\cong\CC$ (the existence of such a non-zero homomorphism was already established after \corref{cor:VLDistinct}). 

Finally, the pair $\NatMod_{2p,p-1}$ and $\NatMod_{2p,p}$ is symmetric and the same arguments that proved (\ref{Case2}) above lead to $\dim \Hom(\NatMod_{2p,p-1},\NatMod_{2p,p}) = 1$. The proof that $\dim \Hom(\NatMod_{2p,p'},\NatMod_{2p,p})=0$ for $p'\leqslant p-2$ mimics that of (\ref{Case3}) with the simplification that now $\func{\im}{f} = \NatMod_{2p,p}$ since $\NatMod_{2p,p}$ is irreducible, hence there are only two cases to consider instead of four.
\end{proof}

%
%

\section{Principal Indecomposable Modules} \label{sec:Proj}

In this section, we describe a concrete construction, for $q$ a root of unity, of the \emph{principal indecomposable} modules as submodules of certain induced modules.  These are the indecomposable direct summands of the Temperley-Lieb algebra when one treats it as a module by letting it act on itself by left-multiplication (this is just the regular representation).  It follows immediately that principal indecomposables are canonical examples of \emph{projective} modules.  Indeed, they are precisely the indecomposable projective modules.  Moreover, a standard fact about them \cite{CurRep62} is that there is a bijective correspondence between principal indecomposables and irreducibles given by quotienting the former by its (unique) maximal proper submodule (its radical).  We will therefore denote a principal indecomposable module by $\ProjMod_{n,p}$, understanding that its irreducible quotient is $\IrrMod_{n,p}$.\footnote{In the exceptional case, $\beta = 0$ and $n=2p$, we recall that $\IrrMod_{2p,p} = \set{0}$ and the irreducible quotient of $\NatMod_{2p,p}$ is in fact $\NatMod_{2p,p} = \RadMod_{2p,p} \cong \IrrMod_{2p,p-1}$.  There is therefore no need to define a non-trivial principal indecomposable $\ProjMod_{2p,p}$ when $\beta = 0$.} The properties of projective modules needed for this section are reviewed in \appref{app:Review}.

Note that when the algebra is semisimple, indecomposability implies irreducibility, so
\begin{equation} \label{eqn:PNI}
\ProjMod_{n,p} = \NatMod_{n,p} = \IrrMod_{n,p} \qquad \text{($\tl{n}$ semisimple).}
\end{equation}
We remark that because $\tl{n}$ is semisimple for generic $\beta$ (and its dimension is obviously independent of $\beta$), it follows from \thmDref{thm:W}{thm:GW} that
\begin{equation} \label{eqn:LP=VV}
\sum_i \dim \IrrMod_i \: \dim \ProjMod_i = \dim \tl{n} = \sum_{p=0}^{\flr{n/2}} \dim \NatMod_{n,p} \: \dim \NatMod_{n,p},
\end{equation}
where the sum on the left-hand side is over a {\em complete set} of pairwise non-isomorphic irreducibles $\IrrMod_i$ and the corresponding principal indecomposables $\ProjMod_i$.  We recall from \corref{cor:RootOfUnity} that the sum on the right-hand side corresponds to a complete set of non-isomorphic irreducibles, when $q$ is not a root of unity.  At this point, we know that the $\IrrMod_{n,p}$ are pairwise non-isomorphic, but we are not assured that they form a complete set when $q$ is a root of unity, excluding $\IrrMod_{2p,p} = \set{0}$ from the set when $\beta = 0$.  This completeness will be deduced as a corollary of the principal indecomposable analysis.

We will first show how one can concretely construct the principal indecomposables with the aid of a detailed example.  This serves to nicely illustrate the salient features of the general discussion (which is obscured somewhat by the necessary induction arguments).  So let us take $q = e^{\ii \pi / 4}$ ($\beta = \sqrt{2}$), noting that criticality corresponds to $n-2p+1$ being a multiple of $\ell = 4$.  It may be helpful to recall that the dimensions of the irreducible modules for this $\beta$ and small $n$ were organised in a Bratteli diagram in \figref{fig:IsingDims}.

Consider first the rather trivial algebra $\tl{1}$.  Since it is one-dimensional and spanned by the unit $\wun$, its left regular representation is isomorphic to its only standard module $\NatMod_{1,0}$.  Thus,
\begin{equation}
\tl{1} = \ProjMod_{1,0} \cong \NatMod_{1,0}.
\end{equation}
Now apply the induced module construction to get $\tl{2} \otimes_{\tl{1}} \tl{1} \cong \Ind{\NatMod_{1,0}}$ as $\tl{2}$-modules.  Since $\tl{n+1} \otimes_{\tl{n}} \tl{n}$ is obviously isomorphic to $\tl{n+1}$ (just send $a \otimes b = ab \otimes \wun$ to $ab$), we obtain
\begin{equation} \label{eqn:LRR2}
\tl{2} \cong \Ind{\NatMod_{1,0}} \cong \NatMod_{2,1} \oplus \NatMod_{2,0}.
\end{equation}
Here, we have applied \corref{cor:IndSplit}, using the fact that $\brac{1,0}$ is not critical, because $\beta \neq 0$.  Note that both $\NatMod_{2,1}$ and $\NatMod_{2,0}$ are indecomposable, by \propref{prop:S/R}, and are direct summands of the left regular representation.  They are therefore principal indecomposable modules:  $\ProjMod_{2,1} \cong \NatMod_{2,1}$ and $\ProjMod_{2,0} \cong \NatMod_{2,0}$.

For $n=3$, we similarly apply the induced module construction to \eqnref{eqn:LRR2}, obtaining
\begin{equation}
\tl{3} \cong \brac{\tl{3} \otimes_{\tl{2}} \NatMod_{2,1}} \oplus \brac{\tl{3} \otimes_{\tl{2}} \NatMod_{2,0}} \cong \Ind{\NatMod_{2,1}} \oplus \Ind{\NatMod_{2,0}},
\end{equation}
since tensor product distributes over direct sums.  Now, $\brac{2,0}$ is only critical for $\beta = \pm 1$, hence \corref{cor:IndSplit} applies, giving $\Ind{\NatMod_{2,0}} \cong \NatMod_{3,1} \oplus \NatMod_{3,0}$.  Since $\Ind{\NatMod_{2,1}} \cong \NatMod_{3,1}$, we conclude that $\ProjMod_{3,1} \cong \NatMod_{3,1}$ and $\ProjMod_{3,0} \cong \NatMod_{3,0}$.  Note that this implies that $\tl{3} = 2 \: \ProjMod_{3,1} \oplus \ProjMod_{3,0}$, consistent with \thmref{thm:GW} as $\dim \IrrMod_{3,1} = 2$ and $\dim \IrrMod_{3,0} = 1$.

We pause to remark that for all $n$ considered thus far, we have demonstrated the decomposition 
\begin{equation}
\tl{n} \cong \bigoplus_{p=0}^{\flr{n/2}} \brac{\dim \IrrMod_{n,p}} \: \ProjMod_{n,p}.
\end{equation}
We may therefore deduce from \thmref{thm:GW} that the set of $\IrrMod_{n,p}$ with $0 \leqslant p \leqslant \flr{n/2}$ constitutes a complete set of pairwise non-isomorphic irreducibles, at least for $n \leqslant 3$ and $\beta=\sqrt2$.

At some point, criticality must enter the fray.  For $\beta = \sqrt{2}$, this first occurs for $n=4$.  Inducing our decomposition of $\tl{3}$ and analysing as above, we easily conclude that
\begin{equation}
\tl{4} \cong 2 \: \NatMod_{4,2} \oplus 2 \: \NatMod_{4,1} \oplus \Ind{\NatMod_{3,0}} ,
\end{equation}
hence that $\ProjMod_{4,2} \cong \NatMod_{4,2}$ and $\ProjMod_{4,1} \cong \NatMod_{4,1}$.  As $\dim \IrrMod_{4,2} = \dim \IrrMod_{4,1} = 2$ and $\dim \IrrMod_{4,0} = 1$ (see \figref{fig:IsingDims}), \thmref{thm:GW} requires that
\begin{equation}
\tl{4} \cong 2 \: \ProjMod_{4,2} \oplus 2 \: \ProjMod_{4,1} \oplus \ProjMod_{4,0} \oplus \cdots ,
\end{equation}
where the ``$\oplus \cdots$'' admits that there may be additional unknown principal indecomposables beyond the $\ProjMod_{4,p}$.\footnote{Such additional principal indecomposables would correspond to additional irreducibles.  The existence of these would therefore amount to a negative answer to the question of the completeness of the set of $\IrrMod_{4,p}$.}  Comparing gives $\Ind{\NatMod_{3,0}} \cong \ProjMod_{4,0} \oplus \cdots.$  Now, \corref{cor:IndSplit} does not help us to simplify $\Ind{\NatMod_{3,0}}$ because $\brac{3,0}$ is critical.  Instead, \corref{cor:Induct} gives the exact sequence
\begin{equation} \label{es:P40}
\dses{\NatMod_{4,1}}{}{\Ind{\NatMod_{3,0}}}{}{\NatMod_{4,0}}.
\end{equation}
We note that the composition factors of $\NatMod_{4,0}$ and $\NatMod_{4,1}$ are $\IrrMod_{4,1}$, $\IrrMod_{4,0}$ and $\IrrMod_{4,0}$, and that these are also the composition factors of $\Ind{\NatMod_{3,0}}$ by the exactness of \eqref{es:P40}.  If $\Ind{\NatMod_{3,0}} \ncong \ProjMod_{4,0}$, meaning that there exist additional unknown principal indecomposables (the ``$\oplus \cdots$''), then the additional principal indecomposables, and hence $\Ind{\NatMod_{3,0}}$, will have quotients, hence composition factors, which are not irreducibles of the form $\IrrMod_{4,p}$.  This contradicts our conclusion that the composition factors of $\Ind{\NatMod_{3,0}}$ all have the form $\IrrMod_{4,p}$, so it follows that these additional principal indecomposables do not exist and that we may identify $\ProjMod_{4,0}$ with $\Ind{\NatMod_{3,0}}$.  Consequently, the $\IrrMod_{4,p}$ furnish a complete set of pairwise non-isomorphic irreducibles.

Continuing, we induce our $n=4$ decomposition and argue as above to get
\begin{equation} \label{eqn:LRR5}
\tl{5} \cong 4 \: \ProjMod_{5,2} \oplus 2 \: \ProjMod_{5,1} \oplus \Ind{\ProjMod_{4,0}} ,
\end{equation}
with $\ProjMod_{5,2} \cong \NatMod_{5,2}$ and $\ProjMod_{5,1} \cong \NatMod_{5,1}$.  As $\dim \IrrMod_{5,2} = \dim \IrrMod_{5,1} = 4$ and $\dim \IrrMod_{5,0} = 1$, \thmref{thm:GW} yields the decomposition 
\begin{equation}\label{eq:p40}
\Ind{\ProjMod_{4,0}} \cong 2 \: \ProjMod_{5,1} \oplus \ProjMod_{5,0} \oplus \cdots,
\end{equation}
where we again allow the possibility of additional unknown principal indecomposables.  We therefore need to analyse the structure of $\Ind{\ProjMod_{4,0}}$.  For this, we induce each module of the exact sequence \eqref{es:P40}.  Because induction is right-exact (\propref{prop:IndRExact}), we obtain the exact sequence
\begin{equation} \label{es:IndP40}
\rdses{\Ind{\NatMod_{4,1}} \cong \NatMod_{5,2} \oplus \NatMod_{5,1}}{}{\Ind{\ProjMod_{4,0}}}{}{\Ind{\NatMod_{4,0}} \cong \NatMod_{5,1} \oplus \NatMod_{5,0}}
\end{equation}
which we have simplified using \corref{cor:IndSplit}.  Note that the central element $F_5$ has eigenvalue $0$ on $\NatMod_{5,2}$ and $\NatMod_{5,0}$, but eigenvalue $-2$ on $\NatMod_{5,1}$. Any homomorphism $A\rightarrow B$ will map generalised eigenspaces of $A$ into generalised eigenspaces of $B$ with the same eigenvalue $f$ of $F_n$. It follows that $\Ind{\ProjMod_{4,0}}$ decomposes into projectives as $\ProjMod \oplus \ProjMod'$ (see \propref{prop:ProjSum}) along the generalised eigenspaces of $F_5$ and that the exact sequence \eqref{es:IndP40} is equivalent to the exactness of
\begin{equation} \label{ess:P5p}
\rdses{\NatMod_{5,1}}{\iota}{\ProjMod}{}{\NatMod_{5,1}} \qquad \text{and} \qquad \rdses{\NatMod_{5,2}}{\iota'}{\ProjMod'}{}{\NatMod_{5,0}}.
\end{equation}

Since $\ProjMod_{5,1} \cong \NatMod_{5,1}$ is projective, the first exact sequence splits, hence $\ProjMod \cong \ProjMod_{5,1} \oplus \brac{\ProjMod_{5,1} / \ker \iota}$.  \eqnref{eq:p40} shows that $\Ind{\ProjMod_{4,0}}$ must contain two copies of $\ProjMod_{5,1}$ and these can come only from $\ProjMod$, since $F_n$ has distinct eigenvalues on $\ProjMod_{5,1}$ and $\ProjMod'$. This forces $\ker \iota = 0$, $\ProjMod \cong 2 \: \ProjMod_{5,1}$ and, thus, $\ProjMod' = \ProjMod_{5,0} \oplus \cdots$.  Once again, we rule out the existence of additional principal indecomposables by noting that the composition factors of both $\NatMod_{5,0}$ and $\mathcal{W} = \NatMod_{5,2} / \ker \iota'$ have the form $\IrrMod_{5,p}$.  The exactness of
\begin{equation} \label{es:W5}
\dses{\mathcal{W}}{}{\ProjMod'}{}{\NatMod_{5,0}}
\end{equation}
then finishes the job:  $\ProjMod' \cong \ProjMod_{5,0}$, so $\Ind{\ProjMod_{4,0}} \cong 2 \: \ProjMod_{5,1} \oplus \ProjMod_{5,0}$ and the $\IrrMod_{5,p}$ form a complete set of irreducibles.  It only remains to determine the module $\mathcal{W}$.  This follows quickly from \eqnref{eqn:LP=VV} by comparing dimensions:
\begin{subequations}
\begin{gather}
4 \dim \ProjMod_{5,2} + 4 \dim \ProjMod_{5,1} + \dim \ProjMod_{5,0} = 5 \dim \NatMod_{5,2} + 4 \dim \NatMod_{5,1} + \dim \NatMod_{5,0} \\
\Rightarrow \qquad \dim \ProjMod_{5,0} = \dim \NatMod_{5,2} + \dim \NatMod_{5,0} \qquad \Rightarrow \qquad \mathcal{W} = \NatMod_{5,2}.
\end{gather}
\end{subequations}
We therefore see that $\ker \iota' = 0$ and so we obtain the non-split exact sequence $\ses{\NatMod_{5,2}}{\ProjMod_{5,0}}{\NatMod_{5,0}}$.  This completes our analysis for $n=5$.

One can continue to analyse this example for higher $n$ with few further difficulties.  However, we have seen enough tricks by now to understand the principal indecomposables in general.  
\begin{theorem} \label{thm:IdentP}
Let $q$ be a root of unity and let $\ell$ be the minimal positive integer satisfying $q^{2 \ell} = 1$.  Define $\func{k}{n,p} \in \NN$ and $\func{r}{n,p} \in \set{1,\ldots,\ell-1,\ell}$ by $n-2p+1 = \func{k}{n,p} \ell + \func{r}{n,p}$.  Then, the principal indecomposables $\ProjMod_{n,p}$ of $\tl{n}$ are identified as follows:
\begin{itemize}[leftmargin=*]
\item If $\func{r}{n,p} = \ell$ (so $\brac{n,p}$ is critical), then $\ProjMod_{n,p} \cong \NatMod_{n,p}$.
\item If $\func{k}{n,p} = 0$ (so $\brac{n,p}$ lies to the left of the first critical line) and $\beta \neq 0$, then $\ProjMod_{n,p} \cong \NatMod_{n,p}$.
\item If $\func{k}{n,p} > 0$ and $\func{r}{n,p} \neq \ell$, then $\ProjMod_{n,p}$ is the direct summand of the $\func{r}{n,p}$-fold induced module
\[ \NatMod_{n-\func{r}{n,p},p} \underset{\text{$\func{r}{n,p}$ times}}{\underbrace{\uparrow \cdots \uparrow}} \]
consisting of the generalised eigenspace, under the action of $F_n$, whose generalised eigenvalue is $f_{n,p} = q^{n-2p+1} + q^{-\brac{n-2p+1}}$.  Furthermore, there is a non-split exact sequence
\begin{equation}\label{eq:THEes}
\dses{\NatMod_{n,p+\func{r}{n,p}}}{}{\ProjMod_{n,p}}{}{\NatMod_{n,p}}.
\end{equation}
\end{itemize}
Moreover, $\set{\IrrMod_{n,p} \st p=0,1,\ldots,\flr{n/2}}$ is a complete set of pairwise non-isomorphic irreducibles, except when $n$ is even and $\beta = 0$.  In this latter case, the range must be restricted to $p=0,1,\ldots, n/2-1$.
\end{theorem}
\noindent The proof, which constitutes the remainder of this section, is by induction on $n$.  As in the $\ell = 4$ example detailed above, the key is to note that applying the induced module construction to a projective $\tl{n-1}$-module results in a projective $\tl{n}$-module.  The labour mostly concerns keeping track of the multiplicities with which these projectives appear. Before commencing the induction arguments, it is convenient to deal with the exceptional cases $\beta = 0, \pm 2$. Indeed, the above statement makes it clear that the case $\beta=0$ and $n$ even is special, since there is now no principal indecomposable $\ProjMod_{2p,p}$, the number of distinct principal indecomposables being then one less.
\begin{proposition} \label{prop:IdentPB=0}
When either $\beta = \pm 2$ ($\ell = 1$), or $n$ is odd and $\beta = 0$ ($\ell = 2$), the principal indecomposables of $\tl{n}$ may be identified as $\ProjMod_{n,p} \cong \NatMod_{n,p}\cong\IrrMod_{n,p}$, for $p = 0, 1, \ldots, \flr{n/2}$.  Moreover, the $\IrrMod_{n,p}$ form a complete set of pairwise non-isomorphic irreducibles for $0 \leqslant p \leqslant \flr{n/2}$.

When $n$ is even and $\beta = 0$, we have instead $\ProjMod_{n,p} \cong \Ind{\NatMod_{n-1,p}}$, for $p = 0, 1, \ldots, n/2 - 1$.  The $\IrrMod_{n,p}$ therefore form a complete set of pairwise non-isomorphic irreducibles for $0 \leqslant p \leqslant n/2 - 1$.
\end{proposition}
\begin{proof}
When $\beta = \pm 2$ or $n$ is odd and $\beta = 0$, $\tl{n}$ is semisimple (as remarked after \corref{cor:RadCrit}), hence $\ProjMod_{n,p} \cong \NatMod_{n,p} = \IrrMod_{n,p}$ for all $p$.  So, assume that $n$ is even and $\beta = 0$.  Then, $\ell = 2$ and the semisimplicity of $\tl{n-1}$ gives
\begin{equation}
\tl{n-1} \cong \bigoplus_{p=0}^{n/2-1} \brac{\dim \IrrMod_{n-1,p}} \NatMod_{n-1,p},
\end{equation}
by \thmref{thm:W}, hence
\begin{equation} \label{eqn:TLDecompB=0}
\tl{n} \cong \bigoplus_{p=0}^{n/2-1} \brac{\dim \IrrMod_{n-1,p}} \Ind{\NatMod_{n-1,p}} = \bigoplus_{p=0}^{n/2-1} \brac{\dim \IrrMod_{n,p}} \Ind{\NatMod_{n-1,p}}.
\end{equation}
Here, we have used \corref{cor:DimL}, noting that $\func{r}{n,p} = 1 = \ell-1$.

We see that $\Ind{\NatMod_{n-1,p}}$ is a projective $\tl{n}$-module, for $0 \leqslant p \leqslant n/2-1$, hence it may be written as a direct sum of principal indecomposables.  Because of the exact sequence
\begin{equation}
\dses{\NatMod_{n,p+1}}{}{\Ind{\NatMod_{n-1,p}}}{}{\NatMod_{n,p}},
\end{equation}
we know that $\IrrMod_{n,p}$ is an irreducible quotient of $\Ind{\NatMod_{n-1,p}}$, hence (at least) one of these principal indecomposables is $\ProjMod_{n,p}$.  Thus, $\Ind{\NatMod_{n-1,p}} \cong \ProjMod_{n,p} \oplus \cdots$.  Moreover, it follows from this sequence and the fact that the composition factors of the standard modules all have the form $\IrrMod_{n,p'}$, with $0 \leqslant p' \leqslant n/2-1$, that any other principal indecomposable appearing in this decomposition has the form $\ProjMod_{n,p'}$, with $0 \leqslant p' \leqslant n/2-1$.  But, \thmref{thm:GW} ensures that the multiplicity of $\ProjMod_{n,p}$ in $\tl{n}$ is exactly $\dim \IrrMod_{n,p}$.  We see now that this is only consistent with \eqref{eqn:TLDecompB=0} if $\Ind{\NatMod_{n-1,p}} \cong \ProjMod_{n-1,p}$.  There are therefore no other principal indecomposables, nor other irreducibles.
\end{proof}

\begin{proof}[Proof of \thmref{thm:IdentP}]
Because of \propref{prop:IdentPB=0}, we can (and will) assume throughout the proof of \thmref{thm:IdentP} that $\ell > 2$.  So recall that when $n=1$, $\ProjMod_{1,0} = \tl{1}$ is the only principal indecomposable and it coincides with $\NatMod_{1,0}$ (and $\IrrMod_{1,0}$).  This therefore agrees with the statement of the theorem.  Clearly $\IrrMod_{1,0}$ is the unique irreducible, up to isomorphism. The proof then proceeds in three steps: We first establish that $\set{\IrrMod_{n,p}}$ is a complete set of non-isomorphic irreducibles, then compute how many copies of the projective $\ProjMod_{n,p}$ can be accounted for by the induction hypothesis, and end with the study of the structure of the projective modules.

\smallskip

\noindent{\itshape The set $\set{\IrrMod_{n,p}}$ is a complete set of non-isomorphic irreducibles} --- The module $\tl{n}$ can be written, on the one hand, as
\begin{subequations}
\begin{align}
\tl{n} &= \bigoplus_{p=0}^{\flr{\brac{n-1} / 2}} \brac{\dim \IrrMod_{n-1,p}} \Ind{\ProjMod_{n-1,p}}, \label{eqn:TlByInd} \\
\intertext{by induction over $n$, and, on the other hand, as}
&= \bigoplus_{p=0}^{\flr{n/2}} \brac{\dim \IrrMod_{n,p}}{\ProjMod_{n,p}} \oplus \ProjMod_{\text{new}}, \label{eqn:TlByGW}
\end{align}
\end{subequations}
by \thmref{thm:GW}.  Here, $\ProjMod_{\text{new}}$ is a (possibly empty) direct sum of principal indecomposables that \emph{are not} of the form $\ProjMod_{n,p'}$.  Our first task is to show that $\ProjMod_{\text{new}} = 0$.  This will follow as a consequence of the following analysis of the induced modules $\Ind{\ProjMod_{n-1,p}}$ appearing in \eqref{eqn:TlByInd}.  There are three cases to consider:
\begin{enumerate}[leftmargin=*, label=\textup{(\roman*)}, widest=iii]
\item\label{it:ProofP1} If $\func{k}{n-1,p} = 0$ and $\func{r}{n-1,p} < \ell$, so $\brac{n-1,p}$ lies to the left of the first critical line, then the inductive hypothesis gives $\ProjMod_{n-1,p} \cong \NatMod_{n-1,p}$ and $\Ind{\ProjMod_{n-1,p}} \cong \NatMod_{n,p+1} \oplus \NatMod_{n,p}$, by \corref{cor:IndSplit}.  We therefore identify $\ProjMod_{n,p+1} \cong \NatMod_{n,p+1}$, $\ProjMod_{n,p} \cong \NatMod_{n,p}$ and
\begin{equation} \label{eqn:IndProj1}
\Ind{\ProjMod_{n-1,p}} \cong \ProjMod_{n,p+1} \oplus \ProjMod_{n,p}.
\end{equation}
(Recall that direct summands of projective modules are projective (\propref{prop:ProjSum}).)
\item\label{it:ProofP2} If $\func{r}{n-1,p} = \ell$, so $\brac{n-1,p}$ is critical, then we again have $\ProjMod_{n-1,p} \cong \NatMod_{n-1,p}$, but now \corref{cor:IndSplit} does not apply and we are left with the exact sequence $\ses{\NatMod_{n,p+1}}{\Ind{\ProjMod_{n-1,p}}}{\NatMod_{n,p}}$ that follows from \corref{cor:Induct}.  However, this tells us that the projective module $\Ind{\ProjMod_{n-1,p}}$ has $\IrrMod_{n,p}$ as a quotient, hence that
\begin{equation} \label{eqn:IndProj2}
\Ind{\ProjMod_{n-1,p}} \cong \ProjMod_{n,p} \oplus \ProjMod_{\text{old}}.
\end{equation}
Here, $\ProjMod_{\text{old}}$ is a (possibly empty) direct sum of principal indecomposables that \emph{are} of the form $\ProjMod_{n,p'}$ because the composition factors of $\Ind{\ProjMod_{n-1,p}}$ are those of $\NatMod_{n,p+1}$ and $\NatMod_{n,p}$, hence have the form $\IrrMod_{n,p'}$ for some $p'$. 
\item\label{it:ProofP3} Finally, if $\func{k}{n-1,p} > 0$ and $\func{r}{n-1,p} < \ell$ (all remaining cases), then the inductive hypothesis says that the sequence $\ses{\NatMod_{n-1,p+\func{r}{n-1,p}}}{\ProjMod_{n-1,p}}{\NatMod_{n-1,p}}$ is exact.  Applying the induction functor of \propref{prop:IndRExact} and \corref{cor:IndSplit}, we obtain the exact sequence
\begin{equation} \label{res:Proof3}
\rdses{\NatMod_{n,p+\func{r}{n-1,p}+1} \oplus \NatMod_{n,p+\func{r}{n-1,p}}}{}{\Ind{\ProjMod_{n-1,p}}}{}{\NatMod_{n,p+1} \oplus \NatMod_{n,p}}
\end{equation}
which, when projected onto (generalised) eigenspaces of $F_n$, results in exact sequences
\begin{equation}
\rdses{\NatMod_{n,p+\func{r}{n-1,p}+1}}{}{\ProjMod}{}{\NatMod_{n,p}} \qquad \text{and} \qquad 
\rdses{\NatMod_{n,p+\func{r}{n-1,p}}}{}{\ProjMod'}{}{\NatMod_{n,p+1}},
\end{equation}
where $\Ind{\ProjMod_{n-1,p}} \cong \ProjMod \oplus \ProjMod'$.  The projectives $\ProjMod$ and $\ProjMod'$ then have quotients $\IrrMod_{n,p}$ and $\IrrMod_{n,p+1}$, respectively, so we can write
\begin{equation} \label{eqn:IndProj3}
\Ind{\ProjMod_{n-1,p}} \cong \ProjMod_{n,p} \oplus \ProjMod_{n,p+1} \oplus \ProjMod'_{\text{old}}.
\end{equation}
Here, $\ProjMod'_{\text{old}}$ is a (possibly empty) direct sum of principal indecomposables that \emph{are} of the form $\ProjMod_{n,p'}$ since, as before, all composition factors of $\Ind{\ProjMod_{n-1,p}}$ are of the form $\IrrMod_{n,p'}$ for some $p'$. 
\end{enumerate}
Substituting these conclusions into \eqref{eqn:TlByInd} and comparing with \eqref{eqn:TlByGW}, we see that all the principal indecomposables of $\tl{n}$ have the form $\ProjMod_{n,p}$.  Thus, $\ProjMod_{\text{new}} = 0$ and the $\IrrMod_{n,p}$ with $0 \leqslant p \leqslant \flr{n/2}$ constitute a complete set of pairwise non-isomorphic irreducibles.

\medskip

\noindent{\itshape Counting copies of $\ProjMod_{n,p}$} --- We can now use the above information to complete the proof.  The easiest case is, naturally enough, when $\brac{n,p}$ is sufficiently far from the critical lines.  However, the 
tactic in all cases is the same:  
We fix $p$ and identify some of the $\Ind{\ProjMod_{n-1,p'}}$ that have copies of $\ProjMod_{n,p}$ appearing in their decomposition. 
Then, we compare \eqref{eqn:TlByInd} and \eqref{eqn:TlByGW} to verify that the copies of $\ProjMod_{n,p}$ obtained from these $\ProjMod_{n-1,p'}$ saturate the multiplicity of $\ProjMod_{n,p}$ in the regular representation. This information is then used to identify $\ProjMod_{n,p}$, at least at the level of an exact sequence.

\smallskip

\noindent \textbf{Case 1} [$1 < \func{r}{n,p} < \ell - 1$]:  In this case, both $\brac{n-1,p}$ and $\brac{n-1,p-1}$ are non-critical.  We split the analysis into two sub-cases for clarity:
\begin{itemize}[leftmargin=*]
\item If $\func{k}{n,p} = 0$, then $\brac{n-1,p}$ and $\brac{n-1,p-1}$ lie to the left of the first critical line, so case \ref{it:ProofP1} above applies.\footnote{
Of course, it might happen that $p>\flr{ (n-1)/2}$, so $\brac{n-1,p}$ falls outside the Bratteli diagram.  If so, then we have $n=2p$ and we must formally set $\dim\IrrMod_{n-1,p}$ to zero as indicated by \corref{cor:DimL}.  With this proviso, the argument that follows remains unchanged.
}
\eqnref{eqn:IndProj1} then says that each copy of $\Ind{\ProjMod_{n-1,p}}$ and each copy of $\Ind{\ProjMod_{n-1,p-1}}$ contribute exactly one copy of $\ProjMod_{n,p}$ to $\tl{n}$.  From \eqref{eqn:TlByInd}, this means that we get $\dim \IrrMod_{n-1,p} + \dim \IrrMod_{n-1,p-1}$ copies of $\ProjMod_{n,p}$ in total.  By \corref{cor:DimL}, this is $\dim \IrrMod_{n,p}$ copies, which completely accounts for the multiplicity of $\ProjMod_{n,p}$ in \eqref{eqn:TlByGW}.   
This also means that $\ProjMod_{\text{old}}$ and $\ProjMod'_{\text{old}}$ appearing in \eqref{eqn:IndProj2} and \eqref{eqn:IndProj3} may not contain any $\ProjMod_{n,p}$ such that $1<r(n,p)<\ell-1$. Finally, 
case \ref{it:ProofP1} also gives $\ProjMod_{n,p} \cong \NatMod_{n,p}$, in agreement with \thmref{thm:IdentP}.
\item When $\func{k}{n,p} > 0$, case \ref{it:ProofP3} applies and we find that $\Ind{\ProjMod_{n-1,p}}$ and $\Ind{\ProjMod_{n-1,p-1}}$ again contribute at least $\dim \IrrMod_{n-1,p} + \dim \IrrMod_{n-1,p-1} = \dim \IrrMod_{n,p}$ copies of $\ProjMod_{n,p}$, by \eqnref{eqn:IndProj3} and its analogue with $p \to p-1$.  The projections of the exact sequence \eqref{res:Proof3}, and its $p \to p-1$ analogue, onto the (generalised) eigenspace of $F_n$ of eigenvalue $f_{n,p}$ are
\begin{equation}
\rdses{\NatMod_{n,p+\func{r}{n-1,p}+1}}{\iota}{\ProjMod}{}{\NatMod_{n,p}} \qquad \text{and} \qquad 
\rdses{\NatMod_{n,p-1+\func{r}{n-1,p-1}}}{\iota'}{\ProjMod'}{}{\NatMod_{n,p}},
\end{equation}
which shows that the projectives $\ProjMod$ and $\ProjMod'$ each have at least one direct summand isomorphic to $\ProjMod_{n,p}$.  Noting that $p+\func{r}{n-1,p}+1 = p-1+\func{r}{n-1,p-1} = p+\func{r}{n,p}$, we can rewrite these exact sequences in the form
\begin{equation} \label{es:Nearly1}
\dses{\frac{\NatMod_{n,p+\func{r}{n,p}}}{\ker \iota}}{}{\ProjMod}{}{\NatMod_{n,p}} \qquad \text{and} \qquad 
\dses{\frac{\NatMod_{n,p+\func{r}{n,p}}}{\ker \iota'}}{}{\ProjMod'}{}{\NatMod_{n,p}}.
\end{equation}
Now, if $\ProjMod \ncong \ProjMod_{n,p}$ or $\ProjMod' \ncong \ProjMod_{n,p}$, then we would generate additional copies of projectives $\ProjMod_{n,p'}$ with $F_n$-eigenvalue $f_{n,p'} = f_{n,p}$.  However, all such $\brac{n,p'}$ will have $1 < \func{r}{n,p'} < \ell - 1$, hence are covered by our analysis (see \eqnref{eqn:FOrbitEigs}). Any additional copies of such a $\ProjMod_{n,p'}$ would then contradict \eqnref{eqn:TlByGW}, hence we conclude that $\ProjMod \cong \ProjMod_{n,p}$ and $\ProjMod' \cong \ProjMod_{n,p}$.  The exact sequences \eqref{es:Nearly1} will then prove \thmref{thm:IdentP}, in the case $1 < \func{r}{n,p} < \ell - 1$, once we show that $\ker \iota = \ker \iota' = 0$.  This will follow from a simple dimension argument after we have settled the remaining cases.
\end{itemize}

\smallskip

\noindent \textbf{Case 2} [$\func{r}{n,p} = 1$ or $\func{r}{n,p} = \ell - 1$]:  When $\func{r}{n,p} = 1$, $\brac{n-1,p}$ becomes critical though $\brac{n-1,p-1}$ does not (recall that we may assume that $\ell > 2$).  However, this has almost no effect upon the analysis --- as with Case 1, we arrive at the exact sequences \eqref{es:Nearly1} (the only difference is that we use case \ref{it:ProofP2} above and that the first sequence has $\ker \iota = 0$).\footnote{
From here on, we will omit explicit consideration of the case $\func{k}{n,p} = 0$.  It is easy to see that this case is recovered from the general case by setting modules with labels $p > \flr{n/2}$ to zero.
}
In particular, all copies of $\ProjMod_{n,p}$ in $\tl{n}$ are accounted for by inducing $\ProjMod_{n-1,p}$ and $\ProjMod_{n-1,p-1}$.  However, the additional projectives $\ProjMod_{n,p'}$, with $f_{n,p'} = f_{n,p}$, that would be generated if $\ProjMod \ncong \ProjMod_{n,p}$ or $\ProjMod' \ncong \ProjMod_{n,p}$, could now have $\func{r}{n,p'} = \ell - 1$.  To complete the argument as in Case 1, we therefore need to verify that the $\ProjMod_{n,p'}$ with $\func{r}{n,p'} = \ell - 1$ are likewise also accounted for by inducing appropriate modules $\ProjMod_{n-1,p''}$.

So, suppose that $\func{r}{n,p} = \ell - 1$.  Now, $\brac{n-1,p}$ is non-critical, but $\brac{n-1,p-1}$ is critical (and this makes a difference!).  As in the analysis of Case 1, inducing $\ProjMod_{n-1,p}$ leads to an exact sequence
\begin{equation} \label{es:Nearly2}
\dses{\frac{\NatMod_{n,p+\func{r}{n,p}}}{\ker \iota}}{}{\ProjMod}{}{\NatMod_{n,p}}
\end{equation}
(where $\ProjMod$ is the projection of $\Ind{\ProjMod_{n-1,p}}$ onto the $F_n$-eigenspace of eigenvalue $f_{n,p}$) and thereby to the conclusion that $\Ind{\ProjMod_{n-1,p}}$ contributes $\dim \IrrMod_{n-1,p}$ copies of $\ProjMod_{n,p}$ to $\tl{n}$.  However, $\dim \IrrMod_{n-1,p} = \dim \IrrMod_{n,p}$ in this case, by \corref{cor:DimL}, so all copies of $\ProjMod_{n,p}$ in $\tl{n}$ are accounted for, this time by only inducing $\ProjMod_{n-1,p}$.  The same argument as for Case 1 now identifies $\ProjMod_{n,p}$ as $\ProjMod$ (and $\ProjMod'$) when $\func{r}{n,p} = 1$ or $\ell - 1$.  Again, it only remains to show that $\ker \iota = \ker \iota' = 0$.

\smallskip

\noindent \textbf{Case 3} [$\func{r}{n,p} = \ell$]:  When $\brac{n,p}$ is critical, $\func{r}{n-1,p} = \ell - 1$ and $\func{r}{n-1,p-1} = 1$.  Inducing and projecting as above, we obtain exact sequences for the summands $\ProjMod$ of $\Ind{\ProjMod_{n-1,p}}$ and $\ProjMod'$ of $\Ind{\ProjMod_{n-1,p-1}}$ whose $F_n$-eigenvalue is $f_{n,p}$:
\begin{equation}
\rdses{\NatMod_{n,p + \ell}}{\iota}{\ProjMod}{}{\NatMod_{n,p}} \qquad \text{and} \qquad 
\rdses{\NatMod_{n,p}}{\iota'}{\ProjMod'}{}{\NatMod_{n,p}}.
\end{equation}
Now, all of the standard modules appearing in these sequences are critical, hence irreducible (\corref{cor:RadCrit}).  There are therefore only three possibilities for $\ProjMod$:  First, $\ker \iota = \NatMod_{n,p + \ell}$, hence $\ProjMod \cong \NatMod_{n,p}$.  Second, $\ker \iota = 0$ and $\ProjMod$ is indecomposable with exact sequence $\ses{\NatMod_{n,p + \ell}}{\ProjMod}{\NatMod_{n,p}}$.  Third, $\ker \iota = 0$ and $\ProjMod$ decomposes as $\NatMod_{n,p + \ell} \oplus \NatMod_{n,p}$. Each $\Ind{\ProjMod_{n-1,p}}$ therefore accounts for one copy of $\ProjMod_{n,p}$ in the first two possibilities (for a total of $\dim\IrrMod_{n-1,p}$ copies), but one copy of $\ProjMod_{n,p}$ and one copy of $\ProjMod_{n,p + \ell}$ in the third (so the total becomes $\dim\IrrMod_{n-1,p}+\dim\IrrMod_{n-1,p-\ell}$).  The possibilities for $\ProjMod'$ are likewise $\ProjMod' \cong \NatMod_{n,p}$, $\ProjMod'$ indecomposable with exact sequence $\ses{\NatMod_{n,p}}{\ProjMod'}{\NatMod_{n,p}}$, and $\ProjMod' \cong 2 \: \NatMod_{n,p}$, so each $\Ind{\ProjMod_{n-1,p-1}}$ contributes one copy of $\ProjMod_{n,p}$ in the first two possibilities (for a total of $\dim\IrrMod_{n-1,p-1}$ copies), but two copies in the third (yielding $2\dim\IrrMod_{n-1,p-1}$ copies in all). 

The number of copies of $\ProjMod_{n,p}$ required by \eqnref{eqn:TlByGW} is
\begin{align}
\dim \IrrMod_{n,p} &= \dim \NatMod_{n,p} = \dim \NatMod_{n-1,p} + \dim \NatMod_{n-1,p-1} & &\text{(by \propref{prop:Restrict})} \notag \\
&= \dim \IrrMod_{n-1,p} + \dim \RadMod_{n-1,p} + \dim \IrrMod_{n-1,p-1} + \dim \RadMod_{n-1,p-1} & &\text{(by \propref{prop:S/R})} \notag \\
&= \dim \IrrMod_{n-1,p} + 2 \: \dim \IrrMod_{n-1,p-1} + \dim \IrrMod_{n-1,p-\ell} & &\text{(by \propref{prop:DimR=DimL}).}
\end{align}
Considering all the critical $\brac{n,p}$, we see that this multiplicity can only be attained if the third possibilities for $\ProjMod$ and $\ProjMod'$ occur in each case.\footnote{
Again, any module whose label is out of the allowed range $0\leqslant p\leqslant \lfloor n/2\rfloor$ is understood to be trivial.
}
Then, contributions are received from inducing $\ProjMod_{n-1,p}$, $\ProjMod_{n-1,p-1}$ and $\ProjMod_{n-1,p-\ell}$.  As the second possibility above is 
ruled out, it follows that $\ProjMod_{n,p} \cong \NatMod_{n,p}$ for all critical $\brac{n,p}$, as required.

\medskip

\noindent{\itshape The structure of the projective $\ProjMod_{n,p}$} --- It only remains to prove that $\ker \iota$ (and $\ker \iota'$) are $0$ in the exact sequences \eqref{es:Nearly1} and \eqref{es:Nearly2} that we have derived for the non-critical $\ProjMod_{n,p}$:
\begin{equation} \label{es:Nearly}
\dses{\frac{\NatMod_{n,p+\func{r}{n,p}}}{\ker \iota}}{}{\ProjMod_{n,p}}{}{\NatMod_{n,p}}.
\end{equation}
We remark that this equation still holds when $\brac{n,p}$ lies to the left of the first critical line ($\func{k}{n,p} = 0$) if we understand that the $\NatMod_{n,p'}$ with $p' > \flr{n/2}$ are $0$.

Since we have shown that the $\ProjMod_{n,p}$ form a complete set of principal indecomposables, the constraint \eqref{eqn:LP=VV} relating the dimensions of the $\ProjMod_{n,p}$, $\NatMod_{n,p}$ and $\IrrMod_{n,p}$ may be written in the form
\begin{equation}
\sum_{p=0}^{\flr{n/2}} \dim \IrrMod_{n,p} \: \dim \ProjMod_{n,p} = \sum_{p=0}^{\flr{n/2}} \brac{\dim\NatMod_{n,p}}^2.
\end{equation}
This constraint is now used to determine the dimensions of the projective modules $\ProjMod_{n,p}$ and, therefore, the dimensions of the kernel of $\iota$ and $\iota'$.
The contributions from the critical $\brac{n,p}$ may be cancelled in the constraint because then $\ProjMod_{n,p} \cong \NatMod_{n,p} \cong \IrrMod_{n,p}$.  

We break down the remaining sum over $p$ into non-critical orbits under reflection about critical lines (see the discussion after \corref{cor:2013}). We thus rewrite the original sum as, first, a sum over non-critical orbits, labelled by their leftmost element $p_1$, and then as a sum over elements $p_1 > p_2 > \cdots > p_m \geqslant 0$ of the orbit of $p_1$:
\begin{equation} \label{eq:dimGW}
\sum_{n-2p_1+1 < \ell} \ \sum_{i=1}^m \dim \IrrMod_{n,p_i} \: \dim \ProjMod_{n,p_i} = 
\sum_{n-2p_1+1 < \ell} \ \sum_{i=1}^m \brac{\dim\NatMod_{n,p_i}}^2.
\end{equation}
For the $\brac{n,p}$ contributing to \eqref{eq:dimGW}, the exact sequences \eqref{es:Nearly} and the relations \eqref{eq:LackOfImagination} immediately imply the bounds 
\begin{equation} \label{eqn:DimPBound}
\dim \ProjMod_{n,p_i} \leqslant \dim \NatMod_{n,p_i+\func{r}{n,p_i}} + \dim \NatMod_{n,p_i} = 
   \dim \NatMod_{n,p_{i-1}} + \dim \NatMod_{n,p_i},
\end{equation}
with saturation attained if and only if $\ker \iota = 0$.  
(All dimensions of modules indexed by $p_0$ or $p_{m+1}$ are understood to be zero.) 
These, in turn, give the following bound for the contribution from each non-critical orbit to the left-hand side of \eqref{eq:dimGW}:
\begin{align}
\sum_{i=1}^{m} \dim \IrrMod_{n,p_i} \: \dim \ProjMod_{n,p_i} 
&\leqslant \sum_{i=1}^{m} \brac{\dim \IrrMod_{n,p_i} \: \dim \NatMod_{n,p_i} + \dim \IrrMod_{n,p_i} \: \dim \NatMod_{n,p_{i-1}}} \notag \\
&= \sum_{i=1}^{m} \brac{\brac{\dim \NatMod_{n,p_i}}^2 - \dim \RadMod_{n,p_i} \: \dim \NatMod_{n,p_i} + \dim \IrrMod_{n,p_i} \: \dim \NatMod_{n,p_{i-1}}} \notag \\
&= \sum_{i=1}^{m} \brac{\brac{\dim \NatMod_{n,p_i}}^2 - \dim \IrrMod_{n,p_{i+1}} \: \dim \NatMod_{n,p_i} + \dim \IrrMod_{n,p_i} \: \dim \NatMod_{n,p_{i-1}}} \notag \\
&=\sum_{i=1}^{m}(\dim\NatMod_{n,p_i})^2.
\end{align}
Here, we have used \eqref{eqn:DimPBound}, then \propref{prop:S/R} and, finally, \propref{prop:DimR=DimL}.  For \eqref{eq:dimGW} to hold, this inequality must be an equality for all non-critical orbits, hence the $\ker \iota$ (and $\ker \iota'$) must always vanish.  This proves that the sequence \eqref{eq:THEes} is exact for all non-critical $\ProjMod_{n,p}$ (to the right of the first critical line), completing the proof of \thmref{thm:IdentP}.
\end{proof}

\section{Summary of Results} \label{sec:Summary}

We conclude the article with a brief summary outlining what has been proven and a few ideas concerning what can be done with the results.  First, we have presented the well-known equivalence between the algebraic and diagrammatic definitions of the Temperley-Lieb algebra $\tl{n}$ (with parameter $\beta = q + q^{-1}$), deducing as a consequence that its dimension is given by the $n$-th Catalan number.  We have then discussed the standard $\tl{n}$-modules $\NatMod_{n,p}$, with $0 \leqslant p \leqslant \flr{n/2}$, each of which admits a natural invariant bilinear form, and explained that a study of the irreducibility of the $\NatMod_{n,p}$ may be reduced, in almost all cases, to the consideration of the non-degeneracy of this bilinear form.  Of course, we have also detailed how to analyse the exceptional cases when the natural bilinear form is not useful.

The question of whether these standard modules are irreducible or not turned out to have structural implications for the Temperley-Lieb algebra.  We used restriction from $\tl{n}$ to $\tl{n-1}$ to derive a recursion relation describing the kernel of the bilinear form on each $\NatMod_{n,p}$.  This gave a complete answer to the irreducibility of the $\NatMod_{n,p}$ and led to a criterion for the semisimplicity of $\tl{n}$.  In particular, this recovered the well-known result that the $\tl{n}$ are all semisimple when $q$ is not a root of unity.  While our strategy follows that of \cite{WesRep95} rather closely, we believe that our proof is new.  Moreover, we found that the aforementioned description of the kernel suggested that the standard modules were either irreducible or had an irreducible maximal submodule, the radical.

To prove this suggestion, we turned to the induced modules obtained from the $\NatMod_{n,p}$ by including $\tl{n}$ in $\tl{n+1}$ in the obvious way.  After deriving the required structural information concerning these induced modules, we detailed a non-trivial computation involving the action of a little-known central element $F_n \in \tl{n}$ on the induced modules.  This then allowed us to deduce the existence of non-trivial homomorphisms between certain standard modules.  As a consequence, we obtained the desired irreducibility of the (non-trivial) radicals as well as a complete description of the space of homomorphisms between any two standard modules.  All of these proofs appear to be new.

Finally, we used our knowledge of induced modules to systematically construct the principal indecomposable modules $\ProjMod_{n,p}$ for the Temperley-Lieb algebra.  These turn out to be realised as either standard modules or as submodules, corresponding to a given generalised eigenspace of the central element $F_n$, of a multiply-induced standard module.  The structures of the $\ProjMod_{n,p}$ followed immediately from the analysis, as did the statement that the irreducible $\tl{n}$-modules $\IrrMod_{n,p}$ that can be constructed from the $\NatMod_{n,p}$ are, in fact, exhaustive (up to isomorphism).  While the structure of the principal indecomposables is known (it appeared first in \cite{Martin}), we are confident that our straight-forward proof is also new.  Certainly, it relies upon a particularly remarkable property of $F_n$, namely that it completely distinguishes the non-critical blocks of $\tl{n}$, a property which does not seem to be shared by other better-known central elements.

We include, for convenience, pictorial representations of the structures of the standard and principal indecomposable modules when $q$ is a root of unity.  These take the form of annotated Loewy diagrams which are popular in, for example, logarithmic conformal field theory.  The Loewy diagrams for the standard modules have two forms according as to whether $\brac{n,p}$ is \emph{critical}, meaning $q^{2 \brac{n-2p+1}} = 1$, or not.  Let $\brac{n,p}$ and $\brac{n,p'}$ form a symmetric pair, in the sense of \secref{sec:R=L}, with $p > p'$.  The possible Loewy diagrams for the $\NatMod_{n,p}$ ($0 \leqslant p \leqslant \flr{n/2}$) are then as follows:
\begin{center}
\begin{tikzpicture}[thick,
	nom/.style={circle,draw=black!10,fill=black!10,inner sep=1pt}
	]
\node (l) at (0,2) [] {$\IrrMod_{n,p}$};
\node (mt) at (5,3) [] {$\IrrMod_{n,p}$};
\node (mb) at (5,1) [] {$\IrrMod_{n,p'}$};
\node (r) at (10,2) [] {$\IrrMod_{2p,p-1}$};
\node at (0,0) [nom] {$\NatMod_{n,p}$};
\node at (5,0) [nom] {$\NatMod_{n,p}$};
\node at (10,0) [nom] {$\NatMod_{n,p}$};
\node at (0,-0.75) {$\brac{n,p}$ critical};
\node at (5,-0.75) {$\brac{n,p}$ non-critical};
\node at (5,-1.25) {($n \neq 2p$ or $\beta \neq 0$)};
\node at (10,-0.75) {($n = 2p$ and $\beta = 0$).};
\draw [->] (mt) -- (mb);
\end{tikzpicture}
\end{center}
To explain, these diagrams present the (irreducible) composition factors of the $\NatMod_{n,p}$ as nodes, connected by arrows, that are meant to represent the action of the algebra.  The diagrams on the left and right indicate that $\NatMod_{n,p}$ is in fact irreducible (and isomorphic to $\IrrMod_{n,p}$ and $\IrrMod_{2p,p-1}$, respectively) in these cases.  By contrast, the diagram in the middle has its only arrow pointing towards $\IrrMod_{n,p'}$, indicating that $\NatMod_{n,p}$ has, in this case, a submodule isomorphic to $\IrrMod_{n,p'}$.  In fact, it is the maximal submodule --- quotienting corresponds to removing the node $\IrrMod_{n,p'}$ and the arrow, leaving us with only $\IrrMod_{n,p}$ (and no arrows).  In other words, $\NatMod_{n,p} / \IrrMod_{n,p'} \cong \IrrMod_{n,p}$.  We remark that for sufficiently small $p$, the $p'$ required to form a symmetric pair would be negative.  In that case, $\IrrMod_{n,p'}$ should be understood to be $\set{0}$.  The diagram in the middle should then be taken to degenerate into the diagram on the left.

The Loewy diagrams for the principal indecomposables are only a little more complicated.  First, we remark that when $\beta = 0$ and $n$ is even, $p$ is restricted to the range $0 \leqslant p \leqslant n/2 - 1$; otherwise, the range is $0 \leqslant p \leqslant \flr{n/2}$.  Now, take $p'' > p > p'$ so that $\brac{n,p''}$ and $\brac{n,p}$, as well as $\brac{n,p}$ and $\brac{n,p'}$, form symmetric pairs.  The possible Loewy diagrams for the $\ProjMod_{n,p}$ are then as follows:
\begin{center}
\begin{tikzpicture}[thick,
	nom/.style={circle,draw=black!10,fill=black!10,inner sep=1pt}
	]
\node (l) at (0,2.5) [] {$\IrrMod_{n,p}$};
\node (mt) at (5,4) [] {$\IrrMod_{n,p}$};
\node (ml) at (3.5,2.5) [] {$\IrrMod_{n,p''}$};
\node (mr) at (6.5,2.5) [] {$\IrrMod_{n,p'}$};
\node (mb) at (5,1) [] {$\IrrMod_{n,p}$};
\node (rt) at (10,4) [] {$\IrrMod_{2p+2,p}$};
\node (rr) at (11.5,2.5) [] {$\IrrMod_{2p+2,p-1}$};
\node (rb) at (10,1) [] {$\IrrMod_{2p+2,p}$};
\node at (0,0) [nom] {$\ProjMod_{n,p}$};
\node at (5,0) [nom] {$\ProjMod_{n,p}$};
\node at (10,0) [nom] {$\ProjMod_{n,p}$};
\node at (0,-0.75) {$\brac{n,p}$ critical};
\node at (5,-0.75) {$\brac{n,p}$ non-critical};
\node at (5,-1.25) {($n \neq 2p+2$ or $\beta \neq 0$)};
\node at (10,-0.75) {($n = 2p+2$ and $\beta = 0$).};
\draw [->] (mt) -- (ml);
\draw [->] (mt) -- (mr);
\draw [->] (ml) -- (mb);
\draw [->] (mr) -- (mb);
\draw [->] (rt) -- (rr);
\draw [->] (rr) -- (rb);
\end{tikzpicture}
\end{center}
Again, the diagram on the left indicates that the critical principal indecomposables are irreducible.  To understand the diagram in the middle, note that there is a submodule isomorphic to $\IrrMod_{n,p}$ (the bottom node with arrows only pointing in).  Quotienting by this submodule leads to a module with three composition factors.  It has a submodule isomorphic to $\IrrMod_{n,p''} \oplus \IrrMod_{n,p'}$,\footnote{
One can (and should) ask why this submodule is a direct sum.  Equivalently, why is there no arrow from the factor $\IrrMod_{n,p'}$ to the factor $\IrrMod_{n,p''}$ in the middle diagram?  (That there is no arrow in the opposite direction follows directly from \eqref{eq:THEes}.)  This may be answered by showing that any short exact sequence of the form $\ses{\IrrMod_{n,p''}}{\mathcal{M}}{\IrrMod_{n,p'}}$ necessarily splits.  To see this, note that projectivity implies the existence of a homomorphism $\delta \colon \ProjMod_{n,p'} \rightarrow \mathcal{M}$ such that the following diagram commutes:
\begin{center}
\begin{tikzpicture}[auto,thick,scale=0.5]
\node (ll) at (-4,0) [] {$0$};
\node (l) at (-2,0) [] {$\IrrMod_{n,p''}$};
\node (m) at (0,0) [] {$\mathcal{M}$};
\node (r) at (2,0) [] {$\IrrMod_{n,p'}$};
\node (rr) at (4,0) [] {$0$.};
\node (t) at (0,2) [] {$\ProjMod_{n,p'}$};
\draw [->] (ll) -- (l);
\draw [->] (l) -- (m);
\draw [->] (m) -- (r);
\draw [->] (r) -- (rr);
\draw [->] (t) to node [swap] {$\delta$} (m);
\draw [->] (t) -- (r);
\end{tikzpicture}
\end{center}
Comparing the composition factors of $\mathcal{M}$ and $\ProjMod_{n,p'}$, we see that either $\delta = 0$ or $\func{\im}{\delta} \cong \IrrMod_{n,p'}$.  The former contradicts the commutativity of the diagram, so we conclude that $\IrrMod_{n,p'}$ is a submodule of $\mathcal{M}$, whence $\mathcal{M} \cong \IrrMod_{n,p'} \oplus \IrrMod_{n,p''}$.
} 
and quotienting by this submodule results in something isomorphic to $\IrrMod_{n,p}$.  The diagram on the right is interpreted similarly.

The degeneracy conditions are likewise a little more complicated:  When $p$ is small enough in the middle diagram that $p'$ would have to be negative, one must delete the node $\IrrMod_{n,p'}$ and arrows pointing towards or away from it.  The resulting principal indecomposable therefore has three composition factors.  However, when $p$ is large enough that $p''$ would have to be greater than $\flr{n/2}$, one should delete \emph{both} $\IrrMod_{n,p''}$ and the bottom $\IrrMod_{n,p}$, so that $\ProjMod_{n,p}$ has only two composition factors.  This means that $\ProjMod_{n,p}\cong \NatMod_{n,p}$ when $p$ is on the left of the first critical line. It may, of course, happen for $n$ sufficiently small that there are non-critical $p$ for which both $p'$ and $p''$ fall outside the allowed range.  Then, the middle diagram degenerates into the diagram on the left.  Similarly, the diagram on the right degenerates only when $n=2$ and $p=0$.  In this case, the Loewy diagram for $\ProjMod_{2,0}$ has two composition factors, both isomorphic to $\IrrMod_{2,0}$, connected by an arrow.

Finally, we remark that such a complete description of the principal indecomposables allows one to apply standard tools from homological algebra to answer more advanced questions regarding the variety of possible indecomposable structures.  We shall not do so here, but will content ourselves with mentioning that it is now easy to write down projective presentations and resolutions for standard and irreducible modules and so compute the extension groups between them.  This technology answers, in particular, the question of whether giving a non-split short exact sequence completely determines the module in the middle.  For example, one can use this to show that the characterisation we have obtained for the principal indecomposables, in terms of exact sequences, does in fact identify them up to isomorphism.  One can also use the extension computations to construct injective modules and, indeed, give a complete classification of the indecomposable modules of the Temperley-Lieb algebra (something which is not possible for most associative algebras).  These directions are clearly crucial for a sound mathematical understanding of Temperley-Lieb representation theory as well as necessary for a complete identification of Temperley-Lieb modules in physical models.  In both cases, the structural results proven here provide a springboard from which one can profitably tackle advanced questions.

\bigskip

\section*{Acknowledgements}

The authors would like to thank Alexi Morin-Duchesne for a useful suggestion, Azat Gainutdinov, Gus Lehrer and Paul Martin for helpful discussions, and Fred Goodman for pointing out a gap in an earlier version. DR is supported by an Australian Research Council Discovery Project DP1093910 and YSA holds a grant from the Canadian Natural Sciences and Engineering Research Council. This support is gratefully acknowledged.

\appendix

\section{Central Elements} \label{app:Casimir}

The goal of this appendix is to introduce the central element $F_n \in \tl{n}$ which plays a crucial role in \secDref{sec:R=L}{sec:Proj}. This is not the only important central element in the theory of the Temperley-Lieb algebra. One should also mention the well-known central element $C_n = \brac{t_1 \cdots t_{n-1}}^n \in \tl{n}$, where $t_i = \wun - q u_i$, which is derived from its analogue in the braid group \cite{Chow,Birman}, as well as the central idempotents described by Jones \cite{JonInd83} and Wenzl \cite{Wenzl88}. For our purposes, $F_n$ will suffice (whereas $C_n$ in particular will not). This appendix uses notation and one result from each of \secDref{sec:Diagrams}{sec:Reps}; results from other sections are also used to provide examples. The element $F_n$ is first used in \corref{cor:ResSplit}.

The role of central elements is well-known to physicists, where the quadratic Casimir of a semisimple Lie algebra probably provides the most familiar example. One use for such central elements is to decompose representations into their eigenspaces. More precisely, if $c$ is a central element of an algebra $\mathsf{A}$ and $\mathcal{M}$ a module over this algebra, then the linear map $\phi_c \colon \mathcal{M} \rightarrow \mathcal{M}$ obtained by left-multiplying elements of $\mathcal{M}$ by $c$ defines a homomorphism of modules: $a \phi_c(m) = acm = cam = \phi_c(am)$ for all $a \in \mathsf{A}$, $m \in \mathcal{M}$. The eigenspaces of $\phi_c$ are then submodules of $\mathcal{M}$. However, it might happen that $\phi_c$ is not diagonalisable on $\mathcal{M}$, in which case the generalised eigenspaces $\mathcal M_\lambda=\{m\in\mathcal M\,|\, (\phi_c-\lambda \wun)^{\dim\mathcal M}m=0\}$ are also submodules. In fact, $\mathcal{M}$ splits, as an $\mathsf{A}$-module, into the \emph{direct sum} of the generalised eigenspaces of $\phi_c$. This is a very useful property.

We therefore turn to the central element $F_n\in \tl{n}$ and its properties. The computations which follow, and even the definition of $F_n$ itself, are diagrammatic.  The construction is based upon the recent analysis of a double row transfer matrix \cite{PeaLog06}. It also appears briefly in \cite{PeaSol10} (with $\beta=0$), though it does not seem to have been recognised as central there.  This important property was made explicit in \cite{MorLog10}, where its eigenvalues were computed.  We will follow their approach with only minor changes.

To define $F_n$ in $\tl{n}$, it is useful to introduce the following shorthand (recall that $\beta = q + q^{-1}$):
\begin{equation} \label{eqn:Shorthand}
\parbox{9mm}{\begin{center}
\includegraphics[height=8mm]{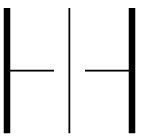}
\end{center}}
\ = q^{1/2} \ 
\parbox{9mm}{\begin{center}
\includegraphics[height=8mm]{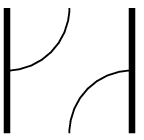}
\end{center}}
\ - q^{-1/2} \ 
\parbox{9mm}{\begin{center}
\includegraphics[height=8mm]{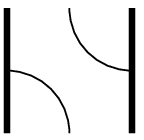}
\end{center}}
\qquad \text{and} \qquad
\parbox{9mm}{\begin{center}
\includegraphics[height=8mm]{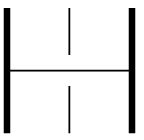}
\end{center}}
\ = q^{1/2} \ 
\parbox{9mm}{\begin{center}
\includegraphics[height=8mm]{tl-boxa}
\end{center}}
\ - q^{-1/2} \ 
\parbox{9mm}{\begin{center}
\includegraphics[height=8mm]{tl-boxd}
\end{center}}
\ \ .
\end{equation}
These \emph{crossings} are to be interpreted as formal objects which will be used as building blocks to form (linear combinations of) $n$-diagrams.  In particular, we define $F_n$ as follows:
\begin{equation}\label{eq:leFn}
F_n = \ 
\parbox{13mm}{\begin{center}
\includegraphics[height=39mm]{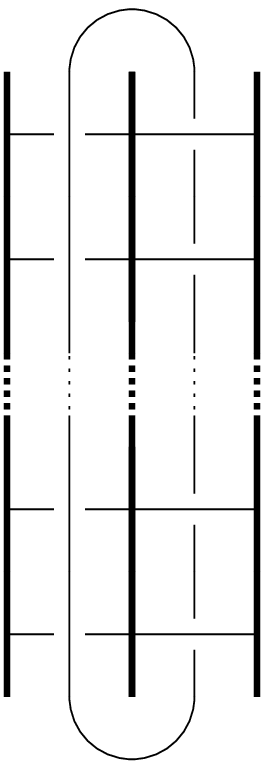}
\end{center}}
\ .
\end{equation}
This looks deceptively simple, but the notation hides a large number of diagrams as each crossing of thin lines stands for a sum of two terms.  Computing even $F_3$ explicitly requires considering a linear combination of $2^6 = 64$ $3$-diagrams, even though $\dim \tl{3} = 5$. For comparison, we list the first few explicit forms for $F_n$ and the better known $C_n$:
\begin{subequations}
\begin{align}
F_1 &= (q^2+q^{-2}) \wun, & 
C_1 &= \wun, \\
F_2 &= (q^3+q^{-3}) \wun - (q-q^{-1})^2 u_1, & 
C_2 &= \wun + q^2 (q-q^{-1}) u_1, \\
F_3 &= (q^4+q^{-4}) \wun - (q-q^{-1}) (q^2-q^{-2}) (u_1+u_2) & 
C_3 &= \wun + q^3 (q^2-q^{-2}) (u_1+u_2) \notag \\
&\mspace{30mu} + (q-q^{-1})^2 (u_1u_2+u_2u_1) & 
&\mspace{30mu} - q^3 (q-q^{-1}) (u_1u_2+u_2u_1).
\end{align}
\end{subequations}
We note that because the number of crossings used to construct $F_n$ is even, the coefficients of the corresponding linear combination of $n$-diagrams will only contain integral powers of $q$ and $q^{-1}$.  Note further that
\begin{equation}
\parbox{9mm}{\begin{center}
\includegraphics[height=8mm]{tl-boxv}
\end{center}}
\raisebox{2mm}{${}^{\ \dag}$} = q^{1/2} \ 
\parbox{9mm}{\begin{center}
\includegraphics[height=8mm]{tl-boxa}
\end{center}}
\ - q^{-1/2} \ 
\parbox{9mm}{\begin{center}
\includegraphics[height=8mm]{tl-boxd}
\end{center}}
\ = \ 
\parbox{9mm}{\begin{center}
\includegraphics[height=8mm]{tl-boxh}
\end{center}}
\ ,
\end{equation}
recalling that the adjoint was chosen in \secref{sec:Reps} to be linear, not antilinear.  It follows that $F_n$ is \emph{self-adjoint}. Finally, it is also true \cite{MorLog10} that $F_n$ is invariant under $q \leftrightarrow q^{-1}$ (we remark that $C_n$ does not have this last property).

\begin{proposition} \label{prop:FCentral}
$F_n \in \tl{n}$ is central.
\end{proposition}
\begin{proof}
The key insight is that
\begin{subequations} \label{eqn:LinkProp}
\begin{equation} \label{eqn:LinkProp1}
\parbox{13mm}{\begin{center}
\includegraphics[height=16mm]{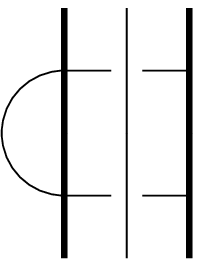}
\end{center}}
\ = q \ 
\parbox{13mm}{\begin{center}
\includegraphics[height=16mm]{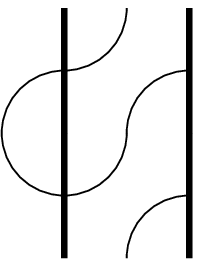}
\end{center}}
\ - \ 
\parbox{13mm}{\begin{center}
\includegraphics[height=16mm]{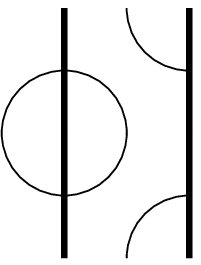}
\end{center}}
\ - \ 
\parbox{13mm}{\begin{center}
\includegraphics[height=16mm]{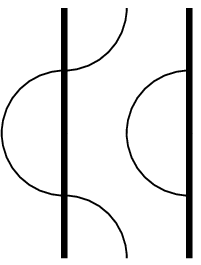}
\end{center}}
\ + q^{-1} \ 
\parbox{13mm}{\begin{center}
\includegraphics[height=16mm]{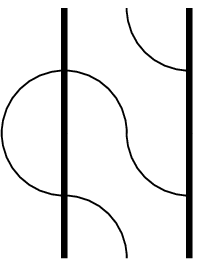}
\end{center}}
\ = -\ 
\parbox{13mm}{\begin{center}
\includegraphics[height=16mm]{tl-ubb3}
\end{center}}
\end{equation}
and similarly, that
\begin{equation} \label{eqn:LinkProp2}
\parbox{13mm}{\begin{center}
\includegraphics[height=16mm]{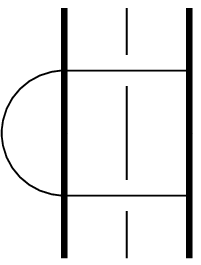}
\end{center}}
\ = -\ 
\parbox{13mm}{\begin{center}
\includegraphics[height=16mm]{tl-ubb3}
\end{center}}
\ .
\end{equation}
\end{subequations}
Thus,
\begin{equation}
\parbox{25mm}{\begin{center}
\includegraphics[height=16mm]{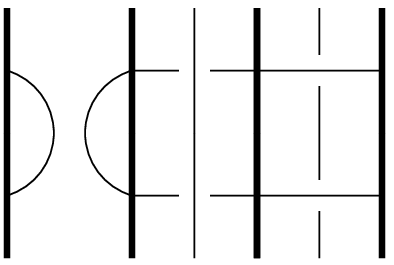}
\end{center}}
\ = -\ 
\parbox{25mm}{\begin{center}
\includegraphics[height=16mm]{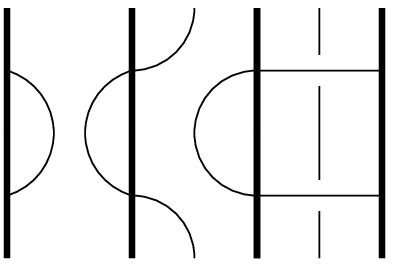}
\end{center}}
\ = \ 
\parbox{25mm}{\begin{center}
\includegraphics[height=16mm]{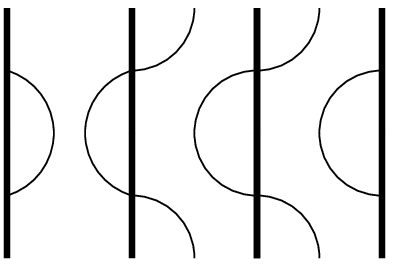}
\end{center}}
\end{equation}
is symmetric under a vertical reflection, hence self-adjoint.  It follows that $u_i F_n$ must be self-adjoint for all $i$, whence
\begin{equation}
u_i F_n = \brac{u_i F_n}^{\dag} = F_n^{\dag} u_i^{\dag} = F_n u_i,
\end{equation}
as required.
\end{proof}
The eigenvalue of $C_n$ on the standard module $\NatMod_{n,p}$ is easily shown to be $q^{2p \brac{n+1-p}}$ (see for example \cite{WesRep95}). The corresponding result for $F_n$ is as follows.
\begin{proposition} \label{lem:FEigs}
The element $F_n$ acts on $\NatMod_{n,p}$ as the identity times $f_{n,p} = q^{n-2p+1} + q^{-\brac{n-2p+1}}$.
\end{proposition}
\begin{proof}
First note that any central element must act as a multiple of the identity on $\NatMod_{n,p}$ by \propref{prop:EndLinks}.  We may therefore determine this multiple by computing the action of $F_n$ on any $\brac{n,p}$-link state (modulo terms which have more than $p$ links).  A convenient choice is the state $z_p \in \NatMod_{n,p}$ which has first $p$ simple links, followed by $n-2p$ defects.

We will break the computation up into two pieces, corresponding to the action on the $p$ simple links and the action on the $n-2p$ defects. The first is easy:  Using \eqnref{eqn:LinkProp} (or rather its adjoint) repeatedly, we obtain
\begin{equation} \label{eqn:ContributionLinks}
\parbox{11mm}{\begin{center}
\includegraphics[height=28mm]{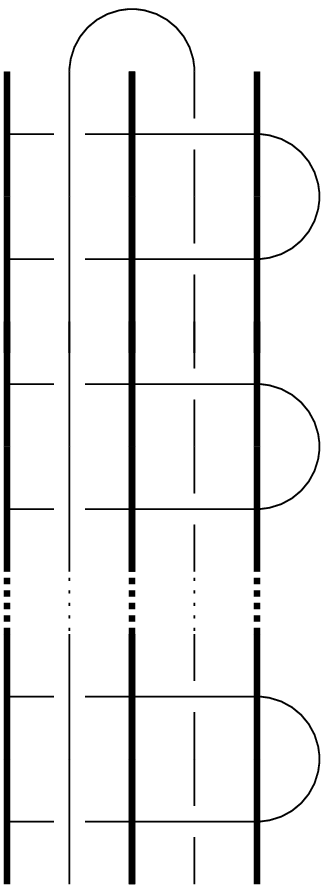}
\end{center}}
\ = \ 
\parbox{11mm}{\begin{center}
\includegraphics[height=28mm]{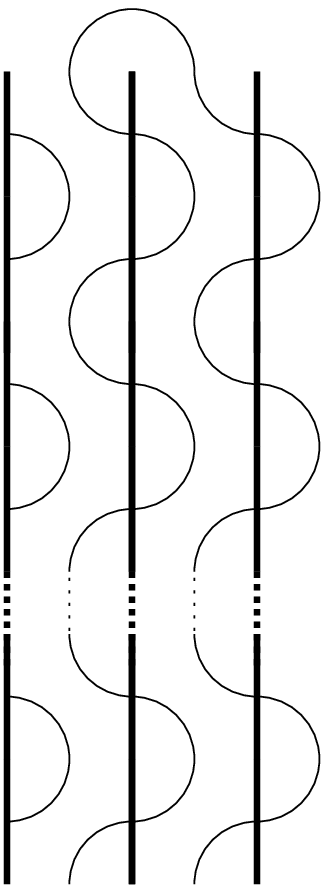}
\end{center}}
\ .
\end{equation}
This will clearly contribute $1$ or $\beta$, according as to whether the loop wiggling down is closed by the action on the defects.

This latter action requires a bit of explanation.  First, apply the definition \eqref{eqn:Shorthand} to expand the top-right partial crossing (the loop at the top of these diagrams stands for the big loop of \eqref{eqn:ContributionLinks}):
\begin{equation} \label{eqn:ContributionLinks2}
\parbox{11mm}{\begin{center}
\includegraphics[height=26mm]{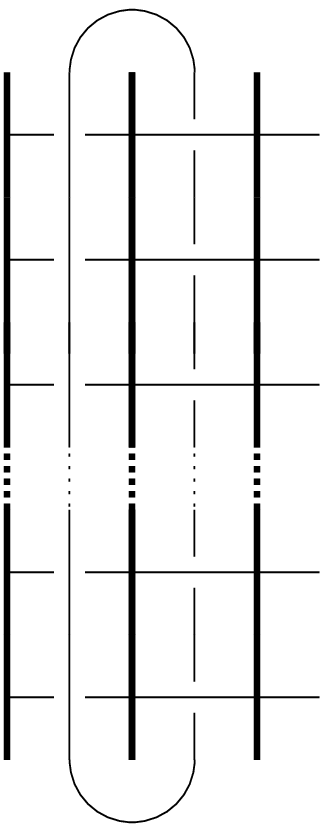}
\end{center}}
\ = q^{1/2} \ 
\parbox{11mm}{\begin{center}
\includegraphics[height=26mm]{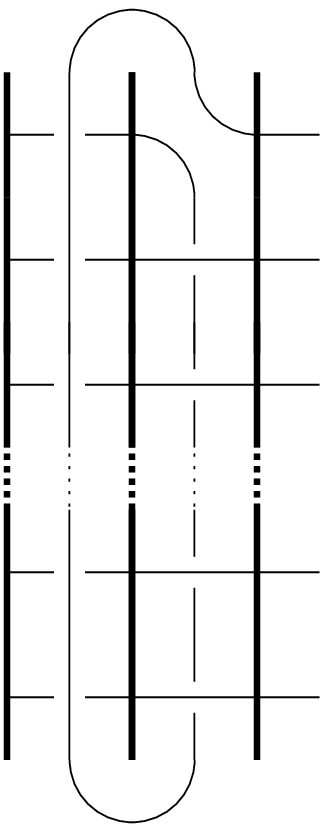}
\end{center}}
\ - q^{-1/2} \ 
\parbox{11mm}{\begin{center}
\includegraphics[height=26mm]{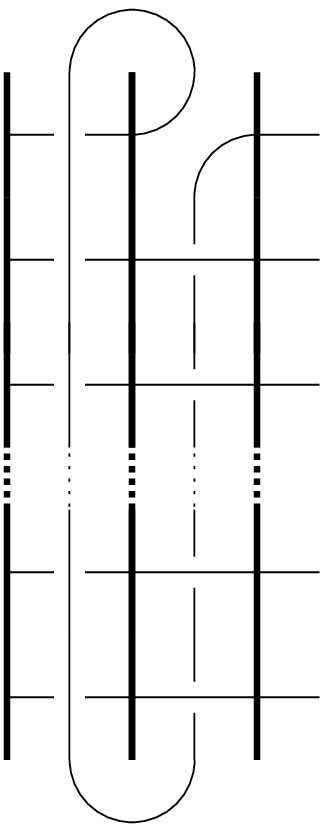}
\end{center}}
\ .
\end{equation}
Consider now the first term on the right-hand side.  When we expand the crossing immediately below that just considered, only one term contributes.  The other leads to a link on the middle vertical border which creates an extra link in the result (using the adjoint of \eqnref{eqn:LinkProp1}), so the result vanishes.  This observation propagates down the right side, yielding a factor of $q^{1/2}$ at each step.  Thus,
\begin{equation}
q^{1/2} \ 
\parbox{11mm}{\begin{center}
\includegraphics[height=26mm]{tl-Fz4}
\end{center}}
\ = q^{\brac{n-2p}/2} \ 
\parbox{11mm}{\begin{center}
\includegraphics[height=26mm]{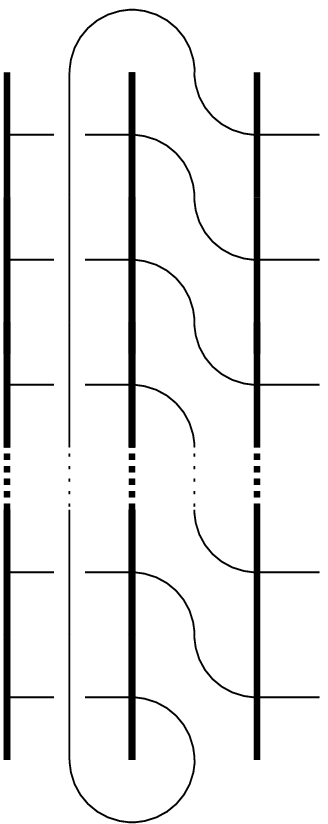}
\end{center}}
\ .
\end{equation}
Now expand the crossing at the top-left.  Because of the big loop at the top, we again find that only one term contributes --- the other corresponds to closing two defects and hence vanishes.  This observation propagates down the left side, picking up a factor of $q^{1/2}$ at each step, until we reach the last crossing for which \emph{both} terms contribute:
\begin{equation}
q^{\brac{n-2p}/2} \ 
\parbox{11mm}{\begin{center}
\includegraphics[height=26mm]{tl-Fz7}
\end{center}}
\ = q^{n-2p-1/2} \ 
\parbox{11mm}{\begin{center}
\includegraphics[height=26mm]{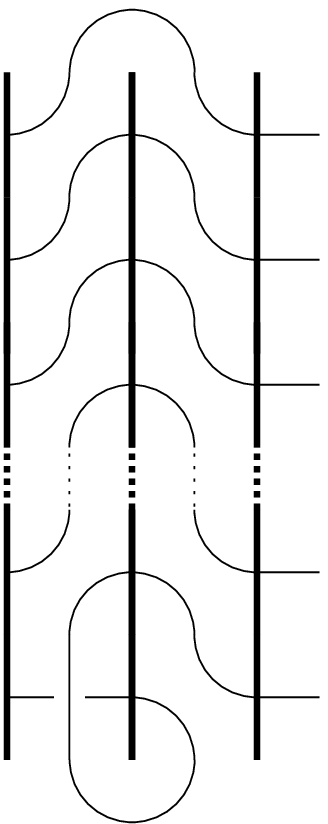}
\end{center}}
\ = q^{n-2p} \ 
\parbox{11mm}{\begin{center}
\includegraphics[height=26mm]{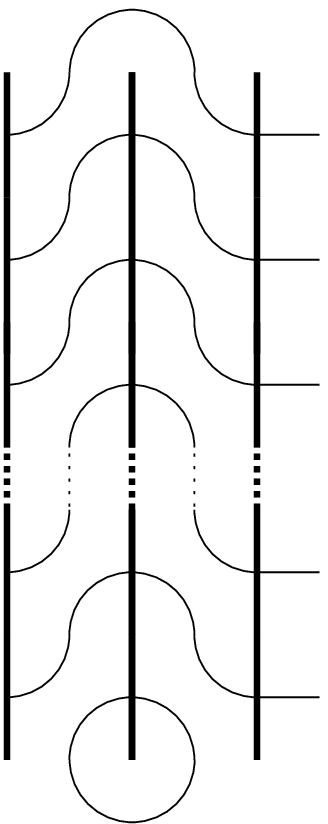}
\end{center}}
\ - q^{n-2p-1} \ 
\parbox{11mm}{\begin{center}
\includegraphics[height=26mm]{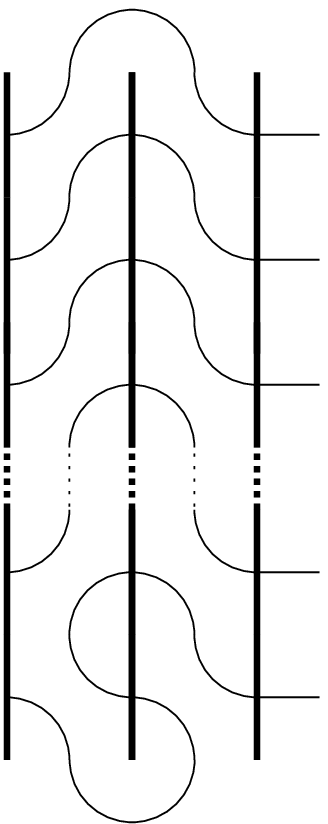}
\end{center}}
\ .
\end{equation}
The overall contribution from this analysis is therefore a factor of $q^{n-2p} \beta - q^{n-2p-1} = q^{n-2p+1}$.

The analysis of the second term on the right-hand side of \eqref{eqn:ContributionLinks2} is almost identical.  After analysing the crossings on the right side, the diagrams become horizontal reflections of those we have just analysed (this reflection corresponds to the automorphism $u_i \leftrightarrow u_{n-i}$) and $q$ is replaced by $q^{-1}$.  The resulting contribution is therefore $q^{-\brac{n-2p+1}}$.  Summing the contributions from both terms now gives the desired result.
\end{proof}

We conclude by indicating why we consider the central element $F_n$ superior to the more familiar $C_n$ for the purpose of analysing the indecomposable structure of $\tl{n}$-modules.  The algebra $\tl{n}$ is semisimple when $q$ is not a root of unity by \corref{cor:RootOfUnity}, so we may as well suppose that there exists a minimal positive integer $\ell$ such that $q^{2\ell} = 1$.  In other words, we have $q = e^{\ii \pi k / \ell}$ for some $k \in \ZZ$ coprime to $\ell$.  The eigenvalue of $F_n$ on the standard module $\NatMod_{n,p}$ is therefore
\begin{equation}
f_{n,p} = 2 \cos \frac{\pi k \brac{n-2p+1}}{\ell}.
\end{equation}
In \secref{sec:Gram}, we defined $\brac{n,p}$ to be critical when $\ell$ divides $n-2p+1$.  In that case, $f_{n,p}$ achieves it maximal or minimal value $\pm 2$.  More generally, we defined critical lines that partition the Bratteli diagram of the Temperley-Lieb algebras in \secref{sec:Explore}.  The crucial property of $F_n$ is that if $\brac{n,p}$ and $\brac{n,p'}$ are distinct pairs, both lying in the same strip bounded by two consecutive critical lines, then $f_{n,p} \neq f_{n,p'}$.  In other words, the eigenvalue of $F_n$ completely distinguishes standard modules within such strips.  In the language of Goodman and Wenzl \cite{GoodWenzl93}, this means that $F_n$ completely characterises the non-critical \emph{blocks} of the Temperley-Lieb algebra $\tl{n}$.  This property is not shared by $C_n$. A simple example is provided at $q=e^{\ii \pi / 6}$ for $\tl 4$. The three modules $\NatMod_{4,2}$, $\NatMod_{4,1}$ and $\NatMod_{4,0}$ lie to the left of the first critical line. The eigenvalues of $C_4$ on each are $1$, $e^{4\pi\ii/3}$ and $1$, and fail to distinguish $\NatMod_{4,2}$ and $\NatMod_{4,0}$. The eigenvalues of $F_4$ are $-\sqrt3$, $0$ and $\sqrt3$.

\section{Basic results}\label{app:Review}

This appendix reviews certain classical tools of representation theory that are not commonly seen in a first course on representation theory, say at the level of \cite{JamesLiebeck}.
These are Wedderburn's theorem and its generalisation to non-semisimple algebras, Frobenius reciprocity, the right-exactness of induction,  the Jordan-H\"older theorem and the basic properties of projective modules. We state them here without proof. The reader can find proofs of Wedderburn's theorem in \cite{Lam1,PieAss82}. Proofs of Frobenius reciprocity for algebras may be found in \cite{Martin, PieAss82}. The right-exactness of tensor products is proven in \cite{HGKAlgI} and projective modules are covered in \cite{CurRep62, PieAss82, Lam2}. Finally, we note that \cite[App.~A]{Mathas} provides a concise introduction to finite-dimensional algebras that gives proofs of Wedderburn's theorem and its non-semisimple generalisation, a proof of the Jordan-H\"older theorem, 
as well as describing what is needed on projective modules in order to understand \secref{sec:Proj}. All modules in this summary should be understood to be finite-dimensional left modules.

Let $\mathsf{A}$ be a finite-dimensional associative algebra over $\mathbb C$. Every element of $\mathsf{A}$ can be seen to act on $\mathsf{A}$ by left multiplication. This action makes $\mathsf{A}$ into a left $\mathsf{A}$-module called the {\em regular module}. The algebra $\mathsf{A}$ is said to be semisimple if the regular module is completely reducible, that is, if it can be written as a direct sum of irreducible modules. (In the case of a finite group $G$, its group algebra $\mathbb CG$ is always semisimple \cite{JamesLiebeck}.) A key consequence of semisimplicity is the following theorem.
\begin{theorem}[Wedderburn] \label{thm:W}
Let $\mathsf{A}$ be a complex, finite-dimensional, semisimple, associative algebra.  Then, the regular module decomposes as
\begin{equation} \label{eqn:WedderburnDecomp}
\mathsf{A} \cong \bigoplus_{i=1}^r \brac{\dim \IrrMod_i} \IrrMod_i,
\end{equation}
where the $\IrrMod_1, \ldots , \IrrMod_r$ form a complete set of non-isomorphic irreducible $\mathsf{A}$-modules.  That is, $\mathsf{A}$ decomposes as the direct sum of irreducibles with each irreducible appearing with multiplicity equal to its dimension.
\end{theorem}
\noindent A useful ``converse'' to this result which will be invoked in \secref{sec:Gram} is the following:
\begin{proposition} \label{prop:WeddConverse}
If the regular representation of a complex, finite-dimensional, associative algebra $\mathsf{A}$ decomposes as in \eqref{eqn:WedderburnDecomp}, where the $\IrrMod_1, \ldots , \IrrMod_r$ form a complete set of non-isomorphic irreducible $\mathsf{A}$-modules, then $\mathsf{A}$ is semisimple.
\end{proposition}
When the algebra $\mathsf{A}$ is not semisimple, the decomposition of its regular module will contain modules that are reducible, but not completely reducible. An {\em indecomposable} module $\mathcal M$ is one which cannot be written as the direct sum of two proper non-trivial submodules. Irreducible modules are examples. When the regular module of an algebra is written as a direct sum of indecomposable modules, those that appear in the decomposition are called the {\em principal indecomposable modules} of the algebra. Let $\{\mathcal P_i\}$ and $\{\mathcal L_i\}$ be complete sets of non-isomorphic principal indecomposable modules  of $\mathsf{A}$ and non-isomorphic irreducible modules, respectively. Wedderburn's theorem does not hold when semisimplicity is relaxed. It is replaced by the following generalisation:
\begin{theorem} \label{thm:GW}
Let $\mathsf{A}$ be a complex, finite-dimensional, associative algebra. The two sets $\{\mathcal P_i\}$ and $\{\mathcal L_i\}$ are put in one-to-one correspondence by associating a given principal indecomposable with its unique irreducible quotient. If $r$ is the common cardinality of these sets, then the regular representation decomposes as
\begin{equation}\label{eq:wedderNSS}
\mathsf{A} \cong \bigoplus_{i=1}^r \brac{\dim \IrrMod_i} \ProjMod_i.
\end{equation}
\end{theorem}

One way to characterise the structure of an $\mathsf{A}$-module $\mathcal M$ is through its composition series. A \emph{filtration} is a sequence of submodules of $\mathcal M$ such that
\begin{equation}
0=\mathcal M_0\subset \mathcal M_1\subset \mathcal M_2 \subset \dots \subset \mathcal M_k=\mathcal M.
\end{equation}
A filtration is said to be a \emph{composition series} if every quotient $\mathcal M_i/\mathcal M_{i-1}$, for $1\leqslant i\leqslant k$, is irreducible (and non-zero). These quotients are called the 
{\em composition factors} of $\mathcal M$. It is easily shown that every (finite-dimensional) module has a composition series. Indeed, a module may have several composition series.  However, the following theorem shows that the composition factors of $\mathcal M$ do not depend upon the choice of composition series.

\begin{theorem}[Jordan-H\"older]\label{thm:JH}Let 
\begin{equation}
0=\mathcal M_0\subset\mathcal  M_1\subset\mathcal  M_2 \subset \dots \subset\mathcal  M_k=\mathcal M 
\qquad \text{and}\qquad 
0=\mathcal N_0\subset\mathcal  N_1\subset\mathcal  N_2 \subset \dots \subset \mathcal N_{\ell}=\mathcal M
\end{equation}
be two composition series of the $\mathsf{A}$-module $\mathcal M$. Then, $k=\ell$ and, up to a permutation of their indices, the composition factors $D_i=\mathcal M_i/\mathcal M_{i-1}$ and $E_j=\mathcal N_j/\mathcal N_{j-1}$, for $1 \leqslant i,j \leqslant k$, coincide.
\end{theorem}

\noindent The composition factors of a submodule $\mathcal N\subset \mathcal M$ form a subset of those of $\mathcal M$. The same holds for quotients of $M$ and more general subquotients.

The theorem of reciprocity due to Frobenius is the 
next result covered in this appendix. It relates modules constructed from known ones through the classical constructions of \emph{restriction} and \emph{induction}. The first is easy to describe: When the algebra $\mathsf{A}$ contains a subalgebra $\mathsf{B}$, any module $\mathcal M$ over $\mathsf{A}$ is also a module over $\mathsf{B}$. (The action of an element of $\mathsf{B}$ on $\mathcal M$ is simply that of this element seen as an element of $\mathsf{A}$.) We will denote this $\mathsf{B}$-module by $\Res{\mathcal{M}}$. As vector spaces, $\mathcal M$ and $\Res{\mathcal{M}}$ are identical and have the same dimension. The second construction, induction, also uses a pair of algebras $\mathsf{B}\subset \mathsf{A}$ as before, but the starting module $\mathcal M$ is now over $\mathsf{B}$. The induced $\mathsf{A}$-module, which we denote by $\Ind{\mathcal{M}}$, is the tensor product $\mathsf{A}\otimes_{\mathsf{B}}\mathcal M$. (\secref{sec:Ind} provides a more detailed definition and several explicit examples.) Restriction and induction are ``dual'' operations in the following sense.
\begin{proposition}[Frobenius reciprocity] \label{prop:Frob}
Let $\mathsf{B} \subset \mathsf{A}$ be two finite-dimensional associative algebras over $\mathbb{C}$. Let $\mathcal{M}$ be a $\mathsf{B}$-module and $\mathcal{N}$ be an $\mathsf{A}$-module. Then, the following isomorphism between vector spaces of module homomorphisms holds:
\begin{equation}
\Hom_{\mathsf{A}}(\Ind{\mathcal{M}},\mathcal{N})\cong \Hom_{\mathsf{B}}(\mathcal{M}, \Res{\mathcal{N}}).
\end{equation}
\end{proposition}
\noindent The most familiar version of reciprocity corresponds to taking $\mathsf{B}$ and $\mathsf{A}$ to be the group algebras $\mathbb{C} H$ and $\mathbb{C} G$ of a pair of finite groups $H\subset G$. Finally, we shall need the behaviour of induction with respect to exact sequences.
\begin{proposition} \label{prop:IndRExact}
Suppose that
\begin{equation}
\dses{\mathcal M_1}{}{\mathcal M_2}{}{\mathcal M_3}
\end{equation}
is an exact short sequence of $\mathsf{B}$-modules. If $\mathsf{B}\subset \mathsf{A}$, then there is an exact sequence involving the induced $\mathsf A$-modules $\Ind{\mathcal M_i}=\mathsf{A}\otimes_{\mathsf{B}}\mathcal M_i$:
\begin{equation}
\rdses{\Ind{\mathcal M_1}}{}{\Ind{\mathcal M_2}}{}{\Ind{\mathcal M_3}}.
\end{equation}
\end{proposition}
\noindent Because induction preserves the exactness of short exact sequences except at the leftmost position, induction is said to be right-exact.

We turn now to projective modules which play a central role in \secref{sec:Proj}. An $\mathsf{A}$-module $\mathcal P$ is projective if, when there are two other $\mathsf{A}$-modules $\mathcal M$ and $\mathcal N$ and homomorphisms $\alpha:\mathcal P\rightarrow \mathcal N$ and $\gamma:\mathcal M\rightarrow\mathcal N$ with $\gamma$ surjective, then there exists a homomorphism $\delta: \mathcal P\rightarrow \mathcal M$ such that $\gamma\circ\delta=\alpha$. Equivalently, with the same input, there exists a homomorphism $\delta: \mathcal P\rightarrow \mathcal M$ such that the following diagram, with the bottom row exact, commutes:
\begin{center}
\begin{tikzpicture}[node distance=1.5cm, auto]
  \node (P) {$\mathcal P$};
  \node (N) [below of=P] {$\mathcal N$};
  \node (M) [left of=N] {$\mathcal M$};
  \node (O) [right of=N] {$0$.};
  \draw[->] (P) to node {$\alpha$} (N);
  \draw[->, dashed] (P) to node [swap] {$\delta$} (M);
  \draw[->] (M) to node [swap] {$\gamma$} (N);
  \draw[->] (N) to node [swap] {$\ $} (O);
\end{tikzpicture}
\end{center}
Here are some basic properties of projective modules.
\begin{proposition} \label{prop:ProjSum} \ 

\begin{enumerate}[label=\textup{(\roman*)}]
\item Any direct sum of projective modules is projective.
\item Every direct summand of a projective module is projective. \label{it:Proj2}
\end{enumerate}
\end{proposition}
\noindent Clearly \ref{it:Proj2} implies that every projective module can be written as a sum of indecomposable projective modules. There are many other ways to characterise projective modules.
\begin{proposition} \label{prop:ProjEquiv}
Let $\mathcal P$ be a module over an associative algebra. The three following conditions are equivalent:
\begin{enumerate}[label=\textup{(\roman*)}]
\item $\mathcal P$ is projective.
\item if $\mathcal M$ and $\mathcal N$ are any modules such that the sequence
$\ses{\mathcal M}{\mathcal N}{\mathcal P}$
is exact, then this sequence splits.
\item $\mathcal P$ is a direct sum of principal indecomposable modules.
\end{enumerate}
\end{proposition}
\noindent The last statement shows the particular role played by the summands of \eqref{eq:wedderNSS}: The principal indecomposables are precisely the indecomposable projective modules.

\raggedright

\end{document}